\documentclass[a4paper,USenglish,cleveref, pdfa]{lipics-v2021}

\pdfoutput=1 
\hideLIPIcs  

\title{Circuits and Backdoors: Five Shades of the SETH} 

\author{Michael Lampis}{Universit\'{e} Paris-Dauphine, PSL University, CNRS UMR7243, LAMSADE, Paris, France}{michail.lampis@dauphine.fr}{https://orcid.org/0000-0002-5791-0887}{}

\authorrunning{M. Lampis} 

\Copyright{Michael Lampis}

\ccsdesc[500]{Theory of computation~Parameterized complexity and exact algorithms} 

\keywords{SETH, SAT Backdoors, Parameterized Complexity, Fine-grained Complexity} 

\category{} 

\relatedversion{} 

\funding{This work was supported by ANR project ANR-21-CE48-0022 (S-EX-AP-PE-AL).}

\nolinenumbers 

\EventEditors{John Q. Open and Joan R. Access}
\EventNoEds{2}
\EventLongTitle{42nd Conference on Very Important Topics (CVIT 2016)}
\EventShortTitle{CVIT 2016}
\EventAcronym{CVIT}
\EventYear{2016}
\EventDate{December 24--27, 2016}
\EventLocation{Little Whinging, United Kingdom}
\EventLogo{}
\SeriesVolume{42}
\ArticleNo{23}

\usepackage{tikz}

\usepackage{enumitem}

\newcommand{\tsat}{3-\textsc{SAT}}

\newcommand{\seth}{\textsc{SETH}}

\newcommand{\ppseth}{$\pw$\textsc{-SETH}}
\newcommand{\pw}{\textrm{pw}}
\newcommand{\pwM}{\textrm{pwM}}
\newcommand{\td}{\textrm{td}}
\newcommand{\tdM}{\textrm{tdM}}
\newcommand{\tw}{\textrm{tw}}
\newcommand{\twM}{\textrm{twM}}

\newcommand{\eps}{\varepsilon}

\newcommand{\circseth}{\textsc{Circuit-SETH}}
\newcommand{\hornseth}{\textsc{HornB-SETH}}
\newcommand{\wpseth}{\textsc{W[P]-SETH}}

\newcommand{\aldseth}{\textsc{LD-C-SETH}}
\newcommand{\pwmseth}{$\pw$\textsc{M-SETH}}
\newcommand{\twmseth}{$\tw$\textsc{M-SETH}}
\newcommand{\tdmseth}{$\td$\textsc{M-SETH}}
\newcommand{\wsatseth}{\textsc{W[SAT]-SETH}}
\newcommand{\logpwmseth}{$\log\!\pw$\textsc{M-SETH}}
\newcommand{\logtdmseth}{$\log\!\td$\textsc{M-SETH}}
\newcommand{\twosatseth}{\textsc{2SatB-SETH}}
\newcommand{\maxsatseth}{\textsc{MaxSat-SETH}}
\newcommand{\hubmaxsatseth}{$(\sigma,\delta)H$\textsc{-MaxSat-SETH}}
\newcommand{\sdhseth}{$(\sigma,\delta)H$\textsc{-SETH}}

\newcommand{\ncseth}{\textsc{NC-SETH}}
\newcommand{\mtsh}{\textsc{Max-3-SAT-SETH}}

\newcommand{\reach}{\textsc{Ann-s-t-Reach}}
\newcommand{\nonreach}{\textsc{Ann-s-t-Non-Reach}}
\newcommand{\kncut}{$k$\textsc{-Neighborhood-Cut}}

\newcommand{\sdh}{$(\sigma,\delta)$-hub}

\begin{document}

\maketitle

\begin{abstract}

The Strong Exponential Time Hypothesis (\seth) is a standard assumption in
(fine-grained) parameterized complexity and many tight lower bounds are based
on it.  We consider a number of reasonable weakenings of the \seth, with
sources from (i) circuit complexity (ii) backdoors for SAT-solving (iii) graph
width parameters and (iv) weighted satisfiability problems. Our goal is to
arrive at formulations which are simultaneously more plausible as hypotheses,
but also capture interesting and robust notions of complexity.  Using several
tools from classical complexity theory we are able to consolidate these
numerous hypotheses into a hierarchy of five main equivalence classes of
increasing solidity. This framework serves as a step towards structurally
classifying a variety of \seth-based lower bounds into intermediate equivalence
classes.


To illustrate the applicability of our framework, for each of our classes we
give at least one (non-\textsc{SAT}) problem which is equivalent to the class
as a characteristic example application. As our main showcase, we consider a
natural parameterization of \textsc{Independent Set} by vertex deletion
distance from several standard graph classes. Our framework allows us to make
connections that would have previously been hard to see: for instance, we
consider the question of whether \textsc{(Weighted) Independent Set} can be
solved faster than $2^kn^{O(1)}$ on graphs which are $k$ vertices away from a
class $\mathcal{C}$ where the problem is tractable, such as interval graphs or
cographs. We provide precise characterizations of the difficulty of breaking
such bounds, in particular proving that obtaining $(2-\eps)^kn^{O(1)}$ time
algorithms for Cograph$+kv$ or Block$+kv$ graphs is \emph{equivalent} to
obtaining a fast satisfiability algorithm for circuits of depth $\eps n$; while
solving the weighted version for Interval$+kv$ graphs is \emph{equivalent} to
the (seemingly) harder problem of obtaining a fast satisfiability algorithm for
\textsc{SAT} parameterized by a \textsc{2-SAT} backdoor.

\end{abstract}

\newpage

\setcounter{tocdepth}{2}

\tableofcontents

\newpage

\section{Introduction} 

The objective of this paper is to study the structural relations between
fine-grained questions in parameterized complexity theory. In order to avoid
becoming too abstract and philosophical, let us begin with a concrete example
that illustrates what we mean by that.

Consider the \textsc{Independent Set} problem and fix a class of graphs
$\mathcal{C}$ on which \textsc{Independent Set} is in P and which satisfies the
very mild condition of being hereditary (closed under vertex deletion). How
hard is it to solve \textsc{Independent Set} on $\mathcal{C}+kv$, that is, the
class of graphs which are $k$ vertex deletions away from being in
$\mathcal{C}$? Assuming the $k$ vertices are marked in the input, the obvious
time upper bound is $2^kn^{O(1)}$: guess which of the $k$ vertices to take in
the solution, simplify the instance, and solve the rest in polynomial time.
Notice that this is a very standard ``distance to triviality''
parameterization, of which numerous are considered in the parameterized
complexity literature, and that our algorithmic argument is based on a trivial
branching procedure.

Given the above, a very natural question to ask is for which classes
$\mathcal{C}$ one could hope to do better than $2^kn^{O(1)}$. Let us
cherry-pick a few example classes which satisfy our criteria, such as interval
graphs, cluster graphs, cographs, and block graphs\footnote{The precise
definitions of these classes are recalled in \cref{sec:indset}, but are not
necessary at this point.}.  Since cluster graphs are a subset of the other
classes, one should clearly start there.  However, a recent result of Esmer,
Focke, Marx, and Rzazewski \cite{EsmerFMR24} shows that even for cluster
graphs, we cannot achieve $(2-\eps)^kn^{O(1)}$, unless the \seth\ is
false\footnote{As we explain later this is not explicitly stated in
\cite{EsmerFMR24} but follows from their construction. We give details in
\cref{sec:indset}.}! It therefore appears that despite the fact that, say,
interval graphs are a much wider class than cluster graphs, the complexity of
all these problems (and many other similar cases, say $\mathcal{C}$ being
chordal, claw-free, $P_5$-free, etc.) simply collapse to one identical case.
Should we view this as cause for celebration (because we have completely
resolved all these problems)? Or should we view this as cause for concern (are
we missing something)?

Before we answer this, let us take a step back into classical complexity
theory. Consider the following questions: we are given a CNF formula $\phi$ and
are asked (i) is $\phi$ satisfiable? (ii) does $\phi$ have a unique satisfying
assignment? (iii) is the majority of assignments satisfying for $\phi$? (iv)
can we satisfy $\phi$ even if we alternate picking values for variables with an
adversary who is trying to falsify it? A knowledgeable reader can easily
recognize that the complexity of these questions is strongly believed to be
very different, as they correspond to complete problems for the classes NP,
$\Sigma_2^p$, PP, and PSPACE respectively. Nevertheless, a reader less familiar
with complexity theory may see a similarity between this situation and the
questions of the previous paragraph. Indeed, all the questions we asked about
$\phi$ can be solved in exponential time (and polynomial space), while
exponential time is necessary even for the easiest one (under standard
hypotheses). We observe then that a large effort of classical complexity theory
has been devoted into understanding the relative structural complexities of
various problems, even (and perhaps especially) in cases where the problems are
believed to have the \emph{same time and space complexity}. This effort has
been quite impactful as the different complexity nuances captured by all the
intermediate classes have often turned out to be important in other aspects
(for instance, kernelization lower bounds typically rely on the assumption that
the polynomial hierarchy does not collapse, which would not make sense as a
statement if we only cared about deterministic time complexity).

Given this history, the thesis of this paper is that if, in the context of
fine-grained parameterized complexity, we settle just for \seth-based lower
bounds, then we are indeed missing something.  Going back to our example, we
would like to replicate this success of classical complexity theory in our
setting and obtain a more nuanced understanding of the complexity of our
questions.  In the same way that traditional complexity theory allows us to go
from saying ``this problem is likely intractable (not in P)'' to saying ``this
problem is exactly as hard as $\Sigma_2^p$ or PSPACE'', we would like to be
able to go from saying ``brute-force is inevitable for \textsc{Independent Set}
on the class $\mathcal{C}+kv$ (under \seth)'' to ``avoiding brute force for
\textsc{Independent Set} on $\mathcal{C}+kv$ is exactly as hard as this class
of problems''.  The obstacle is of course that, at the moment, fine-grained
parameterized theory is missing the framework and the vocabulary necessary to
do this. In other words, we do not yet have a classification effort which would
allow us to take questions which are all ``\seth-hard'' and order them into
equivalence classes, according to their relative difficulty.  

The goal of this paper is exactly to explore this new direction, start building
the basic blocks of such a framework, and show that this is a fruitful and
promising approach.  We will do this by identifying five natural equivalence
classes (\cref{fig:results}), one of which contains the \seth, and giving
natural complete problems for each as well as an informal characterization
connecting them to classical complexity classes.  To showcase the applicability
of this framework we will also apply it to the \textsc{Independent Set} problem
discussed above for the four classes we cherry-picked (\cref{fig:indset}).  The
framework will allow us to compare the relative complexities of cases which do
not have an obvious graph-theoretic relation (for example cographs and interval
graphs are incomparable classes), but also relate them with questions that
initially seem completely unrelated, stemming from circuit complexity or
backdoors for \textsc{SAT}-solving. We see this investigation as a first step
that will allow us to more carefully classify known tight lower bounds from the
literature into equivalence classes, with the added benefit of identifying the
lower bounds which should be considered most plausible (because it must be
taken into account that the \seth\ is not a universally believed assumption).

\subsection{Background and Motivation}

Having summarized the general direction of this paper, let us give a more
detailed review of the relevant background. Recall that the Strong Exponential
Time Hypothesis (\seth\ \cite{ImpagliazzoP01}) states that the satisfiability
of CNF formulas cannot be decided significantly faster than by brute force,
that is, in time $(2-\eps)^n$, assuming the formulas have large enough
arity\footnote{To ease presentation, in this section we use informal statements
such as this one to describe all complexity hypotheses. In case of doubt, we
refer the reader to \cref{sec:hypo}, where precise definitions are given.}.
Since its introduction, the \seth\ has been an indispensable tool in the
development of fine-grained complexity theory. In this paper we focus on its
influence on fine-grained \emph{parameterized} complexity.

One of the most important techniques in the parameterized algorithms toolbox is
the use of dynamic programming algorithms over tree decompositions, and the
question of whether such algorithms can be improved is of central importance.
It was therefore a very exciting development for the field when Lokshtanov,
Marx, and Saurabh \cite{LokshtanovMS18} pioneered a line of work which allowed
to provide tight lower bounds for a large class of algorithms of this type. For
instance, one of their results was that the standard algorithm for
\textsc{3-Coloring}, which has complexity $3^{\tw}n^{O(1)}$, cannot be improved
to $(3-\eps)^{\tw}n^{O(1)}$, where $\tw$ denotes the treewidth. Similar
\emph{tight} results were given in \cite{LokshtanovMS18} and many subsequent
works for many other problems, and also for other parameters (for example
clique-width, cutwidth)
\cite{BojikianCHK23,BorradaileL16,CurticapeanM16,DubloisLP21,DubloisLP22,FockeMINSSW23,FockeMR22,GanianHKOS22,HanakaKLOS20,HegerfeldK23,HegerfeldK23b,JaffkeJ23,KatsikarelisLP19,KatsikarelisLP22,LampisV24,MarxSS21,OkrasaR21}.
An important common point of all these lower bounds is that they are generally
shown to hold \emph{assuming the} \seth. 

\subparagraph*{The \seth\ is too strong (i).} One of the motivations of this
paper is to question the wisdom of basing all these numerous lower bound
results on one-way implications of the \seth. Even though the \seth\ is mostly
considered a reasonable hypothesis, it is important to keep in mind that it is
not universally believed. Indeed, there seem to be serious reasons to doubt its
validity, much more than that of more standard assumptions (we refer the reader
to the argumentation given by Williams \cite{Williams19} who, somewhat
provocatively, assigns the \seth\ only a one-in-four probability of being true;
see also \cite{Williams16}). In other words, even if one believes that the
\seth\ is \emph{probably} true, because its status is sufficiently open one
should not necessarily consider a \seth-based lower bound as the final word on
a problem's status.

This line of reasoning has already been well-developed in the
(non-parameterized) fine-grained complexity community. In particular, there are
several works which reprove (and improve) lower bounds based on the \seth\
using a more plausible assumption, such as the \ncseth, which states the same
lower bound but for circuits of poly-logarithmic depth, rather than CNF
formulas. We refer the reader to the work of Abboud, Hansen, Vassilevska
Williams, and Williams \cite{AbboudHWW16} for a more extensive discussion
of why the \ncseth\ is a preferable hypothesis, or equivalently why
falsifying the \ncseth\ would have far more surprising consequences than
disproving the \seth; see also \cite{AbboudB18,ChenGLRR19,Polak18}. In the
parameterized complexity setting, similar work was only recently initiated
in \cite{Lampis25} (more on this below).

\subparagraph*{The \seth\ is too strong (ii).} The discussion above explains
that the \seth\ is not an ideal starting point for proving lower bounds, as its
validity is still subject to some doubt. There is, however, another, more
counter-intuitive reason to argue that the \seth\ is not ideal, and that is
that the \seth\ is ``too good'' in comparison with the lower bounds we are
interested in. This is a line of reasoning taken by several previous works
\cite{HegerfeldK22,JaffkeJ23}. In particular, Jaffke and Jansen
\cite{JaffkeJ23} showed that the correct base for $q$-\textsc{Coloring} is
exactly $q$ not only when we parameterize by pathwidth or treewidth, but
even when we parameterize by vertex deletion distance to path-forest (a
significantly more restrictive parameter).  More recently, this approach
was further explored by Esmer, Focke, Marx, and Rzazewski
\cite{EsmerFMR24}, who observed that for many of the problems considered in
\cite{LokshtanovMS18} (including $q$-\textsc{Coloring}, \textsc{Independent
Set}, and \textsc{Max-Cut}), a more refined reduction shows that the \seth\
implies a much stronger lower bound: Not only are current algorithms
optimal for treewidth, but they are still optimal for a novel, extremely
restricted parameter, which they call \sdh. A \sdh, (for $\sigma,\delta$ to
be considered absolute constants) is a set of vertices (a modulator) whose
removal leaves a collection of bounded-size (size $\sigma$) components with
bounded degree (degree $\delta$) in the modulator.  The class of graphs
with a \sdh\ of size $k$ is a tiny fraction of the class of graphs of
treewidth $k$, so this is a strong indication that the lower bounds for the
more interesting parameters should be achievable by a much weaker assumption.

Another way to see the above is that the \seth\ is talking about the hardness
of a \emph{different class} of problems than the width-based lower bounds that
it implies. Indeed, observe that we have crucially moved from problems where
the ``obvious'' algorithm is dynamic programming (parameterized by treewidth or
other widths) to problems where the ``obvious'' algorithm is branching (for
example, decide the color of the vertices of the modulator). This leads us to
believe that there are several ``complexity classes'' lying between the \seth\
and the treewidth-based lower bounds discussed above and waiting to be
discovered.

\subparagraph*{Equivalences are better than implications.} Another reason to be
unsatisfied with the state of the art is that one-way reductions from
\textsc{SAT} obscure the relative difficulties of very different problems. To
see what we mean by this, let us go back to the classical complexity setting
and try to compare the difficulty of four problems which are, say, respectively
NP-complete (\textsc{SAT}), $\Sigma_2^p$-complete
($\exists\forall$-\textsc{SAT}), $\#$P-complete ($\#$\textsc{SAT}), or
PSPACE-complete (\textsc{QBF}). A superficial understanding of complexity theory
would claim that these are all ``equally hard'': if one believes that
P$\neq$NP, all these problems need super-polynomial time and polynomial space.
Arguably, one of the main contributions of classical complexity theory is to
offer a better answer and a more nuanced understanding of the relative
complexities of these problems, by classifying them into equivalence classes.
Crucially, this work often ends up agreeing with the practical difficulty of
solving such problems (see for example \cite{MarinNPTG16} for why
\textsc{Quantified SAT} is considered harder than \textsc{SAT} in practice). 

Much of the state of the art in fine-grained parameterized complexity can then
be thought of as being close to the superficial understanding
we sketched above: we show that our problems are harder than the \seth, and
consider the case closed. The thesis of this paper is that there is something
to be gained by more carefully comparing the relative complexities of our
problems, and for this we need to organize them into equivalence classes.
Interestingly, we will show that the equivalence classes we end up with can
sometimes be orthogonal to traditional complexity classes.

\begin{figure}[t]

\input{figure.tex}

\caption{Summary of results. We identify five main equivalence classes, with
arrows indicating implication, that is, pointing to weaker hypotheses.  From
bottom to top the five classes are problems equivalent to the \seth; the \seth\
for \textsc{Max-SAT}; the \seth\ for circuits of linear depth; the \seth\ for
\textsc{2-SAT} backdoors; and the \seth\ for arbitrary circuits.  The figure
also shows some connections with related hypotheses. Full definitions are given
in \cref{sec:hypo}.}\label{fig:results}

\end{figure}

\subsection{Summary of Results}

One way to explain the main direction of our results is to come back to the
recent work of Esmer, Focke, Marx, and Rzazewski \cite{EsmerFMR24} we
mentioned. We share their initial motivation, which is the observation that the
\seth\ does not just imply tight treewidth-based lower bounds, such as those of
\cite{LokshtanovMS18}, but in fact much more. However, where their next step is
to say ``Let us then try to extract the strongest implication of the \seth'',
we say ``Let us then explore how we can replace the \seth\ by a different (more
solid) hypothesis, and try to understand what type of question we capture''. We
will therefore study numerous variants of the \seth, coming from several
different areas, including: (i) circuit versions, similar to the \ncseth\ (ii)
backdoor versions, such as ``If we have a \textsc{SAT} formula with a
\textsc{2-SAT} backdoor, can we avoid brute-forcing the backdoor''? (a backdoor
is a set of variables such that setting them simplifies the formula) (iii)
modulator to graph width versions, such as ``If we have a \textsc{SAT} formula
with a modulator to a formula of constant treewidth, can we avoid brute-forcing
the modulator?'' (iv) weighted versions where the question is whether a
satisfying assignment of weight $k$ exists (v) maximization versions where the
objective is to maximize the number of satisfied clauses.

At this point it is worth pausing and trying to reassure a worried reader. One
of the weak points of fine-grained complexity theory is the proliferation of
reasonable-sounding complexity hypotheses, with no obvious connection to each
other. It is natural to think that our decision to study a relatively large and
diverse collection of variants of the \seth\ can only make things worse.
Thankfully, the main take-home message of our results is that, even though we
consider quite a few variations, we are able to identify only five equivalence
classes (one of which contains the \seth\ itself), which align neatly with
classes from classical complexity theory. In a sense, the conclusion we draw is
optimistic as our results indicate that there is a great deal of
complexity-theoretic structure waiting to be discovered here.

Let us then summarize our results, referring also to \cref{fig:results}. We
identify the following equivalence classes of complexity hypotheses:

\begin{enumerate}

\item Possibly the most interesting of the classes we consider is the middle
one of \cref{fig:results}, which is motivated by the \ncseth. Even though the
\ncseth\ is stated for circuits of poly-logarithmic depth in
\cite{AbboudHWW16}, in practice it is often understood to apply to circuits of
depth $n^{o(1)}$ (e.g.  \cite{ChenGLRR19}).  We take this reasoning to its
logical conclusion and formulate this hypothesis for circuits of depth
$\Omega(n)$, or more precisely for depth at most $\eps n$ (we call this
\aldseth, for Linear-Depth Circuits).  We discover that this hypothesis is
equivalent to stating that brute force is unavoidable for CNF formulas supplied
with a modulator to one of three standard widths: (i) constant pathwidth (ii)
constant treewidth or (iii) logarithmic tree-depth. We find it somewhat
surprising that these parameters, which are know to strictly generalize each
other in the graph-theoretic sense, collapse to the same hypothesis.  We also
show that disproving the \aldseth\ is equivalent to avoiding brute force for
the W[SAT]-complete problem of finding a weight-$k$ satisfying assignment to a
given (general, non-CNF) formula. This class corresponds to the middle box of
\cref{fig:results}. An important remark here is that we are able to extend
these results to \textsc{MaxSAT} (see \cref{sec:maxsattd}), as this requires a
foray into arithmetic circuits and is crucial later on to obtain the precise
characterization of the complexity of \textsc{Independent Set} on graphs which
are close to being cographs or block graphs. This class therefore seems to be
quite robust and contains some very natural equivalent problems unrelated to
\textsc{SAT}.

\item Going further, we consider the \seth\ for circuits without any
restriction (\circseth), an assumption which seems much more plausible than the
\seth. We show that avoiding brute force in this case can be done if and only
if we can avoid brute force in one of two problems (i) deciding satisfiability
for a CNF supplied with a Horn backdoor (ii) deciding weight-$k$ satisfiability
for an arbitrary circuit, a problem complete for W[P]. This is the top-most
class of \cref{fig:results}.

\item Coming back to the \aldseth, we consider the graph width parameters which
are closest to the ones already considered.  For modulators to logarithmic
pathwidth (which would be the logical next step after logarithmic tree-depth)
we show that avoiding brute-force enumeration is possible if and only if we can
avoid it for \textsc{2-SAT} backdoors. As we explain below, one of our main
showcase results is to show that an equivalent characterization of this class
can be given by \textsc{Weighted Independent Set} on $\textrm{Interval}+kv$
graphs, so this class also seems quite natural.

\item In the other direction, the logical next step below modulators to
constant pathwidth is modulators to constant tree-depth. Here we show that the
hypothesis that brute-force algorithms are optimal is in fact equivalent to the
standard \seth. The parameter \sdh\ introduced by \cite{EsmerFMR24} is
therefore included here, as it is a restriction of bounded tree-depth
modulators.  In other words, our main result here is that hypotheses based on
any parameter more restricted than ``modulator to bounded tree-depth'' all
collapse to the normal \seth. This is the bottom class of \cref{fig:results}.
In a similar spirit we show that for \textsc{Max-SAT}, the standard version of
the hypothesis (for the number of variables) and the version for the size of a
\sdh\ are equivalent (this is the second class from the bottom in
\cref{fig:results}).

\end{enumerate}

\subparagraph*{Take-home message} We have identified five natural equivalence
classes of hypotheses.  In order to make sense of these results, let us attempt
to draw an informal analogy with classical complexity theory.  All the
hypotheses we consider concern problems where the obvious algorithm is a
brute-force ``guess-and-check''; what differs is the difficulty of the
verification procedure. If the verification procedure is P-complete (for
example \textsc{Horn-SAT} or evaluating an arbitrary circuit) we end up in the
top-most class of \cref{fig:results}. Similarly, if it is NL-complete
(equivalent to a reachability problem) we end up in the class right below,
while the class of the \aldseth\ seems to correspond to NC$^1$. Finally, the
class immediately above the \seth\ seems somewhat similar to (a restricted
version of) TC$^0$, since in \textsc{Max-SAT} what we need to do if given an
assignment is count the number of satisfied clauses and compare this number
with a target -- this corresponds to the functioning of a threshold gate. What
we are then essentially doing is reconstructing a parallel of a part of
traditional complexity theory around fine-grained variations of the \seth\ and
what remains is to populate these classes with further problems.  Armed with
this intuition, we can identify some further complexity questions as
characteristic examples of each class.

\subparagraph*{Applications and Implications} One motivation we gave for this
investigation was that we would like to increase our confidence in \seth-based
lower bounds for problems parameterized by treewidth. One implication of our
results is that we have managed to do so for many problems.  Indeed, we show
that improving upon the trivial algorithm for $q$-\textsc{Coloring} for
parameters modulator size to constant pathwidth, constant treewidth, or
logarithmic tree-depth is equivalent to the \aldseth.  This implies that
improving upon the standard algorithm for pathwidth or treewidth is at least as
hard as falsifying the \aldseth\ (and hence the \ncseth).  As explained, this
seems significantly more solid than a result based on the \seth. Furthermore,
recent work by Lampis \cite{Lampis25} showed that a number of fine-grained
lower bounds for problems parameterized by pathwidth (including
\textsc{Independent Set}, \textsc{Dominating Set}, \textsc{Set Cover} and
others) are equivalent to the \ppseth, the hypothesis that \textsc{SAT} cannot
be solved in time $(2-\eps)^{\pw}|\phi|^{O(1)}$, where $\pw$ is the input
formula's primal pathwidth. Since the \ppseth\ is implied by the same
hypothesis for parameter modulator to logarithmic pathwidth, all the hypotheses
we consider here except those of the top-most class of \cref{fig:results} also
imply the lower bounds of \cite{Lampis25}. We have thus succeeded in giving
further evidence that known lower bounds for pathwidth and treewidth hold. For
instance, even though it was previously known that no $(3-\eps)^\pw n^{O(1)}$
algorithm is possible for \textsc{Dominating Set} under the \seth, our results
(together with those of \cite{Lampis25}) establish this under the \ncseth.  

Beyond this, we give a characteristic example problem for each class, based on
the intuition linking the classes of \cref{fig:results} with traditional
complexity classes. In particular, we show that the \circseth\ is equivalent to
beating brute-force for the problem of deleting $k$ vertices from a graph, so
that it becomes $r$-degenerate, for any fixed $r\ge 2$; this problem is known
to be P-complete. Similarly, we give a reachability-type problem (\kncut) that
is equivalent to the \twosatseth, which is natural, since reachability-type
problems are usually NL-complete.

These applications are of particular interest for the two bottom classes in
light of the results of \cite{EsmerFMR24}. In this work it was shown that if
the \seth\ is true, it is impossible to avoid brute force for (among others)
$q$-\textsc{Coloring} and \textsc{Max-Cut}, parameterized by \sdh\ size. This
perhaps creates the impression that we should have the same amount of
confidence in the two lower bounds. Contrary to this, we show that their lower
bound for $q$-\textsc{Coloring} is in fact equivalent to the \seth, while their
lower bound for \textsc{Max-Cut} is equivalent to the \maxsatseth. If it were
possible to show that these lower bounds are equivalent to each other, this
would also establish that the \mtsh\ (the restriction of \maxsatseth\ to CNFs
of arity $3$) implies the \seth. As pointed out in \cite{EsmerFMR24}, this is a
major open problem. Our results therefore uncover some (perhaps unforeseen)
nuance in the lower bounds of \cite{EsmerFMR24}.

\subparagraph*{Relation with Time Complexity Classes.} It is worth stressing
that, even though we present \cref{fig:results} as a collection of ``complexity
classes'', and indeed we draw a parallel with classical such classes, the
results we present are in a sense orthogonal to normal notions of time
complexity. What we mean by this is that, as can be seen in \cref{fig:results},
we often show equivalence for problems which (are widely believed to) belong in
distinct classes, namely, problems which are FPT on the one hand and W[SAT]- or
W[P]-complete on the other. In other words, the equivalence classes of
\cref{fig:results} do not indicate that all the problems in each class have the
same time complexity; but rather that if we could improve the running time of
the best algorithm for one, we could improve it for the others, albeit
obtaining different types of running times for each problem. We remark that a
similar situation holds for the \ppseth\ \cite{Lampis25}.

\begin{figure}[h] \centering\input{indsetfigure.tex} \caption{Summary of our
results for \textsc{(Weighted) Independent Set} parameterized by vertex
deletion distance to four graph classes: block, cluster, interval graphs, and
cographs.  Arrows indicate implication and a problem name indicates the
hypothesis that brute-force is inevitable for this case. We consider unweighted
versions (IS), unary weighted (WIS) and binary weighted
(BinWIS).}\label{fig:indset} \end{figure}

\subparagraph*{A showcase application.} At this point we have established a
framework based on five version of the \seth\ which are, we believe, quite
natural and roughly corresponds to five complexity classes, endowed with nice
complete problems. Nevertheless, it needs to be said that complexity classes
can only be considered as interesting as the problems they contain or capture,
or as Papadimitriou and Yannakakis put it, classes that contain interesting
problems are \emph{discovered} not \emph{invented} \cite{PapadimitriouY96}.
Trying to verify that this is the case for our classes, we return to the
motivating example this paper began with: \textsc{(Weighted) Independent Set}
on $\mathcal{C}+kv$ graphs, where $\mathcal{C}$ are the graph classes we
cherry-picked: cluster graphs, cographs, block graphs, and interval graphs. Our
results are summarized in \cref{fig:indset}. We identify the \aldseth\ as a
dividing line for the complexity of this problem: for $\textrm{Cluster}+kv$
graphs the problem is at most as hard as the \aldseth, even if binary weights
are allowed, while for $\textrm{Cograph}+kv$ and $\textrm{Block}+kv$ the
problem is \emph{exactly as hard} as the \aldseth, in both the unweighted case
and the weighted case (with unary weights). However, \textsc{Weighted
Independent Set} (again with unary weights) appears to be strictly harder for
$\textrm{Interval}+kv$ graphs, as we show that this version is
\emph{equivalent} to the \logpwmseth.  Observe that in this way the framework
allows us to make comparisons which are not obvious graph-theoretically:
cographs, block graphs, and interval graphs are incomparable classes, so a
priori one would not have necessarily expected cographs and block graphs to be
equivalent but interval graphs to be harder.  Of course, a broader advantage is
that integrating these questions in our framework allows us to compare these
questions with many other, seemingly unrelated questions (such as whether
$q$-\textsc{Coloring} can be solved faster than brute-force on graphs with a
small modulator to constant pathwidth, which we discussed previously).  We
believe these results show that this is a direction that merits further
investigation.

\subsection{Overview of Techniques}

Let us now sketch the main ideas used to obtain the results mentioned above. In
all cases, a large part of the intuition needed is obtained by the parallels we
drew between these classes and traditional complexity classes. More precisely:

\begin{enumerate}

\item For the hypotheses which turn out to be equivalent to the \aldseth, one
of the main technical tools we rely on is Barrington's theorem
\cite{Barrington89}, which allows us to transform a circuit on $n$ inputs and
depth $\eps n$ into a branching program of constant width and length
exponential in $\eps n$. This program can then be simulated by a formula with
constant pathwidth and a modulator representing the inputs. For the converse
direction we need to show that an algorithm disproving the \aldseth\ can be
used to decide the satisfiability of a CNF $\phi$ with a modulator $M$ of size
$m$, such that deleting $M$ produces a formula of tree-depth at most $O(\log
|\phi|)$. For this, we need to eliminate all the variables of the formula
except those of $M$. We show how to do this in a way that is exponential in the
tree-depth (and hence polynomial in the size of the formula). A particular
technical challenge here is that, to generalize this result to \textsc{MaxSAT}
(which will be crucial for later applications) we need to extend these
arguments to arithmetic circuits and then move back to boolean circuits to
obtain equivalence with the \aldseth.

\item For the equivalence between \twosatseth\ and \logpwmseth\ we rely on the
intuition that both \textsc{2-SAT} and satisfiability for \textsc{SAT} formulas
of pathwidth $O(\log |\phi|)$ are essentially reachability problems. However, a
closer look reveals that a direct reduction is non-obvious, as deciding
\textsc{2-SAT} is in fact more readily seen to be equivalent to
\emph{non-}reachability. We therefore need to transfer the classical
Immerman–Szelepcsényi theorem, which established that NL=coNL, to our setting
\cite{Immerman88}.

\item For the equivalences related to \circseth\ we rely on the intuition that
computing the value of a circuit and deciding the satisfiability of a Horn
formula are both P-complete problems.

\item For the equivalences given in the two bottom classes, the hard part is
showing that (\textsc{Max-})\textsc{SAT} parameterized by a modulator (to
constant tree-depth, or a \sdh) can be reduced to the same problem
parameterized by the number of variables, that is, that we can eliminate the
extra variables. For \textsc{SAT} this is done by an exhaustive application of
the standard resolution rule, starting from the leaves of the rooted trees,
while for \textsc{Max-SAT} parameterized by \sdh, this is done by enumerating
all assignments to each (bounded-size) component and replacing it with new
constraints to its (bounded-size) neighborhood.

\item Finally, for our applications to \textsc{Independent Set} on graph
classes, we need to rely on several unrelated ideas. In particular, we use
standard but non-obvious facts from circuit complexity, such as the fact that
$n$ $n$-bit integers can be summed by a bounded fan-in circuit of depth $O(\log
n)$; we use the fact that the circuits of the \aldseth\ can be considered
acyclic, which in turn allows us to encode them using the tree structure hidden
in cographs and block graphs; we use the fact that the natural DP algorithm for
\textsc{Weighted Independent Set} on interval graphs can be translated to a
reachability question to obtain equivalence with the \logpwmseth\ in this case.

\end{enumerate}

\subsection{Previous Work}

As mentioned, \seth-based lower bounds are an area of very active research in
parameterized complexity. Beyond the results following \cite{LokshtanovMS18},
which establish running time lower bounds of the form $(c-\eps)^kn^{O(1)}$, for
some parameter $k$ and fixed constant $c$, the \seth\ has also been used to
obtain lower bounds for non-FPT problems. The most famous example is
perhaps the result of Patrascu and Williams that $k$-\textsc{Dominating
Set} cannot be solved in time $n^{k-\eps}$, under the \seth\
\cite{PatrascuW10}. To some extent the questions we ask here were already
anticipated in \cite{PatrascuW10}, who note that obtaining an $n^{k-\eps}$
algorithm for $k$-\textsc{Dominating Set} is at least as hard as disproving
the \seth, but at most as hard as disproving the \hornseth. In other words,
they place $k$-\textsc{Dominating Set} somewhere between the top and bottom
boxes of \cref{fig:results}. Although in this paper we have not been able
to identify the precise equivalence class of this problem, our results
improve the upper bound of Patrascu and Williams, by showing that
$k$-\textsc{Dominating Set} is at most in the third equivalence class (as
W$[2]\subseteq$W$[$SAT$]$).


The question of whether the \seth\ is the ``right'' assumption has appeared
before in the literature. In general fine-grained complexity this manifests
itself via works which attempt to replace the \seth\ with more plausible
assumptions, such as the ones mentioned above, or which attempt to prove the
same lower bound based on several hypotheses which have no known relation
\cite{AbboudBHS22,AbboudWY18}. Furthermore, the question of which problems are
exactly equivalent to the \seth\ has already been posed, though few such
problems are known \cite{CyganDLMNOPSW16}. We also note that the connection
between the \ncseth\ and Barrington's theorem was already used in
\cite{ChenGLRR19}.

In the area of parameterized complexity, investigations on whether the \seth\
is the right complexity assumption have mostly focused on whether, assuming the
\seth\ we can obtain the same lower bounds but for more restricted cases than
for treewidth (\cite{EsmerFMR24,JaffkeJ23}). However, a strongly related recent
work \cite{Lampis25} investigated the so-called \ppseth, that is, the
hypothesis that \textsc{SAT} cannot be solved in time
$(2-\eps)^{\pw}|\phi|^{O(1)}$. Even though this is closely related to our
motivation, the \ppseth\ turns out to cover a completely different algorithmic
paradigm, as it essentially speaks about the optimality of dynamic programming
algorithms, while the hypotheses we consider here refer to the optimality of
branching algorithms under different regimes. As a result, most of the
hypotheses we consider here imply the \ppseth, with the exception of \circseth,
for which this remains an interesting open problem.

\subsection{Discussion and Directions for Further Work}

We believe that the results of this paper, together with recent related works
(\cite{EsmerFMR24,Lampis25}) point the way to a new and interesting
research direction investigating the interplay between fine-grained complexity
theory and parameterized input structure. The envisioned goal here is a theory
where lower bound statements form a comprehensible hierarchy where we can
eventually pinpoint the reasons that a lower bound is believed to hold; for
example we have seen that the \aldseth\ is essentially equivalent to improving
upon a generic ``guess-and-check'' algorithm with an NC$^1$ verifier.  For
which other algorithmic paradigms can we identify such natural complete
problems?  Equally importantly, how do different paradigms relate to each
other? In particular, the \ppseth\ informally captures the difficulty of
improving upon a generic DP algorithm over a linear structure, and the
\circseth\ the difficulty of improving upon a generic ``guess-and-check''
algorithm with a polynomial-time verifier. Discovering a connection between
these two hypothesis would allow us to better understand which paradigm is more
powerful; conversely the two may be incomparable, and it would be interesting
to see if this can be expressed via the corresponding hypotheses. 

More broadly, an interesting research direction would of course be to further
populate the classes introduced in this paper. Specific questions include:

\begin{enumerate}

\item Characterize the complexity of \textsc{Independent Set} parameterized by
a modulator to an unweighted interval graph. In the weighted case, the problem
is equivalent to the \logpwmseth; could it be easier (for instance, equivalent
to the \aldseth) for unweighted graphs?

\item As mentioned, Jaffke and Jansen \cite{JaffkeJ23} showed that brute force
is unavoidable for $q$-\textsc{Coloring} parameterized by a modulator to a
path-forest. Our result improves upon this in one sense (because we show
equivalence to the \aldseth, rather than \seth-hardness), but is weaker in
another (because we parameterize by a modulator to a graph of constant
pathwidth, rather than a path). Can this gap be closed? We conjecture that
$q$-\textsc{Coloring} parameterized by feedback vertex set or modulator size to
a path-forest is in fact easier than the \aldseth. Evidence for such a claim
could be established, for instance, by coming up with a reduction from this
problem to \textsc{CNF-SAT}.

\item In a similar vein, a natural question to ask is the precise complexity of
\textsc{Independent Set} parameterized by feedback vertex set. Again, our
results can probably be used to show hardness (indeed \aldseth-equivalence)
when we have a modulator to a graph of constant treewidth, but not for a graph
of treewidth $1$. What is the smallest treewidth $c$ for which parameterizing
\textsc{Independent Set} by vertex-deletion distance to treewidth $c$ is
\aldseth-equivalent?

\end{enumerate}


\subsection{Paper Organization} 

We first remark that we have tried to use a self-explanatory naming scheme,
where \textsc{X}-\seth\ means the hypothesis that it is impossible to improve
upon brute-force search for \textsc{SAT} parameterized by a parameter
\textsc{X}, or when \textsc{X} is a complexity class, for the complete problem
of \textsc{X}.  Nevertheless, the definitions of all hypotheses used in this
paper (and shown in \cref{fig:results}) and the most basic relations between
them are given in \cref{sec:hypo}. Even though \cref{sec:hypo} serves as a
useful reference, we have made an effort to give precise statements in all
theorems, making sure that potentially non-trivial details in the formulations
of various hypotheses, such as the order of quantifications of some constants,
are clearly spelled out.  For the rest of the paper, after some preliminaries,
we present the basic connections for the two bottom classes of
\cref{fig:results} in \cref{sec:maxsat}, and then present one class in each of
\cref{sec:ald}, \cref{sec:logpw}, and \cref{sec:horn}. In
\cref{sec:applications} we present one further characteristic example for each
class. Finally, in \cref{sec:indset} we present our results on
\textsc{Independent Set} summarized in \cref{fig:indset}.

\section{Preliminaries}

We assume the reader is familiar with the basics of parameterized complexity as
given for example in \cite{CyganFKLMPPS15}. We refer the reader there for the
definitions of treewidth and pathwidth. A graph $G$ has tree-depth $k$ if there
is a bijective mapping of the vertices of $G$ to the vertices of a rooted tree
of depth $k$ so that for every edge $uv$, the images of $u,v$ have an
ancestor-descendant relation. We sometimes call such a mapping a tree-depth
decomposition. Following \cite{EsmerFMR24}, for integers $\sigma,\delta$, a
\sdh\ of a graph $G$ is a set of vertices $M$ such that each component of $G-M$
has size at most $\sigma$ and at most $\delta$ neighbors in $M$.

A modulator of a graph is a set of vertices such that removing it leaves a
graph with a desired property. When $G$ is a graph and $M$ a modulator we will
often write $G-M$ to denote the graph obtained after removing the modulator.

We deal with several variations of \textsc{SAT}, the problem of deciding if a
given formula has a satisfying assignment. We write $k$\textsc{-SAT} for the
restriction to clauses of arity at most $k$ and \textsc{Max-SAT} (or
\textsc{Max-}$k$\textsc{-SAT}) for the more general problem of determining the
assignment that satisfies the maximum number of clauses. For these problems we
allow repeated clauses in the input, or equivalently we allow the input to
assign (unary) weights to the clauses. An assignment of weight $k$ is an
assignment that sets exactly $k$ variables to True.

For a CNF formula we define its primal graph as the graph that has a vertex for
each variable and edges between variables that appear in the same clause. When
$\phi$ is a CNF formula and $M$ is a set of variables, we will write $\phi-M$
to denote the formula obtained by deleting every clause that contains a
variable of $M$ and, slightly abusing notation, we will also use $\phi-M$ to
refer to the resulting primal graph.

We will also consider satisfiability problems for general formulas and
circuits. A \emph{boolean} circuit is a DAG where vertices of positive
in-degree are labeled as disjunctions ($\lor$), conjunctions ($\land$), or
negations ($\neg$) and there is a unique sink, called the output gate. Vertices
of zero in-degree are called input gates and the satisfiability problem asks if
there is a truth assignment to the inputs that makes the output gate evaluate
to True. A formula is a circuit where the maximum out-degree of any non-input
gate is $1$. Unless otherwise stated we will assume that the fan-in (in-degree)
of all circuits is bounded, but in most cases unbounded fan-in gates can be
replaced by trees of bounded fan-in gates. The depth of a circuit is the
largest distance from an input gate to the output. The size of a circuit is the
number of arcs of the corresponding DAG. We will also consider
\emph{arithmetic} circuits, where the value of each gate is an integer rather
than a boolean and will mainly be concerned with circuits where the gates
implement the addition ($+$) and maximum ($\max$) operations, that is, they
output the sum or maximum of their inputs respectively. Boolean circuits can be
interpreted as arithmetic circuits via the natural correspondence that False is
$0$ and True is $1$. To make the converse transition, we will make use of the
following known facts (see e.g. \cite{Cook85}).

\begin{theorem}\label{thm:cook} Given an integer $n$ it is possible to
construct in time polynomial in $n$ bounded fan-in boolean circuits with $n^2$
inputs, $O(n)$ outputs, size $n^{O(1)}$ and depth $O(\log n)$ solving the
following problems: (i) computing the maximum of $n$ integers of $n$ bits each
(ii) computing the sum of $n$ integers of $n$ bits each.  Furthermore, there
exists a bounded fan-in boolean circuit with $2n$ inputs, encoding two $n$-bit
binary numbers, that has depth $O(\log n)$ and decides if the first number is
larger than the second.  \end{theorem}

Using \cref{thm:cook} we will be able to convert an arithmetic circuit which
only uses $\max$ and $+$ gates into a boolean circuit, without increasing the
depth by much, even if the $\max$ and $+$ gates do not have bounded fan-in
(because \cref{thm:cook} talks about computing the sum or maximum of $n$
inputs).

Given a CNF formula $\phi$, a \emph{strong backdoor} set is a set of variables
such that for each assignment to these variables we obtain a formula that
belongs to a base class. We will only deal with strong backdoors, so we will
simply refer to them as backdoors. Furthermore, we will deal with
\textsc{2-SAT} and Horn backdoors, where a formula is Horn if each clause
contains at most one positive literal.  Recall that \textsc{2-SAT} is
NL-complete and \textsc{Horn-SAT} is P-complete, therefore both problems are
solvable in polynomial time. This means that for a given backdoor the
satisfiability of $\phi$ can be decided in time $2^b|\phi|^{O(1)}$ by trying
all assignments to the backdoor. The term backdoor was introduced in
\cite{WilliamsGS03}. Note that backdoors and modulators are closely related
concepts and indeed the notion of modulator to constant treewidth which we use
in this paper has also been called a backdoor to constant treewidth
\cite{GaspersS13}.  For more information on backdoors see \cite{GaspersS12}.

In all cases where it is relevant, we will assume that the modulator or
backdoor related to the problem at hand is given in the input. This is because
we want to concentrate on the complexity of solving \textsc{SAT} given the
modulator, rather than the complexity of finding the modulator. This is
standard in parameterized lower bounds, see e.g. \cite{LokshtanovMS18}.

We give a detailed list of the complexity hypotheses that appear in this paper
as well as simple observations about their relations in \cref{sec:hypo}.

\section{Problems Equivalent to (Max-)CNF-SAT}\label{sec:maxsat}

In this section we deal with two hypotheses: the standard SETH and its natural
weakening for \textsc{Max-SAT}, which we refer to as the \maxsatseth. One of
our main motivations is to take a more careful look at the results of
\cite{EsmerFMR24}. Recall that in this recent work it was pointed out that the
SETH is ``too strong'', if one aims to obtain lower bound results regarding
standard width parameters, such as treewidth. For this reason,
\cite{EsmerFMR24} focuses on a parameter which they call \sdh, which is a
severe restriction of treewidth, pathwidth, and tree-depth modulators. Their
main result is that many of the standard lower bounds go through (assuming the
SETH) even for this very restricted parameter.

In this paper we offer a more nuanced view that clarifies a bit more the
picture presented in \cite{EsmerFMR24}. In particular, we show that for some of
their lower bound results, not only is the SETH not too strong, but actually
the lower bounds they obtain are \emph{equivalent} to the SETH. We achieve this
mainly via \cref{thm:seth} below, which shows that the SETH is equivalent to
the \tdmseth, that is the hypothesis that it is impossible to avoid a
brute-force enumeration for a given modulator to a graph of bounded tree-depth.
Since \sdh\ is a parameter that is a restriction of modulators to bounded
tree-depth, we obtain that avoiding brute force enumeration for this parameter
is also equivalent to the SETH (\cref{cor:sdh}). In other words,
\cref{thm:seth} collapses three seemingly distinct parameters (the number of
variables, the size of a \sdh, and the size of a constant tree-depth modulator)
into a single hypothesis.  As we show in \cref{sec:applications}, these results
establish that the lower bound of \cite{EsmerFMR24} for \textsc{Coloring} is
actually equivalent to the SETH, and this remains true even if we parameterize
by tree-depth modulator rather than \sdh.

Despite the above, our results do not indicate that the SETH is the ``correct''
assumption for all the lower bounds of \cite{EsmerFMR24}. In the second part of
this section we deal with the \maxsatseth, the hypothesis stating the same
bound as the SETH but for the problem of maximizing the number of satisfied
clauses of a CNF formula. Since this is clearly a problem at least as hard as
deciding if a CNF formula can be satisfied, this assumption is at least as
plausible as the SETH. Indeed, this assumption is also implied by the so-called
\textsc{Max-3Sat-Hypothesis}, which states the same time bound but for CNFs of
maximum arity $3$.  As mentioned in \cite{EsmerFMR24}, the relation between the
SETH and the \textsc{Max-3Sat-Hypothesis} is currently unknown, which can
perhaps be seen as a reason to believe the \maxsatseth\ as more plausible than
either of the other two hypothesis. Our second result in this section is that
the \maxsatseth\ remains equivalent if we parameterize not by the number of
variables but by \sdh\ size (\cref{thm:maxsat}). This allows us to establish in
\cref{sec:applications} that lower bounds concerning optimization problems,
such as the one on \textsc{Max Cut} parameterized by \sdh\ in
\cite{EsmerFMR24}, are more likely to be equivalent to \maxsatseth, hence
perhaps more believable, than lower bounds equivalent only to the SETH.

Before we go on, let us mention that the results of this section rely on a
precise formulation of the SETH (and related hypotheses) which assumes hardness
for $k$-CNF formulas, for sufficiently large $k$. We note that this is
important in this section, because we are dealing with extremely restrictive
graph parameters. Once we move on to \cref{sec:ald}, it will immediately become
possible to break down large clauses into smaller ones by introducing new
variables, so the arity of the considered CNFs will no longer be of material
importance (see \cref{obs:arity}).

\subsection{Equivalent Formulations of the SETH}

\begin{theorem}\label{thm:seth} The following statements are equivalent:

\begin{enumerate}

\item (SETH is false) There is an $\epsilon>0$ such that for all $k$ there is
an algorithm that solves $k$-\textsc{SAT} in time $(2-\eps)^n|\phi|^{O(1)}$.

\item (\tdmseth\ is false) There is an $\epsilon>0$ such that for all $c,k$
there is an algorithm that takes as input a $k$-CNF formula $\phi$ and a
modulator $M$ of $\phi$ of size $d$, such that $\phi-M$ has tree-depth at most
$c$, and decides if $\phi$ is satisfiable in time $(2-\eps)^d|\phi|^{O(1)}$.

\end{enumerate}

\end{theorem}

\begin{proof}

The second statement clearly implies the first, as the set of all variables is
a modulator to tree-depth $0$. The interesting direction is to show that if we
have a $k$-\textsc{SAT} algorithm running in time $(2-\eps)^n|\phi|^{O(1)}$,
then this implies a faster algorithm parameterized by constant tree-depth
modulator.

We are given a $k$-CNF formula $\phi$ and a modulator $M$ such that removing
$M$ from the primal graph produces a graph of tree-depth $c$, for some fixed
$c$. Assume that for the resulting graph after we remove $M$ we have a forest
of rooted trees, where each tree has height at most $c$, and a bijection of the
variables of $\phi$ (outside $M$) to the vertices of this forest so that two
variables appear in the same clause if and only if one is an ancestor of the
other.

Our strategy will be to use the standard resolution rule to eliminate from
$\phi$ all the variables except those of the modulator. We will then apply the
assumed algorithm for $k'$-\textsc{SAT}, for some $k'$ depending on $c$ and
$k$.

Recall that resolution is the following rule: when a CNF formula contains two
clauses $C,C'$ such that for some variable $x$, $C$ contains $x$ and $C'$
contains $\neg x$, then we can add to the formula the clause that contains all
literals of $C\cup C'$ except $x,\neg x$ without affecting satisfiability. We
will say that we exhaustively resolve over variable $x$ if we (i) apply this
rule to each pair of clauses $C,C'$ containing $x,\neg x$ (ii) after doing so
remove from $\phi$ all clauses that contain $x$ or $\neg x$.

\begin{claim} Exhaustively resolving over a variable $x$ does not affect the
satisfiability of $\phi$. \end{claim}

\begin{claimproof} It is easy to see that the clauses added to $\phi$ do not
affect satisfiability, as they can be inferred from existing clauses. To see
that deleting all clauses that contain the variable $x$ does not affect
satisfiability, we take two cases. If the new formula has a satisfying
assignment $\sigma$ such that for all clauses $C$ which contain $x$ positive in
$\phi$, a literal of $C$ other than $x$ is set to True by $\sigma$, then we can
satisfy $\phi$ by extending $\sigma$ by setting $x$ to False. Otherwise, we
claim that all clauses $C'$ that contain $\neg x$ in $\phi$ have a literal
other than $\neg x$ set to True by $\sigma$, so we can set $x$ to True and
satisfy $\phi$. Indeed, we are now in a case where some clause $C$ which
contains $x$ positive has all its literals set to False by $\sigma$, but this
clause has been paired with all $C'$ which contained $\neg x$. Since the new
formula is satisfied by $\sigma$, each $C'$ must contain a True literal.
\end{claimproof}

Our algorithm is to apply for each $i\in[c]$, starting from $i=c$ and
descending, exhaustive resolution to all the variables which correspond to
vertices of the rooted trees which are at distance $i$ from the root. After
each iteration we have eliminated a layer of leaves and decreased the
tree-depth of $\phi-M$ by one, without affecting the satisfiability of the
formula. We now show by induction that the result of this process after $c$
iterations is a formula with maximum arity $k'\le 2^ck$ whose number of clauses
is at most $m^{2^c}$, where $m$ is the number of clauses of $\phi$.

We show this claim by induction on $c$. For $c=1$, $M$ is in fact a vertex
cover of the primal graph. For each variable $x$ of the independent set we
apply exhaustive resolution. Each new clause is made up of two old clauses, so
has arity at most $2k$. Furthermore, if $x$ appears in $d(x)$ clauses, we will
construct at most $d^2(x)/2$ clauses when resolving over $x$. The total number
of new clauses is then at most $\sum_{x\not\in M}d^2(x)/2 \le m^2/2$, where we
use the fact that $m\ge \sum_{x\not\in M} d(x)$, because the sets of clauses
that contain two distinct variables not in $M$ are disjoint, and because of
standard properties of the quadratic function. Since we can assume that $m$ is
large enough, the total number of clauses we now have is at most
$m+\frac{m^2}{2}\le m^2$.

We now proceed with the inductive step. Suppose that the claim is true for
smaller values of $c$ and we have an instance of maximum arity $k$ and $m$
clauses. By the same arguments as before, applying exhaustive resolution to all
leaves produces an instance where $\phi-M$ has tree-depth $c-1$, the maximum
arity is at most $k'\le 2k$ and the number of clauses is at most $m'=m^2$. By
inductive hypothesis, completing the procedure will produce a formula with
arity at most $2^{c-1}k'\le 2^ck$ and with at most $(m')^{2^{c-1}}\le m^{2^c}$
clauses, as desired.

In the end we have a formula with $n=|M|$ variables, arity $k'\le 2^ck$ and
size $|\phi'| \le |\phi|^{2^c}$. Running the presumed satisfiability algorithm
on this formula implies the theorem.  \end{proof}

\begin{corollary}\label{cor:sdh} There is an $\epsilon>0$ such that for all
$k,\sigma,\delta>0$ there is an algorithm which takes as input a $k$-CNF
formula $\phi$ and a $(\sigma,\delta)$-hub of size $m$ and decides if $\phi$ is
satisfiable in time $(2-\eps)^m|\phi|^{O(1)}$ if and only if the SETH is false.
\end{corollary}

\begin{proof} The proof consists of observing that a $(\sigma,\delta)$-hub is a
modulator to tree-depth at most $\sigma$ and invoking
\cref{thm:seth}.\end{proof}

\subsection{Equivalent Formulation of the SETH for Max-SAT}

\begin{theorem}\label{thm:maxsat}

The following are equivalent:

\begin{enumerate}

\item (\maxsatseth\ is false) There exists $\eps>0$ such that for all $k$ there
is an algorithm that takes as input a $k$-CNF $\phi$ on $n$ variables and
computes an assignment satisfying the maximum number of clauses in time
$(2-\eps)^n|\phi|^{O(1)}$.

\item (\hubmaxsatseth\ is false) There exists $\eps>0$ such that for all
$k,\sigma,\delta$ there is an algorithm that takes as input a $k$-CNF $\phi$
and a \sdh\ $M$ of size $m$ and computes an assignment satisfying the maximum
number of clauses in time $(2-\eps)^m|\phi|^{O(1)}$.

\end{enumerate}

\end{theorem}

\begin{proof}

It is easy to see that the second statement implies the first, as the set of
all variables is a \sdh. Recall that we allow clauses to be repeated in the
input, so the running time of statement 1 cannot be written as $O((2-\eps)^n)$,
because even for $k$-CNF, for $k$ fixed, the size of the formula cannot be
bounded by a function of $n$ alone.

For the interesting direction, we consider one by one the components of
$\phi-M$, each of which has size at most $\sigma$ and at most $\delta$
neighbors in $M$. Without loss of generality we assume that each component has
size exactly $\sigma$ and exactly $\delta$ neighbors, by adding dummy variables
or neighbors as needed.

Take such a component $C$ and for each assignment $s$ to its $\delta$ neighbor
variables in $M$, calculate the assignment to the variables of $C$ that
maximizes the number of clauses involving variables of $C$ which are satisfied
(this can be done by going through the $2^\sigma$ possible assignments to the
variables of the component).  Let $w_s$ be the number of such clauses. Our goal
is to remove the variables of $C$ from $\phi$ and add some clauses involving
the variables of $N(C)$ such that for each assignment $s$ the total weight of
satisfied new clauses is proportional to $w_s$. More precisely, we do the
following: for each assignment $s$ to the variables of $N(C)$ we construct a
clause \emph{falsified} by this assignment (and only this assignment) to
$N(C)$, and give it weight $\sum_{s'\neq s} w_{s'}$, that is give it weight
that is equal to the sum of the weights corresponding to all other assignments
to $N(C)$.  We now observe that if we use the assignment $s$ for $N(C)$, the
total weight we obtain from the new clauses is $w_s+(2^\delta-2)\sum_{s'}
w_{s'}$, that is, the contribution of each assignment (including $s$) is
counted $2^\delta-2$ times, but $w_s$ is counted once more, as a $w_s$ term
appears in the weight of all satisfied clauses.  We increase the target value
of the \textsc{Max-SAT} instance by $(2^\delta-2)\sum_{s'} w_{s'}$ and
implement the weights by taking copies of the constructed clauses. After
performing this process exhaustively we are left with a \textsc{Max-SAT}
instance with maximum arity $\max\{k,\delta\}$ and size at most
$O(2^{\sigma+\delta}|\phi|)$, so executing the supposed algorithm gives the
lemma.  \end{proof}

\section{Modulators to Constant Width and Linear Depth Circuits}\label{sec:ald}

In this section we present a sequence of reductions which show that several
seemingly distinct fine-grained complexity hypotheses are actually the same.
Our results are summarized in \cref{thm:ald} below, where we initially focus on
the decision version of \textsc{SAT} (we discuss \textsc{MaxSAT} further
below).

\begin{theorem}\label{thm:ald} The following statements are equivalent:

\begin{enumerate}

\item (\aldseth\ is false) There is an $\eps>0$ and an algorithm which takes as
input a (bounded fan-in) boolean circuit of size $s$ on $n$ inputs with depth
at most $\eps n$ and decides if the circuit is satisfiable in time
$(2-\eps)^ns^{O(1)}$.

\item (\pwmseth\ is false) There is an $\eps>0$ such that for all fixed $c>0$
we have an algorithm which takes as input a CNF formula $\phi$ and a size $m$
modulator of $\phi$ to pathwidth at most $c$ and decides if $\phi$ is
satisfiable in time $(2-\eps)^m|\phi|^{O(1)}$.

\item (\twmseth\ is false) The same as the previous statement but for a
modulator to treewidth at most $c$.

\item  (\logtdmseth\ is false) The same as the previous statement, but for a
modulator to tree-depth at most $c\log |\phi|$.

\item (\wsatseth\ is false) There is an $\eps>0$ such that there is an
algorithm which takes as input a general Boolean formula $\phi$ on $n$
variables and an integer $k$ and decides if $\phi$ has a satisfying assignment
that sets exactly $k$ variables to True in time $n^{(1-\eps)k}|\phi|^{O(1)}$.

\end{enumerate}

\end{theorem}

In other words, we make connections between three types of hypotheses: one
involving the difficulty of satisfying a circuit (with a mild restriction on
its depth); three different hypotheses involving structural graph parameters;
and a hypothesis on the complexity of the standard complete problem for the
class W[SAT]. Before we proceed to the proofs, let us give some high-level
intuition regarding these connections.

Among the five statements of \cref{thm:ald} three have an obvious chain of
implications: solving \textsc{SAT} parameterized by the size of a modulator to
logarithmic tree-depth is at least as hard as solving it for a modulator to
constant treewidth, which is at least as hard as for a modulator to constant
pathwidth. These relations follow readily from the fact that $\pw\le tw$ and
$\td \le \tw \log n$ for all graphs. What is perhaps a bit surprising is that
these three parameters turn out to be of equivalent complexity.

In order to add the remaining chains of implications to complete \cref{thm:ald}
we need several steps. The one that is conceptually the least obvious is
perhaps the implication 2$\Rightarrow$1. Here we need to prove that the
satisfiability problem for any circuit (with the stated restriction on its
depth) can be reduced to \textsc{SAT} with a modulator to only constant
pathwidth. The idea now is that the variables of the modulator will correspond
to the inputs of the circuit, so we need to encode the workings of the circuit
via a formula of constant pathwidth. This sounds counter-intuitive, since
evaluating a circuit is a P-complete problem, but constant pathwidth
\textsc{SAT} instances can be solved in logarithmic space. This observation
makes it clear why it is crucial that we have placed a restriction on the depth
of the circuit.  In order to achieve our goal we rely on a classical result
from computational complexity theory known as Barrington's theorem
\cite{Barrington89} which allows us to transform any boolean circuit of depth
$d$ into a branching program of length $2^{O(d)}$ and \emph{constant} width.
The fact that the branching program has constant width means that its workings
can be encoded via a constant pathwidth formula, while the fact that the depth
of the input circuit is at most $\eps n$ means that the exponential blow-up in
the size can be tolerated by carefully dealing with the involved constants.

The other interesting implication is 5$\Rightarrow$4. To better explain this,
suppose we were instead trying to prove 1$\Rightarrow$4. In that case we would
try to construct a circuit whose inputs are the variables of the modulator
given to us in statement 4. For this, we would need to somehow eliminate all
the remaining variables, using the fact that they form a graph of logarithmic
tree-depth. For this we use an inductive construction which proceeds bottom-up,
constructing for each node of the tree a formula that eliminates all variables
in its sub-tree. To do this we try both values of the variable and construct
two copies of the formulas we have for its children, one for each value. This
construction is of course exponential in the tree-depth, but since the
tree-depth is $O(\log |\phi|)$, the result is polynomial in the input formula.
To obtain instead 5$\Rightarrow$4 we use standard tricks that encode an
assignment to $n$ variables via a weight $k$ assignment to $2^{n/k}$ variables.

Finally, for 1$\Rightarrow$5 we can construct a circuit with $k\log n$ inputs,
encoding the $k$ variables to be set to True in the sought weight-$k$
assignment. A subtle point here is that we need to ensure that the circuit we
construct does not have too large depth. However, for this we can rely on the
fact that the W[SAT]-complete problem of statement 5 deals with boolean
formulas, that is, circuits of out-degree $1$, which always admit an equivalent
reformulation with logarithmic depth. Hence, we construct an instance with
$k\log n$ inputs and depth $O(\log n)$, which for $k$ large enough satisfies the
requirements of statement 1.

\begin{proof}[Proof of \cref{thm:ald}]

The implications $4\Rightarrow 3 \Rightarrow 2$ follow immediately from the
fact that on all $n$-vertex graphs we have $\tw\le \pw$ and $\td \le \tw\log
n$. The implication $5\Rightarrow 4$ is given in \cref{lem:wsat}. The
implication $2\Rightarrow 1$ is given in \cref{lem:barrington}. The implication
$1\Rightarrow 5$ is given in \cref{lem:ald}.

\end{proof}

\subparagraph*{MaxSAT} In addition to the above we extend the results of
\cref{thm:ald} to show equivalence also with \textsc{MaxSAT} parameterized by
the size of a modulator to logarithmic tree-depth (hence also constant
pathwidth and treewidth).  We present this result in \cref{sec:maxsattd}. This
extension is non-trivial, as it requires to generalize the basic idea used in
\cref{thm:ald} (namely, eliminating variables using a tree-depth decomposition)
from the realm of boolean to the realm of arithmetic circuits, because now to
select the best assignment for a variable we have to calculate the number of
satisfied clauses (this naturally corresponds to a circuit with $\max$ and $+$
gates). The main challenge is to reduce such circuits to boolean ones without
significantly increasing their depth, so that we can still obtain equivalence
with the \aldseth, and for this we need to rely on classical facts from circuit
complexity such as the fact that computing the sum of $N$ integers on $N$ bits
can be done with boolean circuits of depth only $O(\log N)$ (mentioned
previously in \cref{thm:cook}).

\subsection{W[SAT] to Pathwidth, Treewidth, Logarithmic Tree-depth Modulators}

\begin{lemma}\label{lem:wsat} Statement 5 of \cref{thm:ald} implies  statement
4. \end{lemma}

\begin{proof}

We suppose there exist constants $\eps>0, d>0$ and an algorithm that takes as
input a Boolean formula $\psi$ (that is, a Boolean circuit of maximum
out-degree $1$) with $N$ inputs and an integer $k$ and decides if $\psi$ has a
satisfying assignment of weight $k$ in time $N^{(1-\eps)k}|\psi|^d$.

Fix now a constant $c>0$ and we are given a CNF formula $\phi$ and a modulator
$M$ of size $m$ such that removing all vertices of $M$ from the primal graph of
$\phi$ results in a graph of tree-depth $h$ with $h\le c\log |\phi|$.  Suppose
we are given a rooted tree of height $h$ such that the variables of $\phi$
(outside $M$) can be bijectively mapped to nodes of the tree in a way that if
two variables appear in a clause together, then one is an ancestor of the
other. We want to decide the satisfiability of $\phi$ in time $(2-\eps')^m
|\phi|^{O(1)}$.

Our high-level plan is to construct from $\phi$ a Boolean formula $\psi$ and an
integer $k$ satisfying the following:

\begin{enumerate}

\item $\phi$ is satisfiable if and only if $\psi$ has a satisfying assignment
of weight $k$.

\item $\psi$ has $N=O(2^{m/k})$ input variables and can be constructed in time
$2^{m/k}|\phi|^{O(1)}$.

\end{enumerate}

In particular, we will set $k=\lceil\frac{2d}{\eps}\rceil$. Then
$N^{(1-\eps)k}|\psi|^{d}$ can be upper bounded as follows: $N^{(1-\eps)k} =
O(2^{(1-\eps)m})$ and $|\psi|^d = 2^{dm/k}|\phi|^{O(1)}$. But $\frac{d}{k}\le
\frac{\eps}{2}$, therefore the algorithm runs in
$2^{(1-\frac{\eps}{2})m}|\phi|^{O(1)}$.

We now describe how to construct $\psi$ from $\phi$, assuming we have the
rooted tree of height $h\le c\log |\phi|$ we mentioned covering the primal
graph of $\phi$ outside of $M$.  We first describe the construction of an
auxiliary formula $\psi_r$, whose input variables are only the variables of
$M$, where $r$ is the root of the tree of height $h$. The formula will be built
inductively, by constructing a formula $\psi_b$ for each node $b$ of the tree,
starting from the leaves.  In the remainder we use, for a node $b$ of the tree,
$h_b$ to denote the height of $b$ in the tree, that is its maximum distance
from any of its descendants, and $\phi_b$ to denote the formula made up of the
clauses of $\phi$ that contain a variable associated with $b$ or one of its
descendants.  We can assume without loss of generality that all clauses of
$\phi$ contain a variable outside of $M$ (if not, replace $C$ with $(C\lor x),
(C\lor \neg x)$, where $x$ is a new variable outside $M$), hence we have
$\phi_r=\phi$ and also if $b,b'$ are nodes such that $h_{b}=h_{b'}$, then
$\phi_b$ and $\phi_{b'}$ are made up of disjoint sets of clauses of $\phi$. To
see this, if there were a clause that contains a variable associated with $b$
(or one of its descendants) and $b'$ (or one of its descendants) this would add
to the primal graph an edge that is not between an ancestor and a descendant,
contradiction.

The formula $\psi_b$ will have the following properties:

\begin{enumerate}

\item The input variables of $\psi_b$ are the variables of $M$ and all the
variables associated with ancestors of $b$ in the tree.

\item $\psi_b$ is satisfied by an assignment to its input variables if and only
if it is possible to extend this assignment to the variables associated with
$b$ and its descendants, so that all clauses of $\phi_b$ are satisfied.

\item $|\psi_b| = O( 2^{h_b}|\phi_b|)$.

\end{enumerate}

As mentioned, $\psi_r$ will be the formula obtained by this procedure for the
root of the tree and will only have $M$ as input variables.

In order to construct $\psi_b$ we proceed by induction on $h_b$. If $h_b=0$
then $\psi_b$ is constructed by taking $\phi_b$, and obtaining two formulas
$\phi_b^0, \phi_b^1$ by setting the variable associated with $b$ to False,
True, respectively. We set $\psi_b=\phi_b^0\lor \phi_b^1$.  Clearly, the three
properties are satisfied.

Now suppose $h_b>0$ and let $b_1,\ldots,b_s$ be the children of $b$ in the
tree. Let $x$ be the variable associated with $b$. We have already constructed
formulas $\psi_{b_1},\ldots,\psi_{b_s}$. For each $i\in[s]$ let $\psi_{b_i}^0$
(respectively $\psi_{b_i}^1$) be the formula obtained from $\psi_{b_i}$ if we
set $x$ to False (respectively True) and simplify the formula accordingly. Let
$C_b$ be the set of clauses of $\phi$ containing $x$ but no variable associated
with a descendant of $b$. Let $C_b^0$ (respectively $C_b^1$) be the CNF formula
made up of the clauses of $C_b$ we obtain if we set $x$ to False (respectively
True), that is $C_b^0$ is the set of clauses containing $x$ positive, but with
$x$ removed (similarly for $C_b^1$). We define $\psi_b$ as follows:

\[ \psi_b = \left( C_b^0 \land \bigwedge_{i\in[s]} \psi_{b_i}^0 \right) \lor
\left( C_b^1 \land \bigwedge_{i\in[s]} \psi_{b_i}^1 \right) \]

Let us argue why $\psi_b$ satisfies the three properties. First, $\psi_b$ has
the right set of variables, as we have eliminated $x$ from all formulas. 

Second, suppose $\psi_b$ is satisfied by some assignment to its input variables
and suppose without loss of generality that the first term of the disjunction
is satisfied (the other case is symmetric), therefore $C_b^0$ is satisfied. We
set $x$ to False, and thus satisfy all of $C_b$. For all $i\in[s]$ we have that
$\psi_{b_i}^0$ is satisfied, therefore $\psi_{b_i}$ is satisfied by the current
assignment to $M$ and the ancestors of $b$, extended with $x$ set to False. By
induction, this can be extended to an assignment satisfying $\phi_{b_i}$. Since
we have satisfied $C_b$ and $\phi_{b_i}$ for all $i\in[s]$, we have satisfied
$\phi_b$. For the converse direction, if it is possible to extend some
assignment to $M$ and the ancestors of $b$ so that $\phi_b$ is satisfied, then
this assignment satisfies $\psi_b$. Indeed, suppose that in the assumed
extension $x$ is set to False (again the other case is symmetric). This implies
that all the clauses of $C_b^0$ must be satisfied, as they contain $x$
positive. Furthermore, by induction, since we can satisfy $\phi_{b_i}$ (for all
$i\in[s]$) by setting $x$ to False, we also satisfy $\psi_{b_i}$ by setting $x$
to False, hence $\psi_{b_i}^0$ evaluates to True. Therefore, the first term of
the disjunction of $\psi_b$ evaluates to True.

Third, to bound the size of $\psi_b$ we have $|C_b^0|+|C_b^1|\le 2|\phi_b|$ and
we also have $2\sum_{i\in[s]} \psi_{b_i} = O(2^{b_h}\sum_{i\in[s]}|\phi_{b_i}|)
= O(2^{b_h}|\phi_b|)$, (where we used the fact that the $\phi_{b_i}$ are
disjoint) giving the desired bound.  We therefore have $|\psi_r| =
|\phi|^{O(c)} = |\phi|^{O(1)}$.

We now transform $\psi_r$ into a formula $\psi$ such that $\psi$ has a
satisfying assignment of weight $k$ if and only if $\psi_r$ has a satisfying
assignment. To do so, we partition the $m$ input variables of $\psi_r$ into $k$
sets $V_i$, $i\in[k]$, each of size at most $\lceil m/k\rceil$. For each
assignment $\sigma$ to the variables $V_i$ we construct a new variable
$x_{i,\sigma}$. For each $i\in[k]$ we add a clause $\left(\bigvee_{\sigma}
x_{i,\sigma}\right)$. For each $y\in V_i$ let $\sigma^1(y)$ (respectively
$\sigma^0(y)$) be the set of assignments to $V_i$ which set $y$ to True
(respectively False). Then, we replace every appearance of $y$ in $\psi_r$ with
the clause $\bigvee_{\sigma\in \sigma^1(y)} x_{i,\sigma}$ and every appearance
of $\neg y$ with the clause $\bigvee_{\sigma\in \sigma^0(y)} x_{i,\sigma}$.
Note that now we have $|\psi|=O(2^{m/k}|\psi_r|) = 2^{m/k}|\phi|^{O(1)}$, as
promised. Furthermore, it is not hard to see that $\psi$ can only be satisfied
by an assignment of weight $k$ if we select for each $i\in[k]$ a unique
variable $x_{i,\sigma}$ to set to True, thus creating a one-to-one
correspondence between satisfying assignments to $\phi$ and weight-$k$
satisfying assignments to $\psi$.  \end{proof}

\subsection{Linear-depth Circuits and Constant Pathwidth Modulators}

In this section we present \cref{lem:barrington} which, as mentioned will rely
on the transformation of Boolean circuits to bounded-width branching programs
given by Barrington's theorem. Before we begin let us recall some definitions
and background.

\begin{definition} A layered branching program of width $w$ and length $\ell$
with input variables $x_1,\ldots,x_n$ is a DAG with the following properties:

\begin{enumerate}

\item Vertices are partitioned into $\ell$ sets $V_1,\ldots,V_\ell$ (the
layers), such that each arc goes from a vertex in $V_i$ to a vertex in
$V_{i+1}$. For each $i\in[\ell]$ we have $|V_i|\le w$.

\item There is a unique source (the start vertex), while $V_\ell$ contains two
vertices (the accept and reject vertices). All other vertices have strictly
positive in- and out-degree.

\item Each arc is labeled True or False and each vertex with positive
out-degree is labeled with an input variable.

\end{enumerate} 

We say that a branching program accepts an assignment $\sigma$ to its variables
if when we start at the source and we follow a path where for each vertex
labeled $x_i$ we follow the out-going arc labeled $\sigma(x_i)$, we reach the
accept vertex.  \end{definition}

We now recall Barrington's classical result.

\begin{theorem}\label{thm:barrington}(\cite{Barrington89}) There is an
algorithm which, given a boolean circuit on $n$ inputs with depth $d$, produces
an equivalent branching program on the same inputs with width $5$ and length
$O(4^d)$ in time $O(4^d)$.  \end{theorem}

\begin{lemma}\label{lem:barrington} Statement 2 of \cref{thm:ald} implies
statement 1. \end{lemma}

\begin{proof}

The proof relies heavily on \cref{thm:barrington}: given an arbitrary Boolean
circuit we will first construct an equivalent branching program and then
convert this branching program into a CNF formula. The formula will have a
modulator to small pathwidth because the branching program has constant width.

More precisely, suppose that for some $\eps>0$ there exists an algorithm
which for all $c>0$ takes as input a CNF formula $\phi$ and a modulator $M$
of its primal graph to pathwidth $c$, with $|M|=m$, and decides if $\phi$
is satisfiable in time $2^{(1-\eps)m}|\phi|^{c'}$, for some $c'>0$. Let
$\eps' = \frac{\eps}{4c'}$.  We claim that we can obtain an algorithm which,
given an $n$-input circuit $C$ of size $s$ and depth $\eps'n$ decides if $C$ is
satisfiable in time $2^{(1-\eps')n}s^{O(1)}$.

We first invoke \cref{thm:barrington} and obtain a branching program of length
$O(4^{\eps'n})$ and width $5$. We then construct a CNF formula whose variables
are the $n$ variables $x_1,\ldots,x_n$ representing the inputs to the branching
program, plus the $3\ell$ variables $y_{1,t}, y_{2,t}, y_{3,t}$, for
$t\in[\ell]$. The intuitive meaning of these variables is that the assignment
to $y_{1,t}, y_{2,t}, y_{3,t}$ should encode the vertex which the execution of
the branching program will reach in layer $t$. We make an injective mapping
from the vertices of $V_t$ to the assignments of $y_{1,t}, y_{2,t}, y_{3,t}$
and add clauses to exclude assignments not associated with a vertex of $V_t$.
We then add clauses to ensure that arcs are followed: for each assignment to
$y_{1,t}, y_{2,t}, y_{3,t}, y_{1,t+1}, y_{2,t+1}, y_{3,t+1}$ as well as the at
most five input variables used as labels for the vertices of $V_t$, we
construct clauses ensuring the transitions conform to the arcs of the branching
program. Finally, we add clauses ensuring that in layer $1$ we begin at the
start vertex and in the last layer we reach that accept vertex. It is not hard
to see that the constructed CNF formula is satisfiable if and only if there
exists an assignment to the input variables accepted by the branching program,
and hence the circuit.

We observe that in the formula we have constructed there is a modulator of size
$n$ such that removing those variables results in a graph of pathwidth at most
5 (we can use a sequence of bags containing $\{y_{1,t}, y_{2,t}, y_{3,t},
y_{1,t+1}, y_{2,t+1}, y_{3,t+1}\}$ for all $t$). If we execute the supposed
algorithm on this instance the running time will be at most
$2^{(1-\eps)n}|\phi|^{c'}$. We have $|\phi| = O(4^{\eps'n}) = O(2^{\frac{\eps
n}{2 c'}})$. The algorithm therefore runs in $O(2^{(1-\frac{\eps}{2})n}) =
O(2^{(1-\eps')n})$.  \end{proof}

\subsection{Linear-Depth Circuits and Weighted SAT}

\begin{lemma}\label{lem:ald} Statement 1 of \cref{thm:ald} implies statement 5.
\end{lemma}

\begin{proof}

Suppose there exist $\eps>0, c>0$ such that satisfiability of $n$-input
circuits of depth $\eps n$ and size $s$ can be decided in time
$2^{(1-\eps)n}s^c$. 

We are given an $N$-variable formula $\phi$ and an integer $k$ and are asked to
decide if $\phi$ has a weight-$k$ satisfying assignment in time
$N^{(1-\eps')k}|\phi|^{O(1)}$. We can assume without loss of generality that
$|\phi|\le N^{\eps k/6}$, otherwise the problem can be solved in
$N^k|\phi|^{O(1)} = |\phi|^{O(1/\eps)} = |\phi|^{O(1)}$ time by trying out all
assignments. Therefore, $\log |\phi| \le \frac{\eps k\log N}{6}$.

We now consider $\phi$ as a boolean circuit. Using standard techniques, we can
balance $\phi$ to obtain an equivalent circuit of depth at most $3\log |\phi|
\le \frac{\eps k \log N}{2}$ (see e.g. \cite{BonetB94}) while maintaining a
size that is polynomial in $|\phi|$. We add to $\phi$ $k\lceil\log N\rceil$ new
inputs, partitioned into $k$ groups of $\lceil\log N\rceil$ inputs, whose
intuitive meaning is that they encode the indices of the $k$ inputs to $\phi$
which are meant to be set to True in a satisfying assignment.  For each group
we enumerate all $O(N)$ assignments and associate each with an input of $\phi$.
We transform the $N$ input gates of $\phi$ into disjunction gates of fan-in
$k$, whose inputs are conjunctions of fan-in $\lceil \log N \rceil$ connected
to the new inputs in one of the $k$ groups via appropriate $\neg$ gates so that
an input of $\phi$ is set to True if and only if its index is the assignment to
one of the $k$ groups of new inputs. This construction can be made to have
depth $O(\log\log N+\log k)$ by replacing high fan-in gates with binary trees.
It is now not hard to see that the new circuit has a satisfying assignment that
gives distinct assignments to each of the $k$ groups if and only if $\phi$ has
a weight-$k$ satisfying assignment. What remains is to add for each distinct
$i,i'\in[k]$ some gadget that ensures that the assignments to groups $i,i'$ are
not identical (and hence we select an assignment of weight exactly $k$ for
$\phi$). This can easily be done with $O(k^2)$ formulas of size $O(\log N)$ so
the total depth remains at most the depth of $\phi$ plus $O(\log\log N+\log
k)$. We can therefore assume that the depth of the circuit we have now
constructed is at most $\eps k\log N$. The circuit has at most $k\log N+k$
inputs and has been constructed in polynomial time, so running the supposed
satisfiability algorithm takes $2^{(1-\eps)k\log N}2^k|\phi|^{O(1)}$ which for
$N>2^{2/\eps}$ is upper-bounded by $N^{(1-\eps/2)k}|\phi|^{O(1)}$.  \end{proof}

\subsection{MaxSAT and Modulators of Logarithmic
Tree-depth}\label{sec:maxsattd}

We are now ready to show that the results of \cref{thm:ald} remain valid if
instead of \textsc{SAT} we consider the more general \textsc{MaxSAT} problem.
In particular, a fast satisfiability algorithm for circuits of depth $\eps n$
would be sufficient not only to decide satisfiability of formulas which are
close to having logarithmic tree-depth, but also to decide the maximum number
of satisfiable clauses of such formulas. This will become crucial later,
because the versions of \textsc{Independent Set} we consider will be relatively
straightforward to reduce to \textsc{MaxSAT} on formulas which have a small
modulator to instances of constant treewidth (and hence logarithmic
tree-depth).

\begin{theorem}\label{thm:aldsethmaxsat}

The following two statements are equivalent:

\begin{enumerate}

\item (\aldseth\ is false) There is an $\eps>0$ and an algorithm which takes as
input a (bounded fan-in) boolean circuit of size $s$ on $n$ inputs with depth
at most $\eps n$ and decides if the circuit is satisfiable in time
$(2-\eps)^ns^{O(1)}$.

\item (\textsc{MaxSAT} for logarithmic tree-depth modulators) There is an
$\eps>0$ such that for all fixed $c>0$ we have an algorithm which takes as
input a CNF formula $\phi$, a target value $t$, and a size $m$ modulator of
$\phi$ to tree-depth at most $c\log |\phi|$ and decides if there exists an
assignment that satisfies at least $t$ clauses of $\phi$ in time
$(2-\eps)^m|\phi|^{O(1)}$.

\end{enumerate}

\end{theorem}

\begin{proof}

The implication $2\Rightarrow 1$ follows immediately from \cref{thm:ald},
because an algorithm that can solve \textsc{MaxSAT} can also solve
\textsc{SAT}, implying the fourth statement of \cref{thm:ald}. We therefore
focus on the implication $1\Rightarrow 2$. Fix an $\eps>0$ for which we have an
algorithm that can decide the satisfiability of boolean circuits of depth $\eps
n$ in time $(2-\eps)^ns^{O(1)}$.

Fix a $c>0$ and suppose we are given a formula $\phi$ with a modulator $M$ of
size $m$ such that $\phi-M$ has tree-depth at most $c\log|\phi|$. We want to
find an assignment that satisfies as many clauses of $\phi$ as possible in time
$(2-\eps')^m|\phi|^{O(1)}$, for some $\eps'>0$. We will assume without loss of
generality that, for all fixed $\delta>0$ we have $|\phi|<2^{\delta m}$, as
otherwise we can enumerate all assignments to $M$ in time
$2^m<|\phi|^{1/\delta} = |\phi|^{O(1)}$ and for each such assignment compute
its best extension to the rest of the instance in time $\phi^{O(1)}$ (using the
fact that the remaining instance has logarithmic tree-depth).  By choosing an
appropriate $\delta$ we can therefore assume that any quantity that is $O(\log
|\phi|)$ is upper-bounded by $\eps m$.

Let us first give a high-level overview of our approach. We will begin by
constructing an arithmetic circuit using bounded fan-in $\max$ and $+$ gates
and depth $O(\log |\phi|)$. The idea is that the gates of the circuit will
allow us to eliminate all the variables of $\phi$ except those of $M$. In
particular, for each variable of $\phi-M$ we have a $\max$ gate (we select the
better assignment for the variable) whose inputs are $+$ gates with two inputs:
the number of clauses immediately satisfied by giving a specific value to the
variable, and the circuit corresponding to the remaining formula. Hence,
feeding an assignment to the variables of $M$ to the circuit will compute the
maximum number of clauses that can be satisfied by extending this assignment.
Once we have this arithmetic circuit we would like to convert it into a boolean
circuit, however doing this in a naive way will increase the depth to
super-logarithmic in $|\phi|$. We will therefore first ``compress'' the
circuit: we will replace the bounded fan-in gates with gates of fan-in
$\log\log|\phi|$, in the process reducing the depth to
$\log|\phi|/\log\log|\phi|$. We then use the fact that the sum or maximum of
$N$ numbers of $N$ bits can be computed by a boolean circuit of depth $O(\log
N)$, or in our case, the sum of $\log|\phi|$ numbers on $\log|\phi|$ bits can
be computed by a boolean circuit of depth $\log\log|\phi|$. Plugging this
circuit in the place of high fan-in gates produces the final boolean circuit of
the desired depth.

\smallskip

Let us proceed to the construction. We will construct a boolean circuit with
$n=m$ inputs and depth at most $O(\log|\phi|)<\eps m$, on which we can apply
the supposed algorithm.  The circuit will be satisfiable if and only if there
exists an assignment satisfying $t$ clauses of $\phi$, so we will obtain the
implication.

We first construct an arithmetic circuit. In such a circuit the value of each
gate is an integer and we will use the operations $+$ and $\max$. We describe
the construction of such a circuit by induction.

Our arithmetic circuit will have as inputs the variables of $M$ with the
natural correspondence between boolean and integer values, namely, False is $0$
and True is $1$. In order to allow us to encode literals we will also use
$\neg$ gates, which given input $x$ output $1-x$; these will only be used in
the layer immediately after the inputs, while the rest of the circuit will
consist only of $+$ and $\max$ gates of fan-in $2$, as well as integer
constants.

We make the following inductive claim: we can produce from $\phi$ and $M$ an
arithmetic circuit $C$ which satisfies the following:

\begin{enumerate}

\item For each assignment $\sigma$ to $M$, if we feed $\sigma$ to the circuit
$C$ we obtain as output the number of clauses satisfied by the best assignment
of $\phi$ that is consistent with $\sigma$.

\item If the primal graph of $\phi-M$ is connected, then the depth of $C$ is at
most $3\td(\phi-M)+2\log |\phi|+1$.

\item If the primal graph of $\phi-M$ is not connected, then the depth of $C$
is at most $3\td(\phi-M)+2\log |\phi| + 2$.

\end{enumerate}

\begin{claim} An arithmetic circuit satisfying the above can be constructed in
polynomial time, given $\phi$ and $M$ as well as a tree-depth decomposition of
$\phi-M$.  \end{claim}

\begin{claimproof}

We show the claim by induction. For the base case, suppose $M$ contains all
variables of $\phi$. Then, we can construct a circuit as follows: for each
clause, construct a $\max$ gate, and feed it as input all the literals of the
clause (directly or through $\neg$ gates); then add a $+$ gate that computes
the sum of the gates constructed for the clauses. This circuit uses $+$ and
$\max$ gates of unbounded fan-in, but can easily be converted to one that uses
gates of fan-in $2$ and depth at most $2\log|\phi|+1$ as desired.

We now have two inductive cases, depending on whether $\phi-M$ is connected or
not. First, suppose that $\phi-M$ is connected, there exists therefore a
variable $x\not\in M$ such that $\td(\phi-M-\{x\})=\td(\phi-M)-1$. We compute
two formulas $\phi_0,\phi_1$, by setting $x$ to $0$ or $1$ respectively and
simplifying $\phi$ (that is, deleting satisfied clauses and falsified
literals).  Inductively construct circuits $C_0,C_1$ for the two formulas. The
circuit $C$ has a $\max$ gate as output; we feed into this gate two inputs,
calculated by $+$ gates: the first $+$ gate takes as input the output of $C_0$
and an integer equal to the number of clauses that contain the literal $\neg
x$, while the other $+$ gate takes as input the output of $C_1$ and an integer
equal to the number of clauses that contain $x$.  Correctness of the new
circuit can be established by induction.  The depth is at most
$2+\max\{d_0,d_1\}$, where $d_0,d_1$ are the depths of $C_0, C_1$ respectively.
We have $d_0\le 3\td(\phi-M-\{x\})+2\log |\phi| + 2 \Rightarrow 2+d_0 \le
3\td(\phi -M)+2\log |\phi| +1$, as desired.

Suppose then that $\phi-M$ is disconnected and the connected components
correspond to formulas $\phi_1,\phi_2,\ldots, \phi_\ell$. Intuitively, we would
like to construct the circuits $C_1,\ldots, C_\ell$ and then feed their output
to a $+$ gate, but this is slightly complicated from the fact that we want to
keep the fan-in at $2$. If one of the formulas, say $\phi_1$ is large, that is,
$|\phi_1|\ge |\phi|/2$, then we can proceed by building one circuit $C_1$ for
$\phi_1$ and another $C_2$ for $\phi_2\cup\ldots\cup \phi_\ell$, and feeding
their outputs to a $+$ gate. The circuit $C_1$ has depth $1$ less than the
desired bound (as $\phi_1$ has a connected primal graph), while $C_2$ has
sufficiently smaller depth as its size has decreased by a factor of at least
$2$, compared to $\phi$.

We are therefore left with the case where $\phi-M$ is disconnected into
formulas $\phi_1,\ldots,\phi_\ell$, and all these subformulas have size at most
$|\phi|/2$. Partition the formulas $\phi_1,\ldots,\phi_\ell$ into two sets
$L,R$ by initially setting $L=\{\phi_1,\ldots,\phi_\ell\}$ and $R=\emptyset$
and then repeatedly selecting an arbitrary formula from $L$ and moving it to
$R$ as long as we maintain the invariant that the total size of formulas in $L$
is at least $|\phi|/2$. We now are at a situation where the total size of
formulas in $R$ is at most $|\phi|/2$; $L$ is such that if we remove any
formula it contains then its size becomes at most $|\phi|/2$; all formulas
$\phi_i$ have size at most $|\phi|/2$. We now construct three circuits: a
circuit $C_R$ representing all formulas of $R$; a circuit $C_1$ representing an
arbitrarily chosen formula $\phi_i\in L$; and a circuit $C_L$ representing
$L\setminus\{\phi_i\}$. We feed the outputs of $C_1$ and $C_L$ into a $+$ gate,
and the output of that gate with the output of $C_R$ into an $+$ gate, which is
the final output. We observe that because all three circuits represent formulas
whose size is at least a factor of $2$ smaller than $\phi$, their depths are at
least $3$ less than the desired bound, so the two extra layers we added allow
us to stay within the target depth.  \end{claimproof}

We now show how the circuit we have produced can be turned into a circuit of
smaller depth by increasing the fan-in. The reason this will be helpful is that
the depth of a boolean circuit required to add $N$ integers of $N$ bits is
asymptotically the same as for $2$ integers of $N$ bits, so increasing the
fan-in will allow us to more efficiently convert the arithmetic circuit into a
boolean circuit.

\begin{claim} There is a polynomial-time algorithm that takes as input an
arithmetic circuit using $\max$ and $+$ gates of depth $d$ and fan-in $2$ and
outputs an equivalent circuit of depth $O(d/\log d)$ using $\max$ and $+$ gates
of fan-in at most $O(d)$.  \end{claim}

\begin{claimproof}

To simplify notation we assume that $d$ is a power of $2$. Suppose without loss
of generality that the given circuit is arranged in layers, where layers
alternate between $\max$ and $+$ gates.  Suppose the layers are numbered
$1,\ldots,d$ and gates at layer $i$ have inputs in layer $i-1$. We will
conserve the gates in one out of each $\log d$ layers, call these layers the
selected layers.

Consider now a gate $g$ at a selected layer. Its value depends on the values of
at most $2^{\log d}=d$ gates in the previous selected layer. Consider then the
binary tree rooted at $g$ with the $d$ gates from the previous selected layer
as leaves. We can replace this tree with a constant-depth circuit of fan-in
$O(d)$. In particular, our new tree will have a $\max$ gate as output, which
takes as input $O(d)$ $+$ gates, which each take as input some subset of the
$d$ gates of the previous selected layer. This can be done inductively: for
each gate $g'$ in the path from the previous selected layer to $g$ we construct
such a bounded-depth circuit: if $g'$ is a $\max$ gate we take the circuits
constructed for its inputs and merge their outputs; while if $g'$ is a $+$ gate
we apply the distributive property $\max\{a,b\}+\max\{c,d\} =
\max\{a+c,a+d,b+c,b+d\}$ to maintain the property that the output is a $\max$
gate.  Doing this for all gates in selected layers allows us to eliminate
intermediate layers and obtain the desired depth bound.  \end{claimproof}

We now recall (\cref{thm:cook}) that one can perform additions and maximum
calculations of $N$ numbers on $N$ bits with depth $O(\log n)$ boolean
circuits.  Putting this together with the two claims we obtain the proof of the
statement. In particular, we initially constructed a bounded fan-in arithmetic
circuit of depth $O(\log |\phi|)$ and then after applying the second claim and
\cref{thm:cook} converted it to an equivalent boolean circuit of depth
$O(\log|\phi|)$.  We add to the end a boolean comparator that checks if the
output is the binary encoding of at least the target value $t$. The whole
circuit now has depth $O(\log|\phi|)$, which as we explained can be assumed to
be at most $\eps m$, so the supposed algorithm for circuit satisfiability can
be used to solve \textsc{MaxSAT} on the given instance.  \end{proof}

\section{Modulators to Logarithmic Pathwidth and 2-SAT
Backdoors}\label{sec:logpw}

In this section we show an equivalence between the following two hypotheses.

\begin{theorem}\label{thm:2sat}

The following are equivalent:

\begin{enumerate}

\item (\twosatseth\ is false) There is an algorithm that takes as input a CNF
formula $\phi$ and a strong \textsc{2-SAT} backdoor of size $b$, and decides if
$\phi$ is satisfiable in time $(2-\eps)^b|\phi|^{O(1)}$.

\item (\logpwmseth\ is false) There is an algorithm that takes as input a CNF
formula $\phi$, a modulator $M$ of size $m$, and a path decomposition of the
primal graph of $\phi-M$ of width $O(\log|\phi|)$ and decides if $\phi$ is
satisfiable in time $(2-\eps)^m|\phi|^{O(1)}$.

\end{enumerate}

\end{theorem}

Before we proceed, let us give some intuition. One can think of both problems
that we are dealing with here as reachability problems, in the following sense:
if we are given an assignment to a \textsc{2-SAT} backdoor, we can verify if it
can be extended to a satisfying assignment by solving \textsc{2-SAT}, which is
known to be complete for the class NL, that is, equivalent to reachability. On
the other hand, if we have an instance of logarithmic pathwidth, the intuition
of the class EPNL defined in \cite{IwataY15} (see also \cite{Lampis25})
is that solving such an instance corresponds exactly to deciding if a
non-deterministic machine with space $O(\log n)$ will accept a certain input,
that is, this problem is again complete for NL. This intuition leads us to
believe that modulators to logarithmic pathwidth and \textsc{2-SAT} backdoors
are equivalent.

There is, however, one complication in the above reasoning. What we have
established is that once someone supplies us with an assignment to the
\textsc{2-SAT} backdoor, checking if this assignment is correct has the same
complexity as checking whether an assignment to the logarithmic pathwidth
modulator is correct. However, what we really care about is whether an
algorithm which avoids enumerating all assignments in one case can be used to
avoid enumerating the assignments in the other. Hence, an argument which
assumes that the assignment is known cannot be used as a black box. What we
need is a reduction that transfers the answer to all possible assignments for a
backdoor to the possible assignments to a modulator, without going through the
assignments (ideally running in polynomial time).

Given this, it is natural to revisit the reasons for which both \textsc{2-SAT}
and \textsc{SAT} on graphs of logarithmic pathwidth are essentially equivalent
to reachability. We then run into the following obstacle: while \textsc{SAT} on
graphs of logarithmic pathwidth easily reduces to and from reachability,
\textsc{2-SAT} is actually only easily seen to be equivalent to the
\emph{complement} of reachability. The reason \textsc{2-SAT} ends up being
complete for NL is the famous Immerman–Szelepcsényi theorem, which established
that NL=coNL \cite{Immerman88}.

We are therefore in a situation where, if we attempt to go through
reachability, one problem we are dealing with has a quantification of the form
$\exists\exists$ (for logarithmic pathwidth modulators the question is if there
exists an assignment to the modulator so that the resulting graph has an $s\to
t$ path) and the other has the form $\exists\forall$ (does there exist an
assignment to the backdoor so that in the graph representing the \textsc{2-SAT}
implications there is no cycle creating a contradiction). What we need is to
change the inner quantifier without knowing the value of the variables of the
first quantifier block. We are thus led to the following definition of two
auxiliary problems:

\begin{definition} 

Suppose we are given a DAG $G=(V,E)$, a designated source $s$ and sink $t$ of
$G$, a set $M$ of $m$ boolean variables $x_1,\ldots,x_m$, and an annotation
function which assigns to some of the arcs of $G$ a literal (that is, a
variable or its negation) from $M$.

The \reach\ problem is the following: decide if there exists an assignment to
$M$ such that if we only keep in the graph arcs with no annotation or with an
annotation set to True by the assignment, then there is a directed path from
$s$ to $t$.

The \nonreach\ problem is the following: decide if there exists an assignment
to $M$ such that if we keep the same arcs as above, then there is no directed
path from $s$ to $t$.

\end{definition}

Intuitively, \reach\ allows us to encode $2^m$ instances of reachability
simultaneously: for each assignment to the variables of $M$ we obtain a
distinct graph. It is important to note here that \reach\ and \nonreach\ are
\emph{not} complementary problems. In both problems the outer quantification is
existential: we ask if there exists an assignment to $M$ such that, in one case
a path exists, and in the other no path exists.

Our strategy is now the following: we will show that \reach\ is equivalent to
\textsc{SAT} parameterized by a logarithmic pathwidth  modulator (intuitively
this is easy, because both problems have an $\exists\exists$ quantification);
we will show that \textsc{SAT} parameterized by a \textsc{2-SAT} backdoor is
equivalent to \nonreach\ (again, both problems have an $\exists\forall$
quantification).  Finally, we will revisit the proof of the
Immerman–Szelepcsényi theorem and show that the key arguments still go through.
In particular, the Immerman–Szelepcsényi theorem can be seen as a reduction
from reachability to its complement. What we need to show is that this
reduction can still be executed if arcs are annotated, as in our problems, in a
way that preserves the answer for all assignments to $M$. 

\begin{proof}[Proof of \cref{thm:2sat}] We split the proof into three parts. In
\cref{lem:2sat} we show an equivalence between \textsc{2-SAT} backdoors and
\nonreach; in \cref{lem:logpw} we show an equivalence between \textsc{SAT} with
a modulator to logarithmic pathwidth and \reach; and in \cref{lem:immerman} we
show that \reach\ and \nonreach\ have the same complexity by revisiting the
proof of the Immerman–Szelepcsényi theorem.  \end{proof}

\subsection{2-SAT and Non-Reachability}

\begin{lemma}\label{lem:2sat} There is an algorithm solving \nonreach\ in time
$(2-\eps)^mn^{O(1)}$, for some $\eps>0$ if and only if the \twosatseth\ is
false.  \end{lemma}

\begin{proof} The proof is broken down into \cref{lem:2sat1} and
\cref{lem:2sat2}. \end{proof}

\begin{lemma}\label{lem:2sat1} If there is an algorithm falsifying \twosatseth,
then there is an algorithm solving \nonreach\ in time
$(2-\eps)^mn^{O(1)}$.\end{lemma}

\begin{proof}

We use a standard reduction of non-reachability to \textsc{2-SAT}, except we
have to take into account the arc annotations. Suppose we are given a
\nonreach\ instance. Let $G=(V,E)$ be the given DAG on $n$ vertices. We
construct a CNF formula with variables the variables of $M$, given with the
\nonreach\ instance, to which we add a variable $x_v$ for each $v\in V$. We
will set $M$ to be the backdoor set. Now, for each arc $(v_1,v_2)\in E$ we add
the clause $(\neg x_{v_1}\lor x_{v_2})$ to the formula if the arc is not
annotated, otherwise we add the clause $(\neg x_{v_1}\lor x_{v_2}\lor \neg
\ell)$ if the arc is annotated with the literal $\ell$. We also add the unit
clauses $(x_s)$ and $(\neg x_t)$. This completes the construction.

We now claim that if the formula we constructed is satisfiable, using the same
assignment to $M$ yields a digraph where there is no path from $s$ to $t$.
Indeed, suppose that for the same assignment there is a path from $s$ to $t$,
meaning that the assignment sets to True all the literals used as annotations
to the edges of the path. This implies that in the corresponding clauses, the
backdoor literals we have added are set to False by the satisfying assignment.
But then this would create a chain of implications given $x_s\to x_t$, which in
conjunction with the two unit clauses would falsify the formula.

For the converse direction, suppose that there is an assignment to $M$ that
makes $t$ unreachable from $s$. We use the same assignment in the formula and
extend it by setting $x_v$ to True if and only if $v$ is reachable from $s$
under the same assignment to $M$. It is not hard to see that this assignment
will validate the formula: the only possibility of failure is for a clause
representing the arc $(v_1,v_2)$, where $v_1$ is reachable from $s$ and $v_2$
is not. In such a case, the arc must be annotated by a literal that is set to
False by the assignment to $M$, so the corresponding clause is satisfied.

Finally, we observe that the construction is polynomial and the size of the
backdoor is exactly $m$, as each clause contains at most $2$ of the new
variables.  \end{proof}

\begin{lemma}\label{lem:2sat2} If there is an algorithm solving \nonreach\ in
time $(2-\eps)^mn^{O(1)}$, for some $\eps>0$, then the \twosatseth\ is false.
\end{lemma}

\begin{proof}

We want to reduce \textsc{2-SAT} to \nonreach\ while using the annotation to
encode the backdoor variables. We are given a CNF formula $\phi$ with a set
$B=\{y_1,\ldots,y_b\}$ of $b$ backdoor variables and $n$ non-backdoor variables
$X=\{x_1,\ldots,x_n\}$. We assume without loss of generality that each clause
contains exactly $2$ non-backdoor variables: indeed, no clause can contain
three or more such variables (otherwise $B$ would not be a strong backdoor),
and if $C$ contains fewer than $2$ such variables, we can replace it with the
pair of clauses $(C\lor x_i), (C\lor \neg x_i)$, for some $x_i$ that does not
appear in $C$.

We now describe a basic implication digraph: for each variable $x_i\in X$ we
construct two vertices $x_i, \neg x_i$. For each clause $C$, let
$\ell_1,\ell_2$ be the literals of $C$ made up of variables of $X$. We
construct an arc from $\neg \ell_1$ to $\ell_2$ and an arc from $\neg \ell_2$
to $\ell_1$. If $C$ has size $2$ we do not annotate these arcs. Otherwise, if
$C$ contains $b'\le b$ other literals, we subdivide these arcs so that they
become directed paths of length $b'$ and annotate each arc of the resulting
path with the negation of a distinct literal of $C$ involving a variable of
$B$. The intuitive meaning of this construction is that we have a path from
$\neg \ell_1$ to $\ell_2$ and from $\neg \ell_2$ to $\ell_1$ if and only if all
the backdoor literals of $C$ are set to False.

The digraph above is not necessarily a DAG, but we can construct an equivalent
DAG as follows: take $2n$ copies of the graph above, number them $1,\ldots,
2n$, for each vertex $\ell$ in copy $i$ add non-annotated arcs to the same
vertex in all later copies, and for each edge $(\ell_1,\ell_2)$ and each
$i\in[2n-1]$ add an arc from the vertex $\ell_1$ in copy $i$ to the vertex
$\ell_2$ in copy $i+1$ maintaining the annotation of the arc. We now have the
property that there is a path from $\ell_1$ to $\ell_2$ in the original graph
if and only if there is a path from $\ell_1$ in the first copy to $\ell_2$ in
the last copy of the vertex set in the resulting DAG. This equivalence
continues to hold taking into account the arc annotations. Call the DAG we have
constructed the implication DAG and denote it by $D$.

Now our construction is as follows: start with two vertices $s,t$ and for each
$i\in[n]$ construct $2$ copies of the implication DAG, call them $D_{i,0},
D_{i,1}$ for $i\in[n]$.  For $i\in[n]$ put an arc from $s$ to the first vertex
$x_i$ in $D_{i,0}$ from the last $\neg x_i$ in $D_{i,0}$ to the first $\neg
x_i$ in $D_{i,1}$ and from the last $x_i$ in $D_{i,1}$ to $t$.  We set $M=B$
and this completes the construction which can clearly be carried out in
polynomial time. What remains is to argue for correctness.

Suppose now that the original CNF is satisfiable. We want to show that by using
the same assignment to $M$ we eliminate all $s\to t$ paths. Indeed, suppose we
find such a path, which for some $i\in[n]$ must go through $D_{i,0}, D_{i,1}$.
In particular, this means we have a path from $x_i$ to $\neg x_i$ and a path
from $\neg x_i$ to $x_i$ in $D$ while respecting the annotations. However,
recall that for each clause we have added a path of annotated arcs containing
the negations of all backdoor literals of the clause, meaning that if we can
follow an arc, the assignment to $M$ must be setting all the backdoor literals
to False. We conclude that the assignment cannot be satisfying to the formula,
as we can infer the implications $(x_i\to \neg x_i)$ and $(\neg x_i\to x_i)$.

For the converse direction, suppose there is an assignment to $B$ which
eliminates all $s\to t$ paths. We claim the same assignment can be extended to
a satisfying assignment of the original formula. Suppose for contradiction that
this is not the case and suppose that we are working with a maximal
counter-example, that is, for the given $n,b$ we have started our reduction
from a formula $\phi$ that satisfies the following: (i) $\phi$ is unsatisfiable
(ii) there is an assignment to $B$ that eliminates all $s\to t$ paths (iii)
every other $\phi'$ with the same backdoor size $b$ that satisfies (i) and (ii)
either has $n'>n$ non-backdoor variables or the \nonreach\ instance constructed
for $\phi'$ has at most as many arcs as that for $\phi$. In other words, $\phi$
is the counter-example with smallest $n$ and among these has the maximum number
of arcs in the resulting \nonreach\ instance. Assuming these we will reach a
contradiction, which will show that a counter-example cannot exist.

We will attempt to find a good assignment to $\phi$. Consider an $i\in [n]$ and
since there is no path from $s$ to $t$ it must be the case that either there is
no path from (the first) $x_i$ to (the last) $\neg x_i$ in $D$, or that there
is no path from (the first) $\neg x_i$ to (the last) $x_i$ in $D$, or both.

Suppose that there is no path from $x_i$ to $\neg x_i$, but that there is a
path from $\neg x_i$ to $x_i$, for some $i\in[n]$. We then set $x_i$ to True.
Symmetrically, suppose that there is no path from $\neg x_i$ to $x_i$, but that
there is a path from $x_i$ to $\neg x_i$. Then we set $x_i$ to False. Do this
for all variables that satisfy one of these conditions.

We claim that this assignment so far satisfies all clauses containing a
variable that has been set. Indeed, suppose $x_i$ has been set to True but
there is a clause such that all its backdoor literals are false and its other
non-backdoor literal is $\ell$. The interesting case is if $C$ contains $x_i$
negated (otherwise $C$ is satisfied by $x_i$), so we have the directed paths
$x_i\to \ell$ and $\neg \ell \to \neg x_i$, which can be followed under the
current assignment to $M$.  Then, there is a path $\neg \ell \to \neg x_i \to
x_i \to \ell$, which implies that the variable involved in $\ell$ has been set
to a value that satisfies $C$. The reasoning is the same if $x_i$ has been set
to False.

At this point we have an assignment to $B$ and an assignment to $X'\subseteq X$
which together satisfy all clauses involving variables of $X'$. If $X'\neq
\emptyset$ then we are done, because we can give the assignment we have
calculated to $X'$ and remove from $\phi$ all clauses containing a variable
from $X'$. The new formula $\phi'$ would be a smaller counter-example,
contradicting the selection of $\phi$. Indeed, the \nonreach\ instance for
$\phi'$ is a sub-graph of the one for $\phi$, so it cannot contain any $s\to t$
path, while $\phi'$ must be unsatisfiable, because if it were satisfiable,
using the assignment we calculated for $X'$ would satisfy all of $\phi$.

We then must have $X'=\emptyset$. This implies that for all $i\in[n]$ we have
neither a path from $x_i$ to $\neg x_i$ nor from $\neg x_i$ to $x_i$ in $D$. We
claim that setting any $x_i$ to True and simplifying the formula also yields a
smaller counter-example. Indeed, it will be more convenient to do this by
adding the clause $(x_i\lor x_i)$ to $\phi$, obtaining an \nonreach\ instance
with more arcs (this clause cannot have already existed, as then we would have
had $X'\neq \emptyset$). The new formula is clearly also unsatisfiable (we
added a clause to $\phi$ which was assumed unsatisfiable). Furthermore, there
is still no $s\to t$ path in the \nonreach\ instance for $\phi'$. To see this,
suppose that there is now such a path $s\to x_j \to \neg x_j \to \neg x_j \to
x_j \to t$. Previously, neither the $x_j\to \neg x_j$ path (in $D_{j,0}$) nor
the $\neg x_j\to x_j$ path (in $D_{j,1}$) existed, so both paths must be using
the arc $\neg x_i \to x_i$ which is the only newly added arc to $D$. So we
already had paths $x_i \to \neg x_j$ and $\neg x_j\to \neg x_i$, which means we
already had a path $x_i\to \neg x_i$, contradicting the assumption that
$X'=\emptyset$.  We conclude that $\phi'$ is a counter-example with the same
number of variables but strictly more arcs in the corresponding \nonreach\
instance, contradicting the selection of $\phi$.  Hence, no counter-example
exists.  \end{proof}

\subsection{Log-Pathwidth and Reachability}

\begin{lemma}\label{lem:logpw} There is an algorithm solving \reach\ in time
$(2-\eps)^mn^{O(1)}$, for some $\eps>0$ if and only if the \logpwmseth\ is
false.  \end{lemma}

\begin{proof}

For one direction, assume that there is an algorithm solving \reach\ in time
$(2-\eps)^mn^{O(1)}$ and we are given a CNF formula $\phi$ and a modulator $M$
of size $m$ such that removing $M$ from $\phi$ results in a graph of pathwidth
$O(\log |\phi|)$. Suppose we are given a path decomposition of the resulting
graph, and the bags are numbered $B_1,\ldots, B_r$. Recall that each clause of
$\phi$ induces a clique on the primal graph and by standard properties of tree
decompositions for each clause $C$ there is a bag that contains all variables
of $C$ (outside of $M$). Repeating bags as necessary we assume that there is an
injective mapping from the clauses to the bags so that each clause is mapped to
a bag that contains all its variables outside of $M$. Note that if a clause
contains only variables of $M$, then it can be mapped to any bag.

We construct an \reach\ instance as follows. For each $i\in [r]$ and for each
assignment $\sigma$ to the variables of $B_i$ we construct two vertices
$v_{i,\sigma,1}, v_{i,\sigma,2}$. For each $i\in[r-1]$, for each assignment
$\sigma$ to $B_i$ and $\sigma'$ to $B_{i+1}$ such that $\sigma,\sigma'$ agree
on $B_i\cap B_{i+1}$, we add an arc $(v_{i,\sigma,2}, v_{i+1,\sigma',1})$
without annotation. 

For each $i\in[r]$ and assignment $\sigma$ to $B_i$ we do the following: if
$\sigma$ satisfies the (at most one) clause $C$ mapped to $B_i$ or if no clause
is mapped to $B_i$, then add an arc $(v_{i,\sigma,1}, v_{i,\sigma,2})$ without
annotation. If there is a clause $C$ mapped to $B_i$ but $C$ is not satisfied
by $\sigma$, we add an arc $(v_{i,\sigma,1}, v_{i,\sigma,2})$ annotated with
each literal of $C$ from $M$. In particular, if $C$ contains no literals from
$M$ we do not add an arc $(v_{i,\sigma,1}, v_{i,\sigma,2})$, while if it
contains several, we add several parallel arcs (we can then subdivide them to
avoid constructing a multi-graph), each annotated with a distinct literal. We
add a vertex $s$ and arcs $(s, v_{1,\sigma,1})$ for all $\sigma$, as well as a
vertex $t$ and arcs $(v_{r,\sigma,2},t)$. We retain the set $M$ between the two
instances.

It is now not hard to see that there is a one-to-one correspondence between
satisfying assignments to $\phi$ and $s\to t$ paths. Fix an assignment to $M$.
If there is an $s\to t$ path, it must use for each $i$ an arc $(v_{i,\sigma,1},
v_{i,\sigma,2})$, which allows us to infer an assignment $\sigma$ for $B_i$.
This assignment is consistent because we only added arcs between bags for
consistent pairs of assignments. It is satisfying because if a clause $C$ is
mapped to $B_i$ we can only use the arc $(v_{i,\sigma,1}, v_{i,\sigma,2})$ if
$\sigma$ satisfies the clause or if the assignment to $M$ agrees with one of
the literals of the clause; furthermore, all clauses are mapped to some bag.
For the converse direction, if there is an assignment we construct a path by
using in each bag the arc $(v_{i,\sigma,1}, v_{i,\sigma,2})$ corresponding to
the assignment $\sigma$ used for $B_i$.

Let us now move to the other half of the lemma. Suppose there is an algorithm
deciding satisfiability with a modulator $M$ to pathwidth $O(\log |\phi|)$ and
we want to use this algorithm to solve \reach. We construct a formula which
retains the set $M$ of annotation variables as the modulator.  Given a digraph
$G=(V,E)$ with $n$ vertices we assume without loss of generality that $n$ is a
power of $2$ (otherwise add some dummy vertices). We construct a formula with
$n\log n$ extra variables, called $x_{i,j}$, $i\in[n], j\in[\log n]$. The
intended meaning is that the $\log n$ variables $x_{i,1},\ldots,x_{i,\log n}$
encode the index of the $i$-th vertex we would visit in an $s\to t$ path. To
enforce this meaning we assume the vertices of $G$ are numbered from $0$ to
$n-1$ and with each vertex we associate a distinct truth assignment to $\log n$
variables by reading these variables as a binary number. We add clauses
enforcing the variables $x_{1,j}, j\in[\log n]$ to take the value corresponding
to $s$ and the variables $x_{n,j}, j\in[\log n]$ to take the value
corresponding to $t$. For each $i\in[n-1]$ we add at most $n^2$ clauses to the
formula as follows: for each $(u,v)\not\in E$ such that $u\neq v$, we add a
clause ensuring that if the $x_{i,j}$ variables encode $u$ and the $x_{i+1,j}$
variables encode $v$, then the formula is not satisfied; for each $(u,v)\in E$
which is annotated with a literal $\ell$ we add the same clause (which is
falsified if $x_{i,j}, x_{i+1,j}$ encode $u,v$), but we add to this clause the
literal $\ell$. If $(u,v)\in E$ is not annotated or if $u=v$, we do not add a
clause. It is not hard to see that, if we remove the variables of $M$, this
formula has pathwidth $O(\log n)$, because we can make $n-1$ bags, each
containing $x_{i,j}, x_{i+1,j}, j\in[\log n]$ and this covers all edges.

To argue for correctness, if there is an $s\to t$ path in the original instance
for some assignment to $M$ we use the same assignment in $\phi$ and extend it
to a satisfying assignment as follows: if the path visits $r\le n$ vertices in
total, we set for each $i\in[r]$ the assignment to $x_{i,j}, j\in[\log n]$ to
encode the $i$-th vertex of the path; all remaining groups encode $t$. This
satisfies the formula. For the converse direction, we can infer from a
satisfying assignment to the formula a path of length at most $n-1$ by
interpreting the assignment to $x_{i,j}, j\in[\log n]$ as the encoding of the
$i$-th vertex of the path, and the clauses we have added ensure that the path
is valid, that is we are only followed to stay in the same vertex, follow an
arc without annotation, or follow an arc whose annotation is set to True by the
assignment to $M$.  \end{proof}

\subsection{Reachability and Non-Reachability}

\begin{lemma}\label{lem:immerman} For all $\eps>0$, there is an algorithm
solving \reach\ in time $(2-\eps)^mn^{O(1)}$ if and only if there is an
algorithm solving \nonreach\ in time $(2-\eps)^mn^{O(1)}$. \end{lemma}

\begin{proof}

We describe a polynomial-time algorithm which takes as input a DAG
$G_1=(V_1,E_1)$, with $|V_1|=n$, a source $s$, a sink $t$, with the arcs
annotated with variables from a set $M$. The algorithm produces a new DAG
$G_2=(V_2,E_2)$, a source $s_2$, a sink $t_2$, with arcs annotated from the
same set $M$. It guarantees the following: for each assignment to $M$, $G_2$
has an $s_2\to t_2$ path respecting the assignment if and only if $G_1$ does
not have any $s\to t$ path respecting the assignment.  Observe that if we
achieve this, we immediately obtain reductions from \reach\ to \nonreach\ in
both directions, implying the lemma.

We will reuse the ideas of the famous Immerman–Szelepcsényi theorem, so the
reader may think of $G_2$ as a graph that encodes the working of a
non-deterministic machine which is verifying that $G_1$ has no $s\to t$ path.
The construction is of course made more complicated because we also need to
take into account the arc annotations of $G_1$.

Before we begin, let us slightly edit $G_1$. Take $n=|V_1|$ copies of $V_1$,
constructing for each $v\in V_1$ the vertices $v^1,\ldots,v^{n}$. For each
$i\in [ n-1 ]$ and $v\in V_1$ we add the arc $(v^i, v^{i+1})$ without
annotation; for each arc $(u,v)$ and each $i\in [ n-1 ]$ we add the arc $(u^i,
v^{i+1})$, maintaining the annotation of $(u,v)$. Finally, for the vertex
$s^{n-1}$ and all $v\in V_i\setminus\{t\}$ we add the non-annotated arc
$(s^{n-1}, v^{n})$. This transformation has achieved the following: for each
assignment to $M$, an $s\to t$ path exists in $G_1$ if and only if an $s^1\to
t^{n}$ path exists in the new graph; there is a path from $s_1$ to $v^{n}$ for
all $v\in V_1\setminus\{t\}$. As a result the question of whether an $s\to t$
path exists in $G_1$ is equivalent to asking whether the number of vertices
reachable from $s^1$ in $\{ v^{n}\ |\ v\in V_1\}$ is equal to $n$. Call the new
graph $G_1'$ and we want to construct $G_2$ such that for all assignments to
$M$, there is a path $s_2\to t_2$ in $G_2$ if and only if $n-1$ vertices of $\{
v^{n}\ |\ v\in V_1\}$ are reachable from $s^1$ in $G_1'$. In the remainder,
when we refer to the $i$-th layer of $G_1'$ we will mean the set $\{ v^i\ |
v\in V_1\}$. Observe that $G_1'$ has the property that all arcs go from a layer
$i$ to a layer $i+1$. 

We begin the construction of $G_2$ with $n^2$ vertices $s_{i,j}$ for $i,j \in
[n]$, plus a source $s_2$ and a sink $t_2$. The intended meaning is that (for
all assignments to $M$) $s_{i,j}$ will be reachable from $s_2$ if and only if
there are exactly $j$ vertices of $G_1'$ in layer $i$ which are reachable from
$s^1$. We add a non-annotated arc $(s_{n,n-1},t_2)$ and a non-annotated arc
$(s_2,s_{1,1})$. For each $i,j\in [n]$ we want to construct a checker gadget
$C_{i,j}$ such that there will be an arc from $s_{i,j}$ to $C_{i,j}$ and arcs
from $C_{i,j}$ to $s_{i+1,j'}$ for some $j'\in[n]$. We want to achieve the
following property for all assignments $\sigma$ to $M$. Suppose that under some
assignment $\sigma$, there are exactly $j$ vertices of layer $i$ of $G_1'$
reachable from $s^1$. Then, there will be a path from $s_{i,j}$ to $s_{i+1,j'}$
through $C_{i,j}$ if and only if the number of vertices of layer $i+1$ of
$G_1'$ reachable from $s^1$ is $j'$. If we achieve this property we are done,
as any path starts from $s_{1,1}$ and it is the case that exactly $1$ vertex of
layer $1$ is reachable from $s^1$ ($s^1$ itself), so we can proceed by
induction and we see that there is a path to $s_{n,n-1}$ (hence to $t_2$) if
and only if $n-1$ vertices of the last layer of $G_1'$ are reachable.  Hence,
there will be a $s_2\to t_2$ path in $G_2$ if and only if there is no $s_1\to
t_1$ path in $G_1$.

Let us now describe the construction of the checker $C_{i,j}$. For each
$\alpha,\beta \in \{0,\ldots,n\}$ we construct a vertex $x_{i,j,\alpha,\beta}$.
We place an incoming non-annotated arc to $x_{i,j,0,0}$ from $s_{i,j}$. The
intended meaning of these vertices is the following: assuming that $j$ is
correct (that is, there are exactly $j$ vertices of layer $i$ reachable from
$s^1$ in $G_1'$), there should exist a path from $s_{i,j}$ to
$x_{i,j,\alpha,\beta}$ if and only if among the $\alpha$ first vertices of
layer $i+1$ of $G_1'$ there are exactly $\beta$ which are reachable from $s^1$.
Clearly, the property holds for $x_{i,j,0,0}$. Given this intended meaning we
place a non-annotated arc from $x_{i,j,n,\beta}$ to $s_{i+1,\beta}$ for all
$\beta\in [n]$. Now if we establish that the vertices $x_{i,j,\alpha,\beta}$
satisfy the property we stated, then the checker gadget functions correctly.
What remains then is to add some gadgets that imply this property.

For each $\alpha\in\{0,\ldots,n-1\}$ and $\beta\in\{0,\ldots,n-1\}$ we will
construct a gadget called $\mathrm{Adj}_{i,j,\alpha,\beta}$ and a gadget called
$\mathrm{NonAdj}_{i,j,\alpha,\beta}$, each with a single input and output
vertex.  The intended meaning is that a path should exist from the input to the
output of the $\mathrm{Adj}$ gadget if and only if the $(\alpha+1)$-th vertex
of layer $i+1$ of $G_1'$ is reachable from $s^1$; such a path should exist in
the $\mathrm{NonAdj}$ gadget if and only if the same vertex is \emph{not}
reachable from $s^1$. Given this meaning we place a non-annotated arc from
$x_{i,j,\alpha,\beta}$ to the input of $Adj_{i,j,\alpha,\beta}$ and from the
output of that gadget to $x_{i,j,\alpha+1,\beta+1}$; and an annotated arc from
$x_{i,j,\alpha,\beta}$ to the input of $NonAdj_{i,j,\alpha,\beta}$ and from the
output of that gadget to $x_{i,j,\alpha+1,\beta}$. What remains is to describe
these gadgets and prove they work as expected.

The $\mathrm{Adj}_{i,j,\alpha,\beta}$ gadget is conceptually simple. Suppose
that $v$ is the $(\alpha+1)$-th vertex of $V_1$ and we are testing if $v^{i+1}$
is reachable from $s^1$ in $G_1'$. The gadget is simply a copy of the first
$i+1$ layers of $G_1'$, with input vertex $s^1$ and output vertex $v^{i+1}$.
Clearly, there is a path from the input to the output of the gadget if and only
if $v^{i+1}$ is reachable from $s^1$ in $G_1'$.

The $\mathrm{NonAdj}_{i,j,\alpha,\beta}$ gadget is more complicated as it
relies on the key insight of the Immerman–Szelepcsényi theorem which is that we
can certify that a vertex $v$ is not reachable from $s$ at distance $i+1$ if we
know how many vertices are reachable at distance $i$, because we can list all
such vertices, verify that their cardinality is correct, and verify that none
of them have an arc to $v$. For each $\gamma \in \{0,\ldots,n\}$ and $\delta\in
\{0,\ldots, j\}$ we construct a vertex $y_{i,j,\alpha,\beta,\gamma,\delta}$,
whose intended role is the following: assuming $j$ is correct, there is a path
from the input of the gadget to $y_{i,j,\alpha,\beta,\gamma,\delta}$ if and
only if among the first $\gamma$ vertices of the $i$-th layer of $G_i$ there
are at least $\delta$ vertices which are reachable from $s^1$ but which do not
have an arc that leads to $v$ under the current assignment, where again $v$ is
the $(\alpha+1)$-th vertex of $V_1$. The input vertex of the gadget is
$y_{i,j,\alpha,\beta,0,0}$ and the output is $y_{i,j,\alpha,\beta,n,j}$. For
all $\gamma\in \{0,\ldots,n-1\}$ and $\delta\in \{0,\ldots, j\}$ we add a
non-annotated arc from $y_{i,j,\alpha,\beta,\gamma,\delta}$ to
$y_{i,j,\alpha,\beta,\gamma+1,\delta}$. The intended meaning is that if among
the first $\gamma$ vertices of the $i$-th layer there are $\delta$ satisfying
the property above, then there certainly are at least $\delta$ vertices
satisfying this property among the first $\gamma+1$ vertices. What remains is
to add a gadget so that we can transition from
$y_{i,j,\alpha,\beta,\gamma,\delta}$ to
$y_{i,j,\alpha,\beta,\gamma+1,\delta+1}$ if and only if the $(\gamma+1)$-th
vertex for the $i$-th layer, call it $u^i$, satisfies the property. If there is
a non-annotated arc from $u^i$ to $v^{i+1}$, we do nothing. If there is no such
arc or the arc is annotated, we need to check if $u^i$ is reachable from $s^1$.
We therefore add an $\mathrm{Adj}$ gadget checking this, and add a
non-annotated arc from $y_{i,j,\alpha,\beta,\gamma,\delta}$ to the input of
that gadget. From the output of that gadget we add an arc to
$y_{i,j,\alpha,\beta,\gamma+1,\delta+1}$ annotated as follows: if
$(u^i,v^{i+1})$ is not an arc of $G_1'$, then we add no annotation; otherwise
$(u^i,v^{i+1})$ is annotated with a literal $\ell$, and we use the annotation
$\neg \ell$ for the new arc. This completes the construction.  

Correctness now follows from the description of the gadgets. It is worth noting
that a vertex $u^i$ can be used as part of the certificate that $v^{i+1}$ is
not reachable only when the (possible) arc $(u^i,v^{i+1})$ cannot be used
either because it does not exist or because of its annotation. Since we have
added the same annotation negated to the output arc of the $\mathrm{NonAdj}$
gadget, this correctly checks if $u^i$ is a valid part of the certificate that
$v^{i+1}$ is not reachable. The rest of the proof follows along the lines of
the classical Immerman–Szelepcsényi theorem \cite{Immerman88}.  \end{proof}

\section{Circuit-SAT and Horn Backdoors}\label{sec:horn}

Our goal in this section is to prove the following theorem:

\begin{theorem}\label{thm:horn}

The following statements are equivalent:

\begin{enumerate}

\item (\circseth\ is false) There exists $\eps>0$ and an algorithm that takes as input a Boolean
circuit of size $s$ on $n$ inputs and decides if the circuit is satisfiable in
time $(2-\eps)^ns^{O(1)}$.

\item (\hornseth\ is false) There exists $\eps>0$ and an algorithm that takes as input a CNF formula
$\phi$ together with a strong Horn backdoor of size $b$ and decides if $\phi$
is satisfiable in time $(2-\eps)^b|\phi|^{O(1)}$.

\item (\wpseth\ is false, weak version) There exist $\eps>0, k_0>0$ and an
algorithm which for all $k>k_0$ takes as input a \emph{monotone} Boolean
circuit of size $s$ and decides if the circuit can be satisfied by setting at
most $k$ inputs to True in time $O(s^{(1-\eps)k})$.

\item (\wpseth\ is false, strong version) There exists $\eps>0$ and an
algorithm which takes as input a Boolean circuit of size $s$ on $n$ inputs and
an integer $k$ and decides if the circuit can be satisfied by setting exactly
$k$ of its inputs to True in time $n^{(1-\eps)k}s^{O(1)}$.

\end{enumerate}

\end{theorem}

Before we go on let us give some intuition. Similarly to \cref{sec:logpw}, it
will be helpful to recall some basic facts from computational complexity
theory. For the first two statements, we are dealing with the satisfiability
problem for circuits and \textsc{SAT} with a Horn backdoor.  Intuitively, it is
natural to expect these two problems to be equivalent because once we fix
either the input to a given circuit or the assignment to the Horn backdoor,
deciding if the circuit is satisfied or if the backdoor assignment can be
extended to a satisfying assignment are problems with the same classical
complexity, that is, they are P-complete \cite{0072413}. As in \cref{sec:logpw}
this intuition is not in itself sufficient to imply the equivalence, because
what we want is not to claim that once we have an assignment, checking it is of
the same complexity in both cases, but rather that if we can avoid checking all
assignments in one case, we can avoid it in both. Nevertheless, the classical
results on the P-completeness of these two problems are indeed the main
ingredient we need, and our proof mostly consists of verifying that we can
extend them to the new setting. We then go on to consider the prototypical
W[P]-complete problem of the third statement. It is natural to expect that
since the first and third statement refer to circuit satisfiability problems,
without any particular restriction on the circuits, it should be possible to
transform one to the other, and indeed we are able to do so using standard
techniques, similar to what we did in \cref{thm:ald}.

\begin{proof}[Proof of \cref{thm:horn}]

The proof is given in the four lemmas below which show the implications
$1\Rightarrow 2$ (\cref{lem:sathorn1}), $2\Rightarrow 1$ (\cref{lem:sathorn2}),
$3\Rightarrow 1$ (\cref{lem:sathorn3}), $1\Rightarrow 4$ (\cref{lem:sathorn4}).

We note that the implication $4\Rightarrow 3$ is easy to see. First, statement
3 talks about monotone circuits, while statement 4 about general circuits.
Second, if we have an algorithm running in time $n^{(1-\eps)k}s^{O(1)}$ we can
see that $n\le s$ (as the input gates are part of the circuit) so the running
time is at most $O(s^{(1-\eps)k+c})$ for some constant $c$. We can now set
$k_0=\frac{2c}{\eps}$ and we then have $\eps k> \frac{\eps k}{2} + c$ if
$k>k_0$, therefore the running time is at most $O(s^{(1-\frac{\eps}{2})k})$.
\end{proof}

\subsection{Circuit-SAT and Horn Backdoors}

\begin{lemma}\label{lem:sathorn1} Statement 1 of \cref{thm:horn} implies
statement 2.  \end{lemma}

\begin{proof}

Suppose that we have an algorithm solving \textsc{Circuit SAT} in time
$(2-\eps)^{n}s^{O(1)}$, for some fixed $\eps>0$. We are given a CNF formula
$\phi$ and a strong Horn backdoor of size $b$ of $\phi$. Our plan is to
construct a Boolean circuit $C$ satisfying the following:

\begin{enumerate}

\item $C$ is satisfiable if and only if $\phi$ is.

\item $C$ has $b$ inputs and can be constructed in time $|\phi|^{O(1)}$.

\end{enumerate}

Clearly, if $C$ can be constructed in time polynomial in $|\phi|$, then its
size is also polynomial in $|\phi|$, so the supposed algorithm for
\textsc{Circuit SAT} gives us the promised algorithm for \textsc{SAT} with a
Horn backdoor.

The high-level idea now is to implement a circuit which takes as input the
values of the $b$ backdoor variables of $\phi$  and then runs the standard
polynomial-time algorithm to check satisfiability of Horn formulas. Recall that
this algorithm starts with the all-False assignment to all (non-backdoor)
variables, and then iterates, at each step setting a variable to True if this
is forced. This will result in either a satisfying assignment or a
contradiction after at most $n$ iterations, where $n$ is the number of
(non-backdoor) variables.

Suppose that the variables of $\phi$ that belong to the backdoor are
$y_1,\ldots,y_b$ and the remaining variables are $x_1,\ldots,x_n$. Assume
$\phi$ has $m$ clauses, numbered $c_1,\ldots,c_m$.

We construct a circuit with the following gates:

\begin{enumerate}

\item $b$ input gates, identified with the variables $y_i, i\in[b]$.

\item For each non-backdoor variable $x_i, i\in[n]$, we construct $n+2$
disjunction ($\lor$) gates $x_{i,t}, t\in\{0,\ldots,n+1\}$. The intuitive
meaning is that $x_{i,t}$ should evaluate to True if the Horn satisfiability
algorithm would set variable $x_i$ to True after $t$ iterations of the main
loop. One of the inputs of each $x_{i,t}$ is the gate $x_{i,t-1}$, for all
$i,t\in[n+1]$. We specify the remaining inputs below.

\item For each clause $c_j, j\in[m]$, we construct a disjunction ($\lor$) gate
using the variables $x_{i,n+1}$ and the input gates as inputs. In particular,
if $x_i$ appears positive in $c_j$ we connect $x_{i,n+1}$ to the gate
representing $c_j$, and if it appears negative, we connect $x_{i,n+1}$ to a
negation ($\neg$) gate whose output is sent to the gate representing $c_j$, and
similarly for input gates.  We route the output of all gates representing
clauses to a $\land$ gate which is the output of the whole circuit.
Intuitively, the circuit contains a copy of $\phi$, with the $x_i$ variables
replaced by the gates $x_{i,n+1}$.

\end{enumerate}

To complete the construction we describe the gates we need to add so that the
$x_{i,t}$ gates play their intuitive role. For each variable $x_i, i\in[n]$,
for each $t\in[n+1]$ and each clause $c_j, j\in[m]$, such that $x_i$ appears
positive in $c_j$ we add a conjunction ($\land$) gate $c_{i,j,t}$ whose output
is sent to $x_{i,t}$. The inputs to $c_{i,j,t}$ are all gates $x_{i',t-1}$ such
that $x_{i'}$ appears (negative) in $c_j$, all gates $y_{i''}$ such that
$y_{i''}$ appears negative in $c_j$, and the negation of each gate $y_{i''}$,
such that $y_{i''}$ appears positive in $c_j$. Finally, set all $x_{i,0}$
gates, for $i\in[n]$ to False and simplify the circuit accordingly.  This
completes the construction, which can clearly be carried out in polynomial
time. We have described a circuit with unbounded fan-in, but it is easy to
transform it into a circuit of bounded fan-in without significantly increasing
its size, in polynomial time. Note also that we have used the fact that if
$x_i$ appears positive in $c_j$ every other non-backdoor variable $x_{i'}$ must
appear negative in $c_j$; this follows from the fact that if two variables
appear positive in a clause, at least one of them must be in the backdoor
\cite{NishimuraRS04}.

It is easy to see that if the circuit is satisfiable then so is the formula,
because we can extract an assignment from the values of the input gates and the
gates $x_{i,n+1}$. This assignment satisfies the formula since we have added a
copy of $\phi$ at the end of the circuit, with inputs the $y_i$ and $x_{i,n+1}$
gates. 

For the converse direction, we first make some easy observations: for each
$i\in[n]$ and $t\in[n+1]$ we have that if $x_{i,t-1}$ evaluates to True, then
$x_{i,t}$ evaluates to True, since $x_{i,t-1}$ is an input to $x_{i,t}$. As a
result, as $t$ increases, the set of gates $x_{i,t}$ which evaluate to True can
only monotonically increase. Furthermore, if for some $t\in[n+1]$ we have for
all $i\in[n]$ that $x_{i,t-1}$ and $x_{i,t}$ evaluate to the same value, then
for all $t'>t$ and for all $i$ we have $x_{i,t'}$ evaluates to the same value
as $x_{i,t}$. This is because the $c_{i,j,t}$ gates and the $c_{i,j,t+1}$ gates
will evaluate to the same values, forcing $x_{i,t+1}$ to have the same value as
$x_{i,t}$, and we can proceed by induction. It now follows that $x_{i,n}$
evaluates to the same value as $x_{i,n+1}$ for all $i\in[n]$. Indeed, there are
$n+2$ possible values of $t$ (time steps), at each step the set of True gates
can only monotonically increase, so we have at most $n+1$ distinct sets of True
gates. As a result, by pigeonhole principle, there exists $t$ such that
$x_{i,t}$ evaluates to the same value as $x_{i,t+1}$ for all $i\in [n]$, which
implies the same for $t=n$.

Suppose now that $\phi$ is satisfiable and fix a satisfying assignment. Using
the same assignment restricted to the backdoor variables as input to our
circuit we claim that we obtain a satisfying assignment to the circuit. To see
this, we claim that if for some $t\in[n+1]$ we have that $x_{i,t}$ evaluates to
True, then the satisfying assignment to $\phi$ must also set $x_i$ to True. For
the sake of contradiction, suppose that $t$ is minimum such that $x_{i,t}$
evaluates to True but the satisfying assignment sets $x_i$ to False.  Since
$x_{i,t}$ is True and $t$ is minimum, one of the $c_{i,j,t}$ gates must
evaluate to True.  But this means that the satisfying assignment sets all
literals of the clause $c_j$ to False, except $x_i$, because by the minimality
of $t$ we know that for each $x_{i',t-1}$ which evaluates to True, $x_{i'}$ is
True in the satisfying assignment.  Hence, the satisfying assignment must set
$x_i$ to True.

We can now conclude that if $x_{i,n+1}$ evaluates to True, then the satisfying
assignment to $\phi$ sets $x_i$ to True. Suppose then for contradiction that
the circuit outputs False, which implies that there is a gate added in the
third step, representing $c_j$, which evaluates to False. It must be the case
that $c_j$ has as input an $x_{i,n+1}$ which evaluates to False while $x_i$ is
set to True in the satisfying assignment, because in all other cases the inputs
to $c_j$ would be the same as those dictated by the satisfying assignment to
$\phi$ and $c_j$ would evaluate to True. Furthermore, $x_i$ must appear
positive in $c_j$ and is therefore the only non-backdoor variable which appears
positive in $c_j$. We now claim that the gate $c_{i,j,n+1}$ must evaluate to
True, reaching a contradiction.  Indeed, because $x_{i',n}$ and $x_{i',n+1}$
evaluate to the same value for all $i'\in[n]$, all non-backdoor variables of
$c_j$ except $x_i$ (which all appear negative in $c_j$) are set to True in
$x_{i,n+1}$, hence in $x_{i,n}$, and backdoor variables are shared between the
gates $c_j, c_{i,j,n+1}$, so all inputs to $c_{i,j,n+1}$ evaluate to True,
contradiction.  Hence, we have an assignment satisfying the circuit.
\end{proof}

\begin{lemma}\label{lem:sathorn2} Statement 2 of \cref{thm:horn} implies
statement 1.  \end{lemma}

\begin{proof}

Suppose we have an algorithm solving \textsc{SAT}, given a Horn backdoor of
size $b$, in time $(2-\eps)^b|\phi|^{O(1)}$. We are now given a circuit $C$
with $n$ inputs and size $s$. Our plan is to construct from $C$ a CNF formula
$\phi$ and a Horn backdoor such that:

\begin{itemize}

\item $\phi$ is satisfiable if and only if $C$ is.

\item The backdoor has size $n$ and $\phi$ can be constructed in time
$s^{O(1)}$.

\end{itemize}

Clearly, if $\phi$ can be constructed in time $s^{O(1)}$, it also has size
polynomial in $s$, so the supposed algorithm for \textsc{SAT} with a Horn
backdoor would give us an algorithm for \textsc{Circuit-SAT}.

The high-level idea is now to revisit the standard proof that \textsc{Horn-SAT}
is P-complete, which relies on a reduction from the \textsc{Circuit Value
Problem}. More precisely, assume that the given circuit is represented by a DAG
that contains only $\land,\lor$, and $\neg$ gates of maximum fan-in $2$. We
first transform the circuit (in polynomial time) so that $\neg$ gates only
appear directly connected to the inputs. Since the circuit is a DAG, we can
assume that the gates are topologically sorted.  We construct $\phi$ as
follows:

\begin{enumerate}

\item For each input gate of $C$ we construct a variable $x_i$, $i\in[n]$.
These variables will form the Horn backdoor.

\item For each input, $\land$, and $\lor$ gate $c$ of $C$, we construct two
variables: $c^p$ and $c^n$. Intuitively, in circuit assignments that set $c$ to
True, $c^p$ must be set to True and $c^n$ to False, while the opposite must
happen for circuit assignments that set $c$ to False. Thus, for each such gate
we add the clause $(\neg c^p\lor \neg c^n)$.

\item For each input gate $c$ corresponding to variable $x_i$ we add the
clauses $(x_i\leftrightarrow c^p)$ and $(x_i\leftrightarrow \neg c^n)$.

\item For each negation gate $c$, whose input corresponds to the variable
$x_i$, we add the clauses $(x_i\leftrightarrow \neg c^p)$ and
$(x_i\leftrightarrow c^n)$.

\item For each conjunction gate $c$, whose inputs are gates $c_1,c_2$ we add
the clauses $(\neg c_1^p \lor \neg c_2^p \lor c^p), (\neg c_1^n\lor c^n), (\neg
c_2^n\lor c^n), (\neg c^p\lor c_1^p), (\neg c^p\lor c_2^p)$ 

\item For each disjunction gate $c$, whose inputs are gates $c_1,c_2$ we add
the clauses $(\neg c_1^n \lor \neg c_2^n \lor c^n), (\neg c_1^p\lor c^p), (\neg
c_2^p\lor c^p), (\neg c^n\lor c_1^n), (\neg c^n\lor c_2^n)$ 

\item For the unique output gate $c$ of the circuit we add the unit clauses
$c^p$ and $\neg c^n$.

\end{enumerate}

The construction can be performed in polynomial time. Furthermore, all clauses
except those of the third and fourth steps are Horn clauses. The clauses of the
third and fourth steps all have size $2$ and contain a variable from the
backdoor, so we have a valid strong Horn backdoor of size $n$, as promised. It
is now not hard to see that if $C$ has a satisfying assignment, then we can
obtain a satisfying assignment to $\phi$ by giving to the variables of the
backdoor the values of the satisfying assignment to $C$ and then for each gate
$c$ of the circuit setting $c^p, c^n$ to True, False respectively if $c$ is set
to True and to False, True respectively otherwise.

For the converse direction, suppose we have a satisfying assignment to $\phi$.
We claim that using the assignment to the backdoor variables as input to $C$
satisfies the circuit. To see this, we prove by induction that for each gate
$c$, exactly one of $c^p, c^n$ must be set to True in the satisfying
assignment, and $c^p$ is set to True if and only if gate $c$ has value True in
the circuit $C$ for the given input assignment. Note that it is clear that at
most one of $c^p, c^n$ is set to True, due to the clauses of the second step.
Furthermore, for the base case we observe that input and negation gates satisfy
the property, due to the clauses of the third and fourth steps. Consider then a
gate $c$ with inputs $c_1,c_2$, which are assumed to satisfy the property. We
have four cases:

\begin{enumerate}

\item $c$ is a conjunction that has value True in $C$. Then $c_1,c_2$ must
evaluate to True in $C$, since these gates satisfy the property we have $c_1^p,
c_2^p$ set to True, so the clause $(\neg c_1^p\lor \neg c_2^p\lor c^p)$ ensures
that $c^p$ is True (and $c^n$ is False) as desired.

\item $c$ is a conjunction that has value False in $C$. Then one of $c_1,c_2$
must evaluate to False in $C$, say without loss of generality $c_1$, and since
these gates satisfy the property we have $c_1^n$ set to True. So, the clause
$(\neg c_1^n\lor c^n)$ ensures that $c^n$ is True (and $c^p$ is False) as
desired.

\item $c$ is a disjunction that has value True in $C$. Then one of $c_1,c_2$
must evaluate to True in $C$, say without loss of generality $c_1$, and since
these gates satisfy the property we have $c_1^p$ set to True. So, the clause
$(\neg c_1^p\lor c^p)$ ensures that $c^p$ is True (and $c^n$ is False) as
desired.

\item $c$ is a disjunction that has value False in $C$. Then $c_1,c_2$ must
evaluate to False in $C$, since these gates satisfy the property we have
$c_1^n, c_2^n$ set to True, so the clause $(\neg c_1^n\lor \neg c_2^n\lor c^n)$
ensures that $c^n$ is True (and $c^p$ is False) as desired.

\end{enumerate}

Therefore, $\phi$ is satisfiable if and only if $C$ is.  \end{proof}

\subsection{Circuit-SAT and Weighted Circuit-SAT}

\begin{lemma}\label{lem:sathorn3} Statement 3 of \cref{thm:horn} implies
statement 1.  \end{lemma}

\begin{proof}

Suppose we have an algorithm which takes as input a monotone boolean circuit
$C'$ of size $S$ and $N$ inputs and an integer $k$ and decides if $C'$ has a
satisfying assignment setting $k$ inputs to True in time $O(S^{(1-\eps)k})$,
for some $\eps>0$, assuming $k>k_0$ for some fixed $k_0$. We will use this
supposed algorithm to solve \textsc{Circuit-SAT}. We are given a circuit $C$
with $n$ inputs and size $s$. We will assume that $s<2^{\frac{\eps}{2k_0}n}$,
because otherwise $2^n\le s^{\frac{2k_0}{\eps}} = s^{O(1)}$, so we immediately
have a polynomial time algorithm for such cases by trying out all input
assignments.

We construct a monotone circuit $C'$ and an integer $k=k_0$ with the following
properties:

\begin{enumerate}

\item $C'$ has a satisfying assignment setting $k$ of its inputs to True if and
only if $C$ is satisfiable.

\item $C'$ can be constructed in time $O(2^{n/k}s)$.

\item $C'$ has $N=O(2^{n/k})$ inputs and size
$S=O(Ns)=O(2^{\frac{n}{k}+\frac{\eps n}{2k}})$.

\end{enumerate}

If we achieve the above, then we obtain the lemma as follows: we can decide
satisfiability of $C$ by running the reduction and then the supposed algorithm
on $C'$. The running time is dominated by the second step which takes time
$O(S^{(1-\eps)k}) = O(2^{(1-\eps)(1+\frac{\eps}{2})n}) = O(2^{(1-\eps')n})$ for
some appropriate $\eps'$, as desired.

Let us then describe the construction of $C'$ in which we start from $C$ and
add some gadgets. We assume that after running standard transformations $C$
only contains negation gates at the first layer, that is, every input of a
negation is an input to $C$.  Partition the inputs of $C$ arbitrarily into $k$
groups of (almost) equal size.  Each group contains at most
$\lceil\frac{n}{k}\rceil$ inputs of $C$.  Consider a group $I_i$ of input gates
of $C$, for $i\in[k]$ and enumerate all $2^{|I_i|} = O(2^{n/k})$ possible
assignments to the inputs of $I_i$. For each such assignment $\sigma$ we
construct a new input gate $g_{i,\sigma}$. We convert the gates of $I_i$ from
input gates to disjunction ($\lor$) gates and for each assignment $\sigma$ we
connect $g_{i,\sigma}$ to all gates of $I_i$ such that $\sigma$ assigns them
value True. Similarly, for each negation gate whose input is in $I_i$ we
convert it into a disjunction gate and for each $\sigma$ we connect
$g_{i,\sigma}$ to all gates which were previously negations such that $\sigma$
was assigning to their input the value False. Since we have converted all
negations to disjunctions the circuit is now monotone.

At this point it is not hard to see that the new circuit has $N=O(k2^{n/k}) =
O(2^{n/k})$ inputs (as $k$ is a fixed constant), its size is at most $O(Ns)$
(in the worst case we connected each new input to each old gate) and that a
satisfying assignment to $C$ exists if and only if there exists a satisfying
assignment to $C'$ that sets exactly one new input gate from each of the $k$
groups to True and all others to False.

We now add some gadgets to ensure that a satisfying assignment of weight $k$
must set one input from each group to True. For each $i\in[k]$ add a
disjunction ($\lor$) gate $t_i$ and connect all $g_{i,\sigma}$ to $t_i$. Add a
conjunction gate that will be the new output of the circuit and feed it as
input all $t_i$, for $i\in[k]$ as well as the output gate of $C$. This
completes the construction. The total size is still $S=O(Ns)$ and the
construction can be performed in linear time in the size of the output. We can
also transform the circuit into an equivalent bounded fan-in circuit without
increasing the size by more than a constant factor. With the extra gates we
added in the last step we ensure that at least one input gate from each group
is set to True, so since we are looking for an assignment of weight exactly
$k$, exactly one input is set to True in each group. \end{proof}

\begin{lemma}\label{lem:sathorn4} Statement 1 of \cref{thm:horn} implies
statement 4. \end{lemma}

\begin{proof}

Suppose there is an algorithm that takes as input a boolean circuit $C'$ of
size $S$ with $N$ inputs and decides if $C'$ is satisfiable in time
$2^{(1-\eps)N}S^c$, for some constants $\eps>0, c>0$. We are given a circuit
$C$ with $n$ inputs and size $s$ as well as an integer $k$ and we want to
decide if $C$ has a satisfying assignment of weight exactly $k$ in time
$n^{(1-\eps')k}s^{O(1)}$.  We will do this by constructing a circuit $C'$ with
the following properties:

\begin{enumerate}

\item $C'$ is satisfiable if and only if $C$ has a satisfying assignment of
weight $k$.

\item $C'$ has at most $N=k\log n+k$ input gates and can be constructed in time
$O(sk^2n^2)$.

\end{enumerate}

If we achieve the above we obtain the desired algorithm because we can produce
$C'$ and run the supposed \textsc{Circuit-SAT} algorithm on it. The running
time is dominated by the second step which takes $2^{(1-\eps)N}S^c$. We have
$2^{(1-\eps)N} \le n^{(1-\eps)k}2^k$ and $S^c = O(s^c (kn)^{2c} )$. Observe now
that we can assume that $n$ and $k$ are sufficiently large. In particular, if
$k<8c/\eps$, then since without loss of generality $s>n$, the trivial $O(n^ks)$
time algorithm which checks all assignments of weight $k$ would run in time
$s^{O(1)}$. But if $n,k$ are sufficiently large, then
$2^k(kn)^{2c}<n^{\frac{\eps k}{2}}$ and the total running time is at most
$n^{(1-\frac{\eps}{2})k}s^{O(1)}$ as desired.

We construct $C'$ by adding some gadgets to $C$.  First, number the input gates
of $C$ as $c_1,\ldots,c_n$. We transform them into disjunction gates. For each
$i\in[n]$ we construct $k$ new gates $c_i^1,\ldots,c_i^k$ and feed their output
into $c_i$. The intuitive meaning is that $c_i^j$ will be True if $c_i$ is the
$j$-th input set to True in a satisfying assignment of weight $k$. Each
$c_i^j$, for $i\in[n], j\in[k]$ is a conjunction gate.

We now construct $k\lceil \log n \rceil \le k\log n +k$ new input gates,
partitioned into $k$ groups $I_1,\ldots, I_k$, each containing $\lceil \log
n\rceil$ input gates.  The intuitive meaning is that the assignment to the
gates of $I_j$ should encode the index of the $j$-th input of $C$ set to True
in a satisfying assignment. For each $j\in[k]$ and each assignment $\sigma$ to
the input gates of $I_j$ we interpret $\sigma$ as a binary number in $[n]$ and
suppose that this number is $i$. We then connect the input gates of $I_j$ to
$c_i^j$ via appropriate negations such that $c_i^j$ evaluates to True if and
only if the gates of $I_j$ are given assignment $\sigma$.

At this point it is not hard to see that $C$ is satisfiable with an assignment
of weight $k$ if and only if $C'$ is satisfiable via an assignment which for
each distinct $j,j'\in[k]$ gives distinct assignments to $I_j, I_{j'}$. What
remains is to add some gadgets so that $I_j, I_{j'}$ cannot have the same
assignment, so that satisfying assignments to $C'$ encode satisfying
assignments of weight exactly $k$ to $C$. This is easily achieved by adding for
each pair of distinct $j,j'\in[k]$ a disjunction gate, which takes its input
from $2\lceil \log n\rceil$ conjunction gates, where each conjunction gate
checks if the assignments to $I_j, I_{j'}$ differ in a corresponding bit by
giving False to the bit of $I_j$ and True to the bit of $I_{j'}$ or vice-versa. 

Let us analyze the properties of this construction. The number of inputs of
$C'$ is at most $k\lceil\log n\rceil \le k\log n + k$. The number of wires we
added is at most $O(kn\log n)$ for the wires incident on $c_i^j$ gates, and
$O(k^2\log n)$ for the gadget that ensures we have an assignment of weight $k$.
The construction can be performed in linear time in the size of the output and
the total size is at most $O(sk^2n^2)$.  \end{proof}

\section{Applications}\label{sec:applications}

In this section we list five characteristic example of our equivalence classes.
In each case we highlight a problem and a lower bound on its time complexity
such that breaking the bound is equivalent to falsifying the hypotheses of the
corresponding class.

\subsection{SETH-equivalence for Coloring}

\begin{theorem}\label{thm:color1} For all $q\ge 3$, there is an $\eps>0$ and an
algorithm solving $q$-\textsc{Coloring} in time $(q-\eps)^mn^{O(1)}$, where $m$
is the size of a given modulator to constant tree-depth, if and only if the
SETH is false. The same holds if $m$ is the size of a given \sdh. \end{theorem}

\begin{proof}

It was already established in \cite{EsmerFMR24} that if the SETH is true, then
there is no algorithm for $q$-\textsc{Coloring} with the stated running time,
even for parameter the size of a given \sdh. Therefore, if the SETH is true, no
such algorithm exists for parameter size of modulator to constant tree-depth.

We establish that if the SETH is false, then we do get such a
$q$-\textsc{Coloring} algorithm for parameter modulator to constant tree-depth.
We are given an $n$-vertex graph $G$ and a modulator $M$ of size $m$ so that
$G-M$ has tree-depth $c$. We will use \cref{thm:seth} and reduce this instance
to an instance of $k$-\textsc{SAT} with a modulator to constant tree-depth.
According to \cref{thm:seth}, if the SETH is false, we can solve such an
instance is time $2^{(1-\eps)m'}|\phi|^{O(1)}$, where $m'$ is the size of the
new modulator, for some $\eps$. In particular, we will make sure that $m'\le
(1+\frac{\eps}{2})m\log q+O(1)$, and that $\phi$ has size polynomial in $n$.
Then, the running time of the whole procedure will be at most
$2^{(1-\eps)(1+\frac{\eps}{2})m\log q}n^{O(1)} = q^{(1-\eps')m}n^{O(1)}$ for
some appropriate $\eps'$.

To begin our construction, we consider every vertex of $G-M$. For each such
$v$, we define $q$ variables $x_v^1,\ldots, x_v^q$, indicating informally that
$v$ takes color $i\in[q]$. We add the clause $(x_v^1\lor\ldots\lor x_v^q)$. For
each $uv\in E$ with $u,v\not\in M$ we add for each $i\in[q]$ the clause $(\neg
x_v^i \lor \neg x_u^i)$. It is not hard to see that the graph we have
constructed so far has tree-depth at most $cq$. All the other variables we will
construct will belong in the new modulator.

Consider now the vertices of $M$ and for some integer $\gamma$ to be defined
later partition them into groups of size at most $\gamma$ as equitably as
possible. We have then at most $\tau=\lceil\frac{m}{\gamma}\rceil \le
\frac{m}{\gamma}+1$ groups, $M_1,\ldots, M_{\tau}$. For some integer $\rho$,
also to be defined later, we replace each $M_i$ with a set of $\rho$ new
variables $Y_i$. We select the integers $\gamma,\rho$ so that we have the
following:

\[ 2^{(1-\frac{\eps}{2})\rho} \le q^\gamma \le 2^\rho \]

Since $q^\gamma\le 2^\rho$, we can map injectively all proper colorings of
$G[M_i]$ to assignments of $Y_i$. For each assignment that is not an image of a
proper coloring we add a clause ensuring that this assignment cannot be used in
a satisfying assignment. Furthermore, for each $uv\in E$ such that $u,v$ belong
in distinct sets $M_i,M_j$ we consider all pairs of colorings of $M_i,M_j$ that
set $u,v$ to the same color. For each such pair, we add a clause ensuring that
$Y_i,Y_j$ cannot simultaneously take the assignments which are the images of
these colorings. Finally, for each $u\in M_i$ and $v\not\in M$ such that $uv\in
E$, for each color $j\in [q]$ we add for each coloring of $M_i$ that colors $u$
with $j$ a clause that contains $\neg x_v^j$ and literals which are all
falsified if we select for $Y_i$ the image of this coloring of $M_i$.

This completes the construction and it is not hard to see that there is a
one-to-one correspondence between proper colorings of $G$ and satisfying
assignments of the CNF formula we have constructed. To bound the size of the
new modulator we observe that we have $\frac{m}{\gamma}+1$ groups of $\rho$
variables, but $(1-\frac{\eps}{2})\rho \le \gamma \log q$, therefore, $\rho\le
(1+\frac{\eps}{2})\gamma\log q$, so the total number of variables of the
modulator is at most $(\frac{m}{\gamma}+1)\rho\le (1+\frac{\eps}{2})m\log
q+O(1)$.

Finally, to argue that we can select appropriate $\gamma, \rho$, we observe
that it is sufficient that $2^{\rho} > 2^{(1-\frac{\eps}{2})\rho}q$, because
then there is an integer power of $q$ between $2^{\rho}$ and
$2^{(1-\frac{\eps}{2})\rho}$. This gives $\rho > \frac{4\log q}{\eps}$, so
selecting $\rho$ to be any integer at least this large also allows us to select
an appropriate $\gamma$.  \end{proof}

\subsection{MaxSAT-equivalence for Max Cut}

\begin{theorem} There is an $\eps>0$ and an algorithm which takes as input an
$n$-vertex graph $G$ and a \sdh\ of size $h$ and computes a maximum cut of $G$
in time $(2-\eps)^hn^{O(1)}$ if and only if the \maxsatseth\ is false.
\end{theorem}

\begin{proof}

One direction is easy: given a \textsc{Max Cut} instance $G=(V,E)$, we can
construct a CNF formula by constructing a variable $x_v$ for each vertex $v$
and for each edge $uv$ the clauses $(x_u\lor x_v)$ and $(\neg x_u\lor \neg
x_v)$. The primal graph of this formula is simply $G$, so the \sdh\ is
preserved. It is now not hard to see that the maximum number of clauses that it
is possible to satisfy in the new instance is exactly
$|E|+\textrm{max-cut}(G)$, so an algorithm as described in the second statement
of \cref{thm:maxsat} gives the desired algorithm for \textsc{Max Cut}.

For the converse direction, we are given a $k$-CNF formula $\phi$ with $n$
variables and $m$ clauses and a target value $t$ and are asked if we can
satisfy at least $t$ clauses. We reuse the reduction of \cite{LokshtanovMS18},
except with smaller weights, so that we obtain a \sdh. Let us give the relevant
details. We will first describe a reduction to edge-weighted \textsc{Max Cut}
and then remove the weights.

For each variable of $\phi$, $x_1,\ldots,x_n$ we construct a vertex with the
same name. We also construct a vertex $x_0$. For each clause $C_j$ of $\phi$ we
construct a cycle of length $4k+1$, one of whose vertices is $x_0$, while the
others are new vertices. Each edge of this cycle has weight $8k$. Number the
vertices of the cycle $p_{j,1},\ldots,p_{j,4k}$, starting from a neighbor of
$x_0$. Suppose that $C_j$ contains literals $\ell_1,\ell_2,\ldots,\ell_{k'}$
for $k'\le k$. For each $i\in[k']$ if $\ell_i$ is a positive appearance of a
variable $x_r$, then we connect $p_{j,4i+2}, p_{j,4i+3}$ to $x_r$, otherwise if
it is a negative appearance of $x_r$, then we connect $p_{j,4i+1}, p_{j,4i+2}$
to $x_r$. These edges have weight $1$. Furthermore, we attach $k-k'$ leaves to
$p_{j,1}$ with edges of weight $1$. We now claim that there is an assignment
satisfying $t$ clauses if and only if there is a cut of weight at least
$32k^2m+km+t$.

For the forward direction, if there is such an assignment, place all the
vertices representing False variables on the same side as $x_0$ in the cut and
all the vertices representing True variables on the other side. For each
unsatisfied clause, pick a bipartition of the corresponding odd cycle so that
$4k$ of the $4k+1$ edges of the cycle are cut and the uncut edge is incident on
$x_0$.  This contributes $32k^2$ from the edges of the cycle.  We also obtain
at least $k'$ from the edges connecting the cycle to the variable vertices and
$k-k'$ from the attached leaves by placing each leaf to the opposite side of
its neighbor, so we get $32k^2+k$ for each unsatisfied clause. For each
satisfied clause $C_j$, let $\ell_i$ be the first literal that is set to True
by the assignment.  For $i'<i$ we set $p_{j,4i'+2}, p_{j,4i'+4}$ on the same
side as $x_0$ and $p_{j,4i'+1}, p_{j,4i'+3}$ on the other side; while for
$i'>i$ we use the opposite partition. We now have that the vertex preceding
$p_{j,4i+1}$ and the vertex after $p_{j,4i+4}$ in the cycle are both on the
side of $x_0$. For the four remaining vertices, if $\ell_i$ is a variable that
appears positive, we set $p_{j,4i+2},p_{j,4i+3}$ on the same side as $x_0$,
cutting both edges connecting these vertices to the corresponding variable
vertex, and $p_{j,4i+1},p_{j,4i+4}$ on the other side; while if $\ell_i$ is a
variable that appears negative, we set $p_{j,4i+1}, p_{j,4i+2}, p_{j,4i+4}$ on
the side opposite $x_0$ and $p_{j,4i+3}$ on the side of $x_0$, again cutting
both edges to the corresponding variable vertex. We now observe that for a
satisfied clause we have cut $32k^2+k+1$ edges, obtaining the bound on the
total size of the cut.

For the converse direction we observe that an optimal cut must cut $4k$ of the
$4k+1$ edges of the odd cycle constructed for each clause. Indeed, it is not
possible to cut more (since the cycle is odd), while if we leave two edges
uncut, we may replace the solution by a solution that cuts $4k$ of the edges.
This will add $8k$ to the total weight, more than making up for the at most
$2k$ lost edges connecting the cycle to the rest of the graph. Thus, for each
cycle we obtain $32k^2$ from the edges inside the cycle. Furthermore, we can
obtain at most $k+1$ from the remaining edges incident on the cycle. To see
this, note that we have added $k'$ pairs of edges (one pair for each literal of
the clause) which are adjacent to consecutive vertices in the cycle. If we cut
both edges in a pair, an edge of the cycle is uncut, so there is at most one
pair of cut edges and from each other pair we cut at most one edge. Taking into
account the $k-k'$ incident on leaves (which are always cut), we have that each
cycle contributes at most $32k^2+k+1$ to the cut. Hence, there must be at least
$t$ cycles contributing exactly this much, which we claim correspond to
satisfied clauses by the assignment we extract by setting to False variables
whose corresponding vertices are on the same side as $x_0$. Indeed, take a
cycle representing $C_j$ and suppose that $C_j$ is not satisfied by this
assignment. However, there exists a literal $\ell_i$ in $C_j$ such that we cut
both edges added for $\ell_i$. If $\ell_i$ is positive, then
$p_{j,4i+2},p_{j,4i+3}$ have a common neighbor. If $C_j$ is not satisfied, this
neighbor must be on the same side as $x_0$, so $p_{j,4i+2},p_{j,4i+3}$ are on
the other side. It is now not hard to see via a parity argument that two edges
of the cycle must be uncut, contradiction. The reasoning is similar if $\ell_i$
is negative, as then $p_{j,4i+1},p_{j,4i+2}$ must have a common neighbor, which
must be on the side not containing $x_0$ (otherwise $C_j$ is satisfied), so
$p_{j,4i+1}, p_{j,4i+2}$ are on the same side as $x_0$, and again two edges are
uncut. Hence, a cut of the specified size leads to an assignment satisfying at
least $t$ clauses.

In order to remove weights, for each edge $uv$ of weight $w$ we construct $w$
parallel paths of length $3$ between $uv$ and increase the target by $2w$. It
is not hard to see that any decision on $u,v$ can be extended in a way that
cuts at least $2w$ of the new edges, and it is possible to cut all $3w$ new
edges if and only if $u,v$ are on different sides.

After the transformation we have that the variable vertices plus $x_0$ form a
hub of size $n+1$ such that its removal leaves components with $O(k^2)$
vertices. Furthermore, each component is connected to at most $k$ vertices of
the hub. We therefore have a \sdh, for $\sigma=O(k^2)$ and $\delta=k$.
\end{proof}

\subsection{Coloring and Pathwidth Modulators}

\begin{theorem} For all $q\ge3$ we have the following. There is an $\eps>0$
such that for all $c>0$ there is an algorithm that takes as input an $n$-vertex
graph $G$ and a modulator $M$ of size $m$ such that $G-M$ has pathwidth at most
$c$ and decides if $G$ is $q$-colorable in time $(q-\eps)^mn^{O(1)}$ if and
only if the \aldseth\ is false. The same is true if we replace ``pathwidth
$c$'' with ``tree-depth $c \log n$'' in the preceding statement. \end{theorem}

\begin{proof}

Clearly, an algorithm parameterized by modulator to tree-depth $c\log n$
implies the existence of an algorithm parameterized by modulator to pathwidth
$c$. We show two statements: if the \aldseth\ is false, then there is an
algorithm for parameter modulator to tree-depth $c\log n$; and if there is an
algorithm for parameter modulator to pathwidth $c$, then the \aldseth\ is
false.

The first claim readily follows from the proof of \cref{thm:color1}. Recall
that in that proof we showed how to transform an instance $G$ of
$q$-\textsc{Coloring} with a modulator $M$ to tree-depth $c$ to an instance of
\textsc{SAT} with a modulator of size essentially $|M|\log q$, while ensuring
that in the new instance the removal of the modulator leaves a graph of
tree-depth at most $cq$. We run the same reduction, except now the tree-depth
is at most $cq\log n$. One of the statements of \cref{thm:ald} states that we
can avoid brute-force for \textsc{SAT} parameterized by a modulator to
logarithmic tree-depth if and only if the \aldseth\ is false.

For the second claim we need to reduce \textsc{SAT} with parameter modulator to
pathwidth $c$ to $q$-\textsc{Coloring}. Suppose for some $q\ge 3$ there exists
$\eps>0$ such that for all $c'$ there is a $q$-\textsc{Coloring} algorithm
which takes as input a modulator to pathwidth $c'$ of size $m'$ and solves
$q$-\textsc{Coloring} in time $q^{(1-\eps)m'}n^{O(1)}$.

We are given a CNF formula $\phi$ and a modulator $M$ of size $m$ such that
$\phi-M$ has pathwidth $c$. Using \cref{obs:arity} we can assume that $\phi$ is
a 3-CNF formula. Select integers $\gamma,\rho$ such that we have the following:

\[ q^{(1-\frac{\eps}{2})\gamma} < 2^\rho < q^\gamma\]

For such integers to exist it is sufficient to have $q^\gamma >
2q^{(1-\frac{\eps}{2})\gamma}$, which gives $\gamma>\frac{2}{\eps \log q}$.
Selecting $\gamma = \lceil \frac{2}{\eps\log q} \rceil$ then guarantees that an
appropriate integer $\rho$ exists.

We will first reduce our 3-CNF formula to an instance of \textsc{CSP} with
alphabet $[q]$ and arity at most $3\gamma$. In this problem we are given a list
of variables which take values in $[q]$ and a list of constraints, each of
which involves at most $3\gamma$ variables. Each constraint lists the
assignment combinations to the involved variables which are acceptable. The
goal is to select an assignment so that we have selected an acceptable
combination in each constraint.  Observe that $q$-\textsc{Coloring} is a
special case of this problem where the arity is $2$ and all constraints state
that the endpoints of an edge must receive distinct colors.

We now partition $M$ into groups of variables of size $\rho$ as equitably as
possible, constructing at most $\frac{m}{\rho}+1$ groups. For each group we
construct in our \textsc{CSP} instance a group of $\gamma$ variables. We
construct an injective mapping from boolean assignments to a group of variables
of the original instance to an $q$-ary assignment to the corresponding group of
variables of the new instance. Since $2^\rho<q^\gamma$, this is always
possible. For each variable of $\phi$ outside of $M$ we construct a variable
whose acceptable values are only $\{1,2\}$ (representing True, False). For each
clause $C$ of $\phi$ we construct a constraint with maximum arity $3\gamma$: if
a variable of $C$ is part of a group $M_i$ of variables of $M$, we involve all
the $\gamma$ variables representing this group in the constraint; otherwise we
involve the single variable that represents the boolean variable outside of $M$
appearing in $C$. We now list all combinations of assignments to the at most
$3\gamma$ involved variables and we set as unacceptable those combinations
which are not an image of boolean assignments that satisfy the clause. This
completes the construction and it is not hard to see that the answer is
preserved. Furthermore, if we consider the primal graph of the \textsc{CSP}
instance, the graph is unchanged once we remove the (new) modulator. The new
modulator contains at most $(\frac{m}{\rho}+1)\gamma$ variables. But
$(1-\frac{\eps}{2})\gamma\log q <\rho$ so the size of the new modulator is at
most $m'\le (1-\frac{\eps}{2})\frac{m}{\log q} + O(1)$. If we had an algorithm
solving a \textsc{CSP} instance $\psi$ with bounded arity and alphabet $[q]$ in
time $q^{(1-\eps)m'}|\psi|^{O(1)}$, where $m'$ is the modulator to constant
pathwidth, this would give the desired algorithm for deciding our original CNF.
We will therefore show how to obtain this from the supposed
$q$-\textsc{Coloring} algorithm.

We first reduce the arity of the \textsc{CSP} instance to $2$. For each
constraint $C$ involving more than $2$ and at most $3\gamma$ variables, let
$\mathcal{S}_C$ be the set of acceptable assignments, which has cardinality at
most $q^{3\gamma}=O(1)$. We construct $2|\mathcal{S}_C|+1$ new variables and
connect them in an odd cycle with constraints of arity $2$, ensuring that all
the variables which are involved in a constraint take distinct values, that
$|\mathcal{S}_C|$ variables, which are pairwise non-adjacent, have possible
values $\{1,2,3\}$, and that all the others have possible values $\{1,2\}$. The
intuition here is that we have formed an odd cycle and we can use colors
$\{1,2,3\}$ to color it. Since the cycle is odd, we must use color $3$
somewhere, and if we use color $3$ once we can 2-color the rest. The idea then
is to associate each of the $|\mathcal{S}_C|$ variables that have three
possible values with a satisfying assignment to $C$ and add constraints that
ensure that if this variable takes color $3$ then the variables of $C$ must
agree with the corresponding assignment. For this, for each of the
$|\mathcal{S}_C|$ variables of the cycle that represents some assignment
$\sigma$ we add constraints of arity $2$ involving this variable and each
variable of $C$. The constraints state that if the variable of the cycle has
value $1$ or $2$, then everything is acceptable, otherwise the other variable
must agree with $\sigma$. We do this exhaustively and it is not hard to see
that the answer is preserved. We have to argue that the pathwidth of the primal
graph (after removing the modulator) is almost preserved. Indeed, for each
constraint $C$, find a bag that contains all its variables (outside of $M$),
and after this bag add a sequence of copies of it, where we add a path
decomposition of the cycle we have constructed. Repeating this will not add to
the pathwidth more than a small additive constant.

Finally, we want to reduce the \textsc{CSP} of arity $2$ and alphabet $q$ to
$q$-\textsc{Coloring}. We will reduce to \textsc{List Coloring} where all lists
are subsets of $[q]$. Adding a clique of size $q$ and appropriately connecting
to the graph to simulate the lists will then complete the reduction. To reduce
to \textsc{List Coloring} we can use the weak edge gadgets of \cite{Lampis20}.
For each variable of the \textsc{CSP} instance we construct a vertex. For each
constraint over $u,v$ we consider every assignment to $u,v$ that is ruled out
by this constraint. A weak edge is a gadget that we can add between $u,v$ that
will rule out the corresponding coloring, but will always be colorable if we
pick any other color combination for $u,v$. In fact, a weak edge is implemented
using a path of length $3$ between $u,v$ (with appropriate lists) so using weak
edges instead of edges does not affect the pathwidth of the graph by more than
an additive constant. By adding weak edges for each forbidden combination of
assignments we obtain an equivalent $q$-\textsc{Coloring} instance. The size of
the modulator has been preserved and we obtain the theorem.  \end{proof}

\subsection{k-Neighborhood Cut}

In this section we consider the following problem:

\begin{definition} In the \kncut\ problem we are given as input a DAG $G$ with
two designated vertices $s,t$. The question is whether it is possible to select
$k$ vertices of $G$, other than $s,t$, such that deleting all the out-neighbors
of the selected vertices results in a graph where no path from $s$ to $t$
exists.  \end{definition}

Before we proceed, let us note that it is trivial that \kncut\ admits an
algorithm running in time $n^{k+O(1)}$: simply try out all solutions and for
each solution check if there is a path from $s$ to $t$ avoiding the
out-neighborhoods of selected vertices. Furthermore, it is not hard to see that
the problem is W[2]-hard parameterized by $k$. Indeed, take an instance of
$k$-\textsc{Dominating Set} $G=(V,E)$ and construct a DAG that consists of two
copies of $V$, call them $V_1,V_2$, such that there is an arc from each $u\in
V_1$ to each $v\in V_2$ if $u=v$ or $uv\in E$. Add a vertex $s$ with arcs to
all of $V_2$ and a vertex $t$ with arcs from all of $V_2$. Clearly, selecting a
dominating set in $V_1$ is necessary and sufficient to eliminate all paths from
$s$ to $t$. From the results of Patrascu and Williams \cite{PatrascuW10}, the
above reduction implies that \kncut\ cannot be solved in time $n^{k-\eps}$,
assuming the SETH.

In this section we show the same lower bound but assuming the weaker assumption
that \textsc{SAT} cannot be solved faster than brute-force for instances with a
\textsc{2-SAT} backdoor. More interestingly, we show that this is precisely the
right assumption for this problem, in the sense that a better algorithm can be
obtained if and only if the assumption is false.  

\begin{theorem} There exist $\eps>0, k\ge 2$ such that \kncut\ can be solved in
time $n^{k-\eps}$ if and only if the \twosatseth\ is false.  \end{theorem}

\begin{proof}

We will show reductions in both directions from \nonreach\ which was
established to be equivalent to \twosatseth\ in \cref{sec:logpw}. As a
reminder, in \nonreach\ we are given a DAG with its arcs annotated from a set
of $m$ boolean variables. We are asked to select an assignment to the variables
so that in the DAG resulting from keeping only non-annotated arcs or arcs whose
annotation is set to True by the assignment there is no $s\to t$ path, for two
given vertices $s,t$. We have shown that \twosatseth\ is equivalent to assuming
that it is impossible to solve \nonreach\ without considering essentially all
assignments to the $m$ boolean variables.

First, suppose we have an instance $G=(V,E)$ of \nonreach\ on a DAG with $n$
vertices and a set $M$ of $m$ boolean variables. We construct a DAG $G'$ as an
instance of \kncut\ as follows. We begin with the same DAG $G$. We then
partition the variables into $k$ sets, call them $M_1,\ldots,M_k$, each of size
at most $\frac{m}{k}+1$. For each set $M_i$ we enumerate all assignments
$\sigma$, and for each assignment we construct a vertex $x_{i,\sigma}$. We
construct for each $i\in[k]$ a vertex $y_i$ and add arcs $sy_i$, $y_it$ as well
as all arcs $x_{i,\sigma}y_i$. For each arc $e$ of $G$ that is annotated with a
literal from a variable of $M_i$, we subdivide the arc and call the new vertex
$z_e$. We add an arc to $z_e$ from each $x_{i,\sigma}$ such that $\sigma$ sets
the literal of $e$ to False. This completes the construction, which can be
carried out in time at most $2^{m/2}n^{O(1)}$ for all $k\ge 2$. The new graph
has $N$ vertices with $N\le 2^{m/k}+n^{O(1)}$. 

We now observe that the reduction above preserves the solution. In particular,
because of the vertices $y_i$ we are forced to select exactly one
$x_{i,\sigma}$ for each $i\in[k]$, so there is a one-to-one correspondence
between assignments to $M$ and potential solutions in the new instance.
Furthermore, an annotated arc can be used if and only if we have selected the
encoding of an assignment that sets the annotation to True.

Suppose that there is an algorithm for \kncut\ running in time $N^{k-\eps}$ for
some fixed $k\ge 2, \eps>0$. We can now solve \nonreach\ by running the
reduction described above and then solving \kncut. The running time will be
dominated by the second part, which would take $N^{k-\eps} =
2^{\frac{m}{k}(k-\eps)}n^{O(1)} = 2^{(1-\frac{\eps}{k})m}n^{O(1)}$, which would
falsify the \twosatseth\ by the results of \cref{sec:logpw}.

For the converse direction, suppose we are given an instance $G=(V,E)$ of
\kncut\ on $n$ vertices.  We will reduce it to \nonreach\ as follows. We start
with the same DAG and construct a set $M$ of $k\log n$ variables (assume
without loss of generality that $n$ is a power of $2$, otherwise adding dummy
vertices to $G$). Intuitively, the variables of $M$ are partitioned into $k$
sets $M_1,\ldots,M_k$, where the assignment of $M_i$ is meant to encode the
selection of the $i$-th vertex of the solution to \kncut.

We now edit $G$ as follows. We replace every vertex $u\in V$ with two vertices
$u_{in}, u_{out}$. Each arc $uv$ is replaced with the arc $u_{out},v_{in}$. We
set $s_{out}$ as the start vertex and $t_{in}$ as the destination vertex. For
each $u\in V$ we further add $k+1$ vertices $u_0,u_1,\ldots,u_{k}$. The idea is
that we should be able to go from $u_{in}$ to $u_{i}$ if and only if the first
$i$ vertices of the solution do not have any arc to $u$ in the original
instance. We add an arc $u_{in}u_0$ and an arc $u_ku_{out}$.

Let us now describe the gadget we add between $u_{i-1}$ and $u_i$ for all
$i\in[k]$. Intuitively, we should be able to go from $u_{i-1}$ to $u_i$ if and
only if the assignment to $M_i$ does not encode a vertex that dominates $u$.
For each $v\in V$ such that $vu\not\in E$ and $v\neq u$ we therefore construct
a path from $u_{i-1}$ to $u_i$ of length $\log n$ going through new vertices.
The arcs of the path are annotated with literals from $M_i$ ensuring that we
can only follow all the arcs if the assignment to $M_i$ encodes vertex $v$.

This completes the construction. Observe that each arc of $G$ has been replaced
by $k$ collections of at most $n$ parallel paths, each of length $\log n$, so
the new graph has at most $n^5$ vertices. We observe that the answer is
preserved. Indeed, a solution $S$ of \kncut\ can be translated to an assignment
to $M$ in the natural way and we can see that if $u\in S$ dominates $v$, then
it is impossible to go from $v_{in}$ to $v_{out}$ with the current assignment.
In the other direction, if we have an assignment, we can extract a set $S$ in
$G$ by interpreting the assignments to the $k$ groups as the indices of the
selected vertices.

Suppose now that there is an algorithm solving \nonreach\ in time
$2^{(1-\eps)m}N^{d}$, for some fixed $d$. The running time of executing the
reduction and then the algorithm would be at most $2^{(1-\eps)k\log n}n^{5d} =
n^{k-\eps k +5d}$. This implies that when $k\ge\frac{5d}{\eps}+1$ the running
time is at most $n^{k-\eps}$, therefore there exist $k,\eps$ for which we can
obtain a fast enough algorithm for \kncut.  \end{proof}

\subsection{Degenerate Deletion}\label{sec:degen}

In this section we consider the problem of editing a graph by deleting vertices
so that the result becomes $r$-degenerate. Recall that a graph is
$r$-degenerate if all of its induced subgraphs contain a vertex of degree at
most $r$. Equivalently, a graph is $r$-degenerate if there exists an ordering
of its vertices $v_1,\ldots,v_n$ such that each $v_i$ has at most $r$ neighbors
among vertices with higher index. 

The problem we are interested in is whether we can delete $k$ vertices of a
given graph so that the graph becomes $r$-degenerate (where $k$ can be thought
of as the parameter and $r\ge 2$ is an absolute constant). This problem was
shown to be W[P]-complete by Mathieson \cite{Mathieson10}. By reusing
essentially the same chain of reductions, we are able to show the following
equivalence between this problem and \circseth.

\begin{theorem} For all $r\ge 2$ we have the following. There exists
$\epsilon>0, k_0>0$ and an algorithm which for all $k>k_0$, if given as input
an $n$-vertex graph, decides if $G$ can be made $r$-degenerate by deleting $k$
vertices  in time $n^{(1-\eps)k}$ if and only if the \circseth\ is false.
\end{theorem}

\begin{proof}

For the direction that \circseth\ implies hardness for our problem we
essentially reuse the reduction of \cite{Mathieson10}, which relies on a
previous reduction of \cite{Szeider05}, and observe that these reductions from
the standard W[P]-complete problem we dealt with in \cref{thm:horn} actually
preserve the value of $k$ exactly and only increase the size of the instance by
a constant factor, therefore the lower bound is preserved as well. For the
other direction we will perform a direct reduction to \textsc{Circuit-SAT}.

Let us give the relevant details. For the first direction, we start with an
instance of monotone weight-$k$ circuit satisfiability. By \cref{thm:horn}, if
we can solve such an instance in time $s^{(1-\eps)k}$ for monotone circuits of
$n$ inputs and size $s$, assuming $k$ is a large enough constant, we will
falsify the \circseth.  As shown by Szeider \cite{Szeider05}, there is a simple
reduction from this problem to the \textsc{Cyclic Monotone Circuit Activation}
problem.  In the latter problem we are given a monotone boolean circuit which
is a general digraph (not a DAG) and are asked to select $k$ initial gates to
activate (set to True). From that point on at each round a gate becomes True if
at least one of its in-neighbors is True for disjunction ($\lor$) gates and if
all its in-neighbors are True, for conjunction ($\land$) gates. The question is
whether there is a selection of $k$ initial gates that eventually makes
everything True.

Szeider observed that we can reduce weight-$k$ monotone circuit satisfiability
to this problem as follows: take $k+1$ copies of the given circuit, identify
their inputs,  and then replace the inputs with conjunctions whose inputs are
the $k+1$ outputs. If there is a weight-$k$ assignment to the original circuit
we use the same activation set and it activates all outputs, therefore
eventually all gates. If on the other hand there is a way to activate
everything, because we have $k+1$ copies, there is at least one copy in which
the output was activated using only the input gates, so from this we can
extract a satisfying assignment to the original circuit. Mathieson
\cite{Mathieson10} presents a reduction from this problem which, for each $r\ge
2$ produces an equivalent instance of $r$-\textsc{Degenerate Deletion}, without
modifying $k$, by replacing each gate with a constant-size gadget. Executing
these reductions starting from a circuit of size $s$ we obtain a graph on
$N=O(s)$ vertices. If the problem could be solved in time $N^{(1-\eps)k}$, we
would have an algorithm deciding weight-$k$ satisfiability of the original
instance in time $O(s^{(1-\eps)k})$, disproving the \circseth. 

For the converse direction, suppose that \circseth\ is false and that we have
an instance of the $r$-\textsc{Degenerate Deletion} problem with budget $k$ on
a graph with $n$ vertices. Assume without loss of generality that $n$ is a
power of $2$ (otherwise add some dummy vertices). We will reduce this to
satisfiability of a circuit with $k\log n$ inputs and size $n^{O(1)}$.
Essentially, this circuit encodes the workings of a polynomial-time algorithm
which checks if a solution is correct. The solution is encoded by the $k\log n$
inputs which are meant to be seen as $k$ groups of $\log n$ bits, encoding the
indices of the selected vertices to delete. The algorithm we are trying to
encode can be seen as follows: at each step we mark some active vertices
(initially this set is the set of deleted vertices); then if a vertex has at
most $r$ unmarked neighbors, we mark it; continue this exhaustively for $n$
rounds and in the end check if all vertices are marked. This algorithm can
easily be encoded by a circuit of $n$ layers, each consisting of $n$ gates,
where vertex $j$ of layer $i$ encodes whether vertex $j$ is marked after $i$
rounds. 

If we had the supposed algorithm for checking circuit satisfiability, it would
run in time $2^{(1-\eps)k\log n}n^{O(1)} = n^{(1-\eps)k+O(1)}$. Setting $k_0$
sufficiently large that $\epsilon k_0/2$ is larger than the constant additive
term in the exponent gives the desired algorithm.  \end{proof}

\section{Independent Set and Classes of Perfect Graphs}\label{sec:indset}

In this section we present some results that seek to clarify the structural
relations between fine-grained questions on the complexity of
\textsc{Independent Set} parameterized by vertex deletion distance to a base
class. Let us recall the general context. Fix a hereditary (induced-subgraph
closed) class $\mathcal{C}$, such that \textsc{Independent Set} is
polynomial-time solvable on graphs from $\mathcal{C}$. We will now consider
instances from the class $\mathcal{C}+kv$, that is, the class of graphs which
can be obtained by a graph in $\mathcal{C}$ by adding at most $k$ new vertices
(and arbitrarily many edges incident on the new vertices). We parameterize by
$k$ and ask what is the complexity of the best FPT algorithm. Recall again that
the problem is clearly solvable in $2^kn^{O(1)}$ by branching on the $k$ extra
vertices, assuming that these are given in the input\footnote{We will always
make this assumption, as we are concerned with the question of exploiting,
rather than finding, this set of vertices.}. Furthermore, this simple algorithm
also applies to \textsc{Weighted Independent Set}, as long as this problem is
also polynomial-time solvable on $\mathcal{C}$.

As a case study, we will consider four standard graph classes $\mathcal{C}$:
Cluster graphs, that is, graphs where every connected component is a clique;
Block graphs, that is, graphs where every $2$-vertex-connected component is a
clique; Cographs, that is, $P_4$-free graphs; and interval graphs.
\textsc{Weighted Independent Set} is known to be in P for all these classes
(even if input weights are encoded in binary). Furthermore, cluster graphs are
clearly a subclass of all the other classes: for block graphs this follows from
the definition; for cographs, this follows because cluster graphs are
$P_3$-free; and for interval graphs this follows because cliques are interval
graphs and interval graphs are closed under disjoint union. However, observe
that the other three classes are incomparable: any tree is a block graph, but a
$K_{1,3}$ where each edge is sub-divided once is neither a cograph nor an
interval graph; a complete bipartite graph $K_{n,n}$ is a cograph, but not a
block graph nor an interval graph; and a disjoint union of a diamond ($K_4-e$)
and a $P_4$ is an interval graph but not a block graph nor a cograph.

We can therefore see that graph theory already gives us some basic information
about the relative difficulties of solving \textsc{Independent Set} on
$\mathcal{C}+kv$ graphs, where $\mathcal{C}$ is one of the classes above, as it
indicates that the easiest case should be when $\mathcal{C}$ is the class of
cluster graphs. Nevertheless, this tells us nothing about the relationship
between the other classes and it also ceases to apply if we want to compare the
complexity of \textsc{Weighted Independent Set} on cluster graphs with that of
\textsc{Independent Set} on the other classes.

Our point of departure is a result which recently appeared in \cite{EsmerFMR24}
and establishes the following:

\begin{theorem}[\cite{EsmerFMR24}]\label{thm:focke} If there exist $\eps>0$ and
an algorithm which takes as input a graph $G=(V,E)$ and a set $M\subseteq V$
such that $G-M$ is a disjoint union of cliques, and solves \textsc{Independent
Set} on $G$ in time $(2-\eps)^{|M|}|V|^{O(1)}$, then the SETH is false.
\end{theorem}

We note that the result of \cite{EsmerFMR24} is stated in terms of the more
general $q$-\textsc{Colorable Deletion} problem, where we want to delete the
minimum number of vertices to make a graph $q$-colorable; \textsc{Independent
Set} is just the case $q=1$. Furthermore, even though this is not explicitly
stated, one can easily verify that in the case $q=1$ the modulator constructed
in the reduction of \cite{EsmerFMR24} is indeed a modulator to a cluster graph,
as all gadgets used for $q=1$ are cliques. We give a self-contained exposition
of this further below.

It would therefore appear that, in some sense, the case is closed. All the
problems we have promised to consider are solvable in time $2^kn^{O(1)}$; but
according to \cref{thm:focke} even the easiest among them is \emph{not}
solvable in time $(2-\eps)^kn^{O(1)}$.  Hence, if one believes the SETH, there
is nothing else to investigate here and all these problems have the same
complexity, despite the fact that from the graph-theoretic point of view,
interval graphs (for instance) are clearly a much wider class of graphs than
cluster graphs.  This seems rather paradoxical.

The results we present in this section will somewhat resolve this tension and
confirm the graph-theory based intuition that these problems are \emph{not} of
the same difficulty.  Concretely, we give a hierarchical ranking of these
problems (taking also into account their weighted versions) which will allow us
to compare their difficulties by examining which case is reducible to which
other, using the main complexity hypothesis we have discussed as our guiding
framework. In summary, we will show the following: (see also
\cref{fig:indset})

\begin{enumerate}

\item As expected, cluster graphs are the easiest of the considered cases, even
if we allow (binary) weights. We show (\cref{thm:cluster-upper}) that refuting
the \aldseth\ would be sufficient to achieve a faster than brute-force
algorithm for this, placing it somewhere below the middle class of
\cref{fig:results}.  For the case of unary weights, we also show that the
problem is at least as hard as the \maxsatseth\ (\cref{thm:cluster-lower}),
placing both weighted variants of this problem between the second and third
classes (from the bottom) of \cref{fig:results}.

\item On the other hand, we show that for both cographs and block graphs,
obtaining a faster than brute-force algorithm seems harder, as we are able to
establish that refuting the \aldseth\ is \emph{both necessary and sufficient}
for this. This places these problems in the third class of \cref{fig:results}.
Equivalence continues to hold even if we allow unary weights.

\item Finally, for the case of interval graphs with unary weights we obtain an
\emph{equivalence} with the \logpwmseth, establishing that obtaining a faster
than brute-force algorithm for this problem is exactly as hard as refuting this
conjecture. So, this problem is placed squarely inside the fourth class (from
the bottom) of \cref{fig:results}.

\end{enumerate}

Our results therefore point to a clear ranking of these questions, albeit not
with respect to their time complexity, but with respect to their structural
complexity. It is also important to note that this investigation goes much
further than just confirming the graph-theory intuition that cluster graphs are
easier, as there is no obvious graph-theoretic reason why (unary)
\textsc{Weighted Independent Set} on cographs should be reducible to (and
probably easier than) the same problem on interval graphs (recall that these
classes are incomparable).

An interpretation of our results is to conclude that some problems seem to be
strictly more difficult than others, even when all the questions we are dealing
with (are likely to) have the same time and space complexity.  For instance,
(unary) \textsc{Weighted Independent Set} parameterized by a modulator to an
interval graph appears \emph{harder than} (that is, appears not reducible to)
the same problem parameterized by a modulator to a cluster graph.  Our argument
for this is that if the contrary were true, two of the classes of
\cref{fig:results} (the ones containing the \aldseth\ and the \logpwmseth)
would collapse. Given the informal intuition that these classes can be thought
of as rough analogues of NC$^1$ and NL, we would find this rather remarkable,
so we interpret this as evidence that the interval graph question is
structurally harder.

Our investigation is then part of a first foray into the structural complexity
of fine-grained parameterized questions. As with classical complexity theory,
one of the great advantages of constructing a framework around equivalence
classes is that this allows us to draw connections and make comparisons between
questions which are at first glance unrelated. Indeed, in
\cref{sec:applications} we consider the complexity of \textsc{Chromatic Number}
on graphs which are $k$ vertices away from having constant pathwidth. At first
glance this type of question seems incomparable to those of this section, since
we are talking about a different problem on a different class of graphs.
Nevertheless, because we show this problem to be equivalent to the \aldseth\ we
immediately obtain a relation between this question and the questions of this
section. In particular, \textsc{Chromatic Number} parameterized by a modulator
to a bounded-pathwidth graph is easier than \textsc{Weighted Independent Set}
on interval graphs, again unless the classes containing the \aldseth\ and
\logpwmseth\ collapse. We believe that many more such characterizations await
to be discovered.

\subparagraph*{Techniques} Let us also briefly sketch some of the main ideas we
use. We start with the case of cluster graphs, where for the lower bound we
rely on a slight modification of the reduction of \cite{EsmerFMR24} (for which
we give a self-contained exposition), while for the upper bound we rely on
classical results from circuit complexity (\cref{thm:cook}). The intuition that
the \aldseth\ morally corresponds to NC$^1$ is crucial here. For block graphs
and cographs, the lower bounds rely on reductions from the \aldseth\ where we
can assume that the input circuit has fan-out $1$, and is therefore acyclic
(thanks to the depth bound). Block graphs and cographs both possess a tree-like
structure that allows us to represent the circuit.  In turn, this structure has
at most logarithmic tree-depth, enabling a reduction to \textsc{MaxSAT} that
gives the upper bounds via \cref{thm:aldsethmaxsat}.

The most technically challenging of the questions we tackle is the case of
weighted interval graphs. The upper bound, given via a reduction to the
\logpwmseth\ is not too difficult, as it consists of an encoding of the direct
DP algorithm for this problem. However, for the converse reduction we need a
gadget that represents all the assignments of a formula of pathwidth $p$ via an
interval graph which for each bag contains $2^p$ intersecting intervals (one
for each assignment). The difficulty is then to ensure that an independent set
selection in one bag correctly propagates to the rest of the graph.

\subsection{Cluster Graphs}\label{sec:cluster}

In this section we present our results for \textsc{Weighted Independent Set}
parameterized by cluster vertex deletion, that is, the size of a modulator to a
cluster graph. For the lower bound we present a construction that is a slight
modification of that used in \cite{EsmerFMR24} and which we will reuse later.
For the upper bound we construct a circuit with one input for each vertex of
the modulator $M$.  The idea is that this circuit should output the weight of
the best independent set consistent with this choice in the modulator, while
still having depth $\eps M$. For this, it is crucial that we can compute the
sum of $N$ integers with $N$ bits in depth $O(\log N)$ (this is a standard but
non-obvious fact from circuit complexity), as this allows us to sum up the
weight contributions from all the cliques of the cluster graph.

\subsubsection{Lower Bound}

Before we proceed to the lower bound, we recall a gadget construction from
\cite{EsmerFMR24} which we will repeatedly reuse (with minor modifications). We
therefore give here all the relevant details, for the sake of completeness.

\subparagraph*{Modulator Gadget.} In several reductions in this section we will
need to reduce a satisfiability problem, where we are given a boolean formula
or circuit with a special set of variables $M$, into an instance of
\textsc{Independent Set} where $M$ is represented by a special set of vertices
$M'$ (a modulator). One obstacle in such constructions is that, even though all
$2^{|M|}$ boolean assignments to the original instance are a priori equally
plausible, in the new instance, solutions which place more vertices of $M'$ in
the independent set have higher objective value, making it complicated to
achieve a one-to-one correspondence between solutions.  We therefore use a
trick, going back to \cite{CyganDLMNOPSW16}, which represents $M$ by a
modulator of slightly larger size (say $|M'|=(1+\eps)|M|$), which allows us to
map assignments of $M$ to (roughly) bisections of $M'$, that is, independent
sets which select roughly half the vertices. By adding appropriate simple
gadgets we can ensure that any reasonable solution will indeed take an
approximate bisection of $M'$ in the independent set, and therefore will encode
an assignment to $M$, while keeping the objective value in the new instance
predictable.

We start with a simple proposition which states that, for all $\eps$, we can
make the extra cost of this gadget at most $(1+\eps)$.

\begin{proposition}\label{prop:gamma} For all $\eps>0$ there exists an integer
$\rho$ and an odd integer $\gamma$ such that $2^\rho\le {\gamma \choose
(\gamma-1)/2}$ and $\gamma\le (1+\eps) \rho$. \end{proposition}

\begin{proof}

First, suppose that $\rho$ is sufficiently large so that there exists an odd
integer $\gamma$ in the interval $[(1+\eps/2)\rho,(1+\eps)\rho]$ (for instance
it suffices to have $\rho>10/\eps$). We will set $\gamma$ to be such an
integer, so $(1+\eps)\rho\ge \gamma\ge (1+\eps/2)\rho$.

We now observe that ${\gamma \choose (\gamma-1)/2}\ge
\frac{2^\gamma}{\gamma+1}$ (because $\sum_{i=0}^\gamma {\gamma\choose i} =
2^\gamma$ and $\gamma\choose i$ is maximized for $i=\frac{\gamma-1}{2}$). Since
our goal is to have ${\gamma \choose (\gamma-1)/2}\ge 2^\rho$ it will suffice
to have $\frac{2^\gamma}{\gamma+1}\ge 2^\rho$. For this, it suffices to have
$2^{\eps \rho/2} \ge (1+\eps)\rho +1$, which always holds for sufficiently
large $\rho$.  \end{proof}

We now describe the modulator gadget we will use in several reductions of this
section. Fix some $\eps>0$ and to ease notation we will switch the informal
roles of $M, M'$ given above. So,  given a set of boolean variables $M'$, we
will construct a graph $G_M$ consisting of a set of vertices $M$ and a disjoint
union of cliques (that is, $M$ is a modulator to a cluster graph), as well as a
target integer $T$. 

More precisely, recall that \cref{prop:gamma} ensures that there exist integers
$\rho,\gamma$, with $\gamma$ odd, so that $\gamma\le (1+\eps)\rho$ and $2^\rho
\le {\gamma \choose (\gamma-1)/2}$.  We partition the variables of $M'$ into
groups of size at most $\rho$, obtaining at most
$\lceil\frac{|M'|}{\rho}\rceil$ groups.  For each group of variables we
construct in our new modulator $M$ a group of $\gamma$ vertices. Therefore,
$|M|\le (1+\eps)|M'|+O(1)$.  

For each such group of vertices we do the following: construct $\gamma$
identical cliques of size $\gamma\choose (\gamma+1)/2$ where for each clique
each vertex is adjacent to a distinct set of vertices of the group of size
$\frac{\gamma+1}{2}$.  Finally, we set the target independent set size $T$ to
be equal to $(\frac{\gamma-1}{2}+\gamma)\lceil\frac{|M'|}{\rho}\rceil$. In
other words, the maximum independent set is expected to select, for each group,
one vertex from each of the $\gamma$ cliques, plus $\frac{\gamma-1}{2}$
vertices of $M$.

\begin{proposition}\label{prop:modulator} Our construction $G_M$ has the
following properties:

\begin{enumerate}

\item $|M|\le (1+\eps)|M'|+O(1)$.

\item For each group of $\rho$ variables of $M'$ and the corresponding group of
$\gamma$ vertices of $M$ we can construct an injective function mapping each
truth assignment of the $\rho$ boolean variables to a distinct set of
$\frac{\gamma-1}{2}$ vertices.

\item All independent sets of $G_M$ have size at most $T$.

\item Any independent set that has size $T$ must contain exactly
$\frac{\gamma-1}{2}$ vertices from each group of $M$.

\item Any independent set of vertices that contains $\frac{\gamma-1}{2}$
vertices from each group of $M$ can be extended to an independent set of size
$T$.

\end{enumerate}

\end{proposition}


\begin{proof}

We have already argued for the first property. For the second property, recall
that $2^\rho\le {\gamma \choose \frac{\gamma-1}{2}}$, so the number of sets of
size $\frac{\gamma-1}{2}$ suffices to construct an injective mapping. For the
third and fourth properties, observe that in a group an independent set that
contains none of the clique vertices can have at most $\gamma$ vertices (from
$M$); while an independent set that contains some vertices of the cliques may
contain at most $\frac{\gamma-1}{2}$ vertices of $M$ from this group (otherwise
the neighborhood of a selected vertex from the clique, which has size
$\frac{\gamma+1}{2}$ would intersect with the selected vertices in $M$).
Hence, no set can have size more than $T$ and to attain this size we must be
selecting exactly $\frac{\gamma-1}{2}$ vertices of $M$ from each group. For the
last property, for each selection of size $\frac{\gamma-1}{2}$ for a group,
each clique contains a vertex whose neighborhood is the complement of this
selection, so we can select this vertex in every clique.  \end{proof}

\begin{theorem}\label{thm:cluster-lower} Suppose there exist $\eps>0$ and an
algorithm which takes as input a vertex-weighted graph $G=(V,E)$ with unary
weights and a modulator $M\subseteq V$ such that $G-M$ is a cluster graph and
computes the maximum weight independent set of $G$ in time
$(2-\eps)^{|M|}|V|^{O(1)}$. Then, the \maxsatseth\ is false. \end{theorem}

\begin{proof}

Shifting notation slightly, fix an $\eps>0$ for which the supposed algorithm
for \textsc{Weighted Independent Set} exists  and runs in time
$2^{(1-\eps)|M|}|V|^{O(1)}$. We will use it to obtain a fast algorithm for
\textsc{Max-SAT}. Suppose that we are given a CNF formula on $n$ variables and
$m$ clauses and a target $t$ and are asked if it is possible to satisfy at
least $t$ clauses. We will show how to do this in time
$(2-\eps')^n(n+m)^{O(1)}$.

Let $M'$ be the set of variables of the input formula. We execute the modulator
gadget construction described above to obtain a graph $G_M$ with the properties
of \cref{prop:modulator}. Assign to each vertex of this construction weight
$m+1$. This part of the construction represents the variables of the input
formula.

To represent the clauses, we consider each clause $c$ in turn. We construct for
each such clause a clique which contains, for each literal $\ell$ of $c$,
$2^{\rho-1}$ vertices. Let $x$ be the variable of the literal $\ell$. Then,
each of these $2^{\rho-1}$ vertices corresponds to an assignment to the
variables of the group containing $x$ that sets $\ell$ to True. Recall that
each such assignment $\sigma$ is mapped injectively to a set of
$\frac{\gamma-1}{2}$ vertices of $M$.  We make the clique vertex representing
$\sigma$ adjacent to all the other ($\frac{\gamma+1}{2}$) vertices of $M$ in
this group. All vertices constructed in this part have weight $1$.

This completes the construction and we claim that the new instance has an
independent set of weight $T(m+1)+t$ if and only if it is possible to satisfy
at least $t$ clauses in the original instance, where $T$ is the integer of
\cref{prop:modulator}. For one direction, if some assignment to the formula
satisfies $t$ clauses, we use the injective mapping from assignments to sets of
vertices of $M$ to obtain a set of $T$ vertices of the first part of the
construction (this is always possible by the last property of
\cref{prop:modulator}), which therefore contributes a weight of $T(m+1)$. We
now claim that for each satisfied clause, there is a vertex in the
corresponding clique that we can add to the set. Indeed, each such clause $c$
must contain a True literal $\ell$, so we have selected an assignment mapped to
a set of $\frac{\gamma-1}{2}$ vertices of $M$ which are disjoint from the
neighborhood of a vertex of the clique representing $\ell$.

For the converse direction, observe that the maximum independent set of the new
instance must select $T$ vertices in the first part. Indeed, since the second
part consists of $m$ cliques, any solution that contains at most $T-1$ vertices
of the first part has weight at most $(T-1)(m+1)+m$, while a solution that
selects $T$ vertices of the first part and nothing else always exists and has
weight $T(m+1)$, which is strictly larger. But then, by the fourth property of
\cref{prop:modulator} we must have selected exactly $\frac{\gamma-1}{2}$
vertices of $M$ from each group. From this selection we extract an assignment
to the variables by selecting for each group the pre-image of the set of
$\frac{\gamma-1}{2}$ vertices of $M$ placed in the independent set under our
injective mapping (if no such pre-image exists, give some arbitrary truth
assignment). We claim this assignment satisfies $t$ clauses. Indeed, there must
be $t$ cliques constructed in the second part where the independent set takes a
vertex, say a vertex representing literal $\ell$ and assignment $\sigma$ to the
variables of the group of the variable of $\ell$. However, since the selected
vertex has $\frac{\gamma+1}{2}$ neighbors in the corresponding group of $M$,
the vertices of $M$ we have selected must be the image of assignment $\sigma$,
so we have indeed given an assignment setting $\ell$ to True.

Finally, observe that $M$ is a modulator to a cluster graph, the reduction
takes polynomial time, and weights are polynomially bounded. Executing the
supposed algorithm on the new instance takes time $2^{(1-\eps)|M|}(n+m)^{O(1)}
= 2^{(1-\eps^2)|M'|}(n+m)^{O(1)} = (2-\eps')^n(n+m)^{O(1)}$, for some
$\eps'>0$.  \end{proof}

\subsubsection{Upper Bound}

Relying heavily on \cref{thm:cook} we have the following:

\begin{theorem}\label{thm:cluster-upper} If the \aldseth\ is false, then we
have the following: there exist $\eps>0$ and an algorithm which takes as input
a graph $G=(V,E)$ and a modulator $M\subseteq V$ such that $G-M$ is a cluster
graph and computes the maximum independent set of $G$ in time
$(2-\eps)^{|M|}|V|^{O(1)}$.  The same holds if the vertices of $G$ have weights
bounded by $2^{|V|^{O(1)}}$ and we want to compute the independent set of
maximum weight. \end{theorem}

\begin{proof}

Suppose the \aldseth\ is false, so there exist $\eps>0$ and an algorithm that
takes as input a bounded fan-in boolean circuit on $N$ inputs with depth at
most $\eps N$ and total size $s$ and decides if the circuit is satisfiable in
time $(2-\eps)^Ns^{O(1)}$.  We will use this algorithm to solve
\textsc{Weighted Independent Set}. We are given a graph $G=(V,E)$ and a
modulator $M$ as in the statement of the lemma.  Let $|M|=m$ and $|V\setminus
M|=n$.

Observe that it is easy to solve \textsc{Weighted Independent Set} of the given
instance in time $2^m n^{O(1)}$: guess the vertices of $M$ to take into the
set, remove all their neighbors from the graph, and then select the maximum
weight vertex from each cluster. We can therefore assume without loss of
generality that for all $\delta>0$ we have $n<2^{\delta m}$, as otherwise,
$2^m<n^{1/\delta}$ and we can solve the given instance in polynomial time.

We now construct a circuit with $N=m$ inputs, one for each vertex of $M$, where
an input is meant to be set to True if the corresponding vertex is placed in
the independent set. For each $v\in V\setminus M$ we construct a gate that
decides if $v$ is available by checking if its neighbors in $M$ have been
selected. This gate is just a conjunction ($\land$) of at most $m$ negated
inputs, so implementing this with bounded fan-in gates gives depth at most
$O(\log m)$.

Let $W$ be the smallest power of $2$ such that the sum of all input weights is
at most $2^W$, that is, the sum of all input weights can be encoded with $W$
bits. We have $W=(n+m)^{O(1)}$, because each input weight can be encoded with
$(n+m)^{O(1)}$ bits. 

For each $v\in V$ we construct a group of $W$ gates, hard-coded to $0$ or $1$,
that encode the weight of $v$. We take the conjunction of each such bit with
the output of the gate that states whether $v$ is available (for $v\not\in M$)
or with the corresponding input gate (for $v\in M$). We now have $n+m$ blocks
of $W$ gates each, encoding the weight that each vertex may contribute to the
solution without selecting the two endpoints of any edge incident on $M$.  What
remains is to add some machinery to ensure we also do not select two vertices
from a clique and that we select the maximum weight vertex from each clique.

Consider a clique $C$ of $G-M$ and let $C=\{v_1,\ldots,v_c\}$. We construct a
\textsc{Max} circuit which, given as input the $c$ blocks of $W$ gates encoding
the weights of vertices from $C$ compatible with a choice for $M$ outputs the
weight of the heaviest available vertex of $C$. As stated in \cref{thm:cook}
this can be done with depth $O(\log W) = O(\log (n+m))$.

Suppose that $G-M$ contains $r$ clique components $C_1,\ldots,C_r$ and for each
such component we have constructed a circuit with $W$ outputs that encodes the
weight of the heaviest available vertex. We now add a circuit that computes the
sum of these numbers as well as the (at most) $m$ weights of the vertices of
the modulator (taking into account that the weight of a vertex whose input is
set to $0$ becomes $0$). By \cref{thm:cook} we construct a circuit with depth
$O(\log W) = O(\log (n+m))$. Finally, given some desired target value, we add
to the end a circuit checking that the calculated sum is at least as large as
the target; this adds another $O(\log W)=O(\log (n+m))$ to the depth.

What remains is to argue that the depth of the constructed circuit is at most
$\eps N = \eps m$. However, we can see that the total depth is at most $O(\log
(n+m))$. We have $\log(n+m)\le \log n + \log m \le \delta m$ for all
$\delta>0$. This follows because we can assume that $m$ is sufficiently large
so that $\log m$ is insignificant compared to $m$; and because as we argued
$n<2^{\delta m}$ for all $\delta>0$.  \end{proof}

\subsection{Cographs}\label{sec:cographs}

In this section we deal with cographs, which as mentioned are exactly the
graphs which exclude $P_4$ as an induced subgraph. We will use an equivalent
characterization of this class of graphs, which states that $G$ is a cograph if
it can be constructed from singleton vertices using repeated applications of
the union and join operations (the join operation takes the union of two graph
$G_1,G_2$ and adds all edges between them). It is known that, given a cograph
$G$, one can construct in polynomial time a binary tree proving that $G$ can be
constructed using these operations \cite{CorneilPS85}.

Our lower bound reduces the \aldseth\ to \textsc{Independent Set} parameterized
by cograph modulator size. The two main intuitions we rely on are that circuits
with $N$ inputs and depth $\eps N$ can, without loss of generality, be
considered acyclic; and that the binary tree construction representing a
cograph can encode a sufficient amount of information to represent an acyclic
(tree-like) circuit.  For the upper bound, we rely on the fact that this
tree-like representation has logarithmic tree-depth.

\subsubsection{Lower Bound}

\begin{theorem}\label{thm:cograph-lower} Suppose there exist $\eps>0$ and an
algorithm which takes as input a graph $G=(V,E)$ and a modulator $M\subseteq V$
such that $G-M$ is a cograph and computes the maximum independent set of $G$ in
time $(2-\eps)^{|M|}|V|^{O(1)}$. Then, the \aldseth\ is false.  \end{theorem}

\begin{proof}

Slightly switching notation, we are given a bounded fan-in boolean circuit with
$N$ inputs, size $s$, and depth $\eps N$ and want to decide if it is
satisfiable in time $2^{(1-\eps)N}s^{O(1)}$, for some $\eps>0$. We first show
that we can assume without loss of generality that the circuit is actually a
formula, that is, all non-input gates have fan-out $1$. 

Indeed, suppose that we have an algorithm for this special case running in time
$2^{(1-\eps)N}s^c$, for some fixed constant $c$ and for circuits of depth at
most $\eps N$.  We will show that if we are given a depth $\eps' N$ circuit
(for some $\eps'<\eps$ to be set later) where some gates have fan-out more than
$1$ we can decide its satisfiability in time $2^{(1-\eps')N}$. More precisely,
in this case we replicate each gate with large fan-out that is at distance $d$
from the output with $2^d$ gates of fan-out $1$, increasing the size of the
circuit by at most a factor of $2^{\eps' N}$. The new circuit has fan-out $1$
and depth $\eps' N<\eps N$, so the supposed algorithm solves it in
$2^{(1-\eps)N}2^{\eps'cN}s^c$. It suffices to have $\eps-\eps'c>\eps'$ or
equivalently to set $\eps'<\frac{\eps}{c+1}$ to obtain the desired algorithm
for the case of circuits with larger fan-out.

We can therefore view the input circuit as a rooted tree, with the output gate
at the root, and each leaf labeled with one of the $N$ input variables or its
negation (which of course may appear on many leaves). Notice that by moving the
negations as far away from the root as possible we can ensure that all internal
gates are conjunctions or disjunctions.

We will now produce an instance of \textsc{Independent Set}. Suppose the
hypothesized algorithm for \textsc{Independent Set} runs in time
$2^{(1-\delta)|M|}|V|^{O(1)}$.

For the first part of our construction, we use the modulator gadget of
\cref{prop:modulator} to represent the $N$ input variables using a modulator
$M$ of size at most $(1+\delta)N+O(1)$. Let $T$ be the target independent set
size in this part and recall that we have now partitioned $N$ into groups of
$\rho$ variables, each of which is represented by $\gamma$ vertices of $M$ such
that each assignment is encoded by a distinct set of $\frac{\gamma-1}{2}$
vertices.

For the second part, for each gate $g$ of our circuit we will consider the
sub-tree rooted at this gate and produce a pair $(G_g, T_g)$ such that (i)
$G_g$ is a cograph (ii) all independent sets of $G_g$ have size at most $T_g$
(iii) it is possible to place $T_g$ vertices of $G_g$ into an independent set
if and only if our selection in the modulator part encodes an assignment that
sets gate $g$ to True. These properties will be established by induction.

We start with the leaves, which are labeled with literals using input
variables. For such a gate labeled with a literal $\ell$ using variable $x$ we
construct a clique on $2^{\rho-1}$ vertices, where each vertex corresponds to
an assignment to the $\rho$ variables of the group of $x$ that sets $\ell$ to
True. Take each such vertex that represents an assignment $\sigma$ and make it
adjacent to the $\frac{\gamma+1}{2}$ vertices of the group in $M$ which are the
complement of the encoding of $\sigma$. The target size for this graph is $1$.
We see that (i) is satisfied as cliques are cographs (ii) is satisfied as the
maximum independent set size of a clique is $1$ and (iii) is satisfied because
if we have selected an assignment to the group containing $x$ that falsifies
$\ell$, then all vertices of the clique have a neighbor in the independent set;
while in the opposite case, one vertex of the clique is available.

We now proceed inductively. For an AND gate $g$ with children $g_1,g_2$,
suppose we have constructed the pairs $(G_{g_1},T_{g_1})$ and
$(G_{g_2},T_{g_2})$. We construct the pair $(G_{g_1}\cup G_{g_2},
T_{g_1}+T_{g_2})$. We have that (i) the new graph is a cograph, as it is the
disjoint union of two cographs (ii) its maximum independent set is at most
$T_{g_1}+T_{g_2}$ (iii) an independent set can have this many vertices only if
it has $T_{g_1}$ vertices from $G_{g_1}$ and $T_{g_2}$ vertices from $G_{g_2}$,
which implies that $g_1,g_2$ evaluate to True, therefore so does $g$;
conversely, if the assignments sets $g$ to True, then $g_1,g_2$ are set to
True, so independent sets of sizes $T_{g_1},T_{g_2}$ exist by induction.

For an OR gate $g$ with children $g_1,g_2$, suppose we have constructed the
pairs $(G_{g_1},T_{g_1})$ and $(G_{g_2},T_{g_2})$, where without loss of
generality $T_{g_1}\le T_{g_2}$. We set $G_g$ to be the join of $G_{g_1}\cup
(T_{g_2}-T_{g_1})K_1$ with $G_{g_2}$ and $T_g=T_{g_2}$. In other words, we add
$T_{g_2}-T_{g_1}$ isolated vertices to $G_{g_1}$, take the join of the two
graphs, and set the target to the higher of the targets of the two children. We
have that (i) the new graph is a cograph, as it is obtained from cographs via
union and join operations (ii) no independent set can have size more than
$T_{g_2}$, since any independent set can only take vertices from one side of
the join operation (iii) if an independent set of the new graph has size
$T_{g_2}$ either $G_{g_2}$ has an independent set of this size (so $g_2$
evaluates to True, therefore so does $g$) or $G_{g_1}$ has an independent set
of size $T_{g_1}$ (and similarly $g$ evaluates to True); conversely if $g$
evaluates to True, then either $g_1$ does too (so we can form an independent
set of size $T_{g_1}$ and augment it with the $T_{g_2}-T_{g_1}$ new vertices)
or $g_2$ does (so we have an independent set of size $T_{g_2}$). 

Once we arrive at the root $r$ we have constructed a cograph $G_r$ and a target
size $T_r$. We set the total target size to be $T+T_r$ and our instance is made
up of $G_r$ and the modulator gadget construction which consists of $M$ and
some disjoint cliques. Therefore, removing $M$ indeed gives us a cograph and
the whole construction can be performed in polynomial time. Correctness is not
hard to see: because $G_r$ cannot have an independent set of size more than
$T_r$ and the modulator part cannot have an independent set of size $T$
(\cref{prop:modulator}) we must take exactly $T$ vertices from the modulator
part and $T_r$ vertices from the cograph. The $T$ vertices from the modulator
part force us to encode an assignment to the $N$ inputs, which must set the
output gate to True (otherwise we cannot attain $T_r$ vertices in the cograph).

Running the hypothetical algorithm on this instance takes time
$2^{(1-\delta)(1+\delta)N}s^{O(1)} = 2^{(1-\delta^2)N}s^{O(1)}$. So, for
$\eps=\delta^2$, if we start with a circuit of fan-out $1$ and depth at most
$\eps N$ we obtain the desired algorithm.  \end{proof}

\subsubsection{Upper Bound}

\begin{theorem}\label{thm:cograph-upper} If the \aldseth\ is false, then we
have the following: there exist $\eps>0$ and an algorithm which takes as input
a graph $G=(V,E)$ and a modulator $M\subseteq V$ such that $G-M$ is a cograph
and computes the maximum independent set of $G$ in time
$(2-\eps)^{|M|}|V|^{O(1)}$.  The same holds if the vertices of $G$ have weights
bounded by $|V|^{O(1)}$ and we want to compute the independent set of maximum
weight. \end{theorem}

\begin{proof}

Suppose the \aldseth\ is false, so by \cref{thm:aldsethmaxsat} there is an
algorithm that takes as input an integer $t$, a CNF formula $\phi$, a modulator
$M'$, and a tree-depth decomposition of the primal graph of $\phi-M'$ of depth
$O(\log|\phi|)$, and decides if at least $t$ clauses of $\phi$ can be
satisfied, in time $(2-\eps)^{|M'|}|\phi|^{O(1)}$. We will use this algorithm
to solve \textsc{Independent Set}.  We are given a graph $G=(V,E)$ and a
modulator $M$ as in the statement of the lemma. Let $|M|=m$ and $|V\setminus
M|=n$. 

Recall that for any cograph $H$ one can construct in polynomial time a rooted
binary tree showing how $H$ can be constructed from isolated vertices using
Union and Join operations. In particular, there exists a rooted binary tree
where each internal node has two children and is labeled Union or Join and
where we can inductively associate with every node of the tree an induced
subgraph of $H$ as follows: leaves correspond to single vertices of $H$; for an
internal node labeled Union (resp.  Join) the corresponding graph is
constructed by taking the disjoint union (resp. the join, that is, the union
where we add all possible edges between the two parts) of the graphs
corresponding to the children. In such a tree the root corresponds to $H$ and
there is a one-to-one correspondence between vertices of $H$ and leaves of the
tree.

We construct a CNF formula for $G$ as follows. First, for each vertex $v\in V$
we construct a variable $x_v$ (indicating whether $v$ is in the maximum
independent set). Then, we construct the binary tree representing the cograph
$G-M$, which has at most $2n$ nodes. For each node $t$ of this tree we
construct a variable $a_t$, informally indicating whether the corresponding
sub-tree is ``active'', that is whether the optimal independent set contains
any vertices from this sub-tree. We now construct the following clauses:

\begin{enumerate}

\item For each $v_1v_2\in E$ such that $\{v_1,v_2\}\cap M\neq \emptyset$ we add
the clause $(\neg x_{v_1} \lor \neg x_{v_2})$.

\item For each leaf node representing the vertex $v\in V\setminus M$ we add the
clause $(\neg x_v \lor a_v)$.

\item For each non-root node $t$ of the binary tree, if $p$ is the parent of
$t$ in the tree we add the clause $(\neg a_t \lor a_p)$.

\item For each Join node $t$ with children $t_1,t_2$ we add the clause $(\neg
a_{t_1}\lor \neg a_{t_2})$.

\end{enumerate}

We set $M'=\{ x_v\ |\ v\in M\}$. We have now constructed a CNF formula $\phi$
with the following properties:

\begin{enumerate}

\item For all $\ell>0$, $G$ has an independent set of size $\ell$ if and only
if $\phi$ has a satisfying assignment setting $\ell$ variables from $\{ x_v\ |\
v\in V\}$ to True.

\item $\phi-M'$ has tree-depth $O(\log n)$, $|\phi|=(n+m)^{O(1)}$, and
$|M'|=|M|=m$. 

\end{enumerate}

To see the second property, observe that $\phi$ is a 2-CNF formula, and once we
remove $M'$ the primal graph has treewidth $2$ (it is a tree, except the two
children of each Join node are adjacent). Therefore, this graph has tree-depth
$O(\log n)$. For the first property, if we have an independent set $S$ of $G$
of size $\ell$, we set the corresponding variables according to their informal
meaning, that is, $x_v$ is True if and only if $v\in S$; and $a_t$ is True if
and only if $S$ contains a vertex mapped to a leaf that is a descendant of $t$.
It is not hard to see that this satisfies $\phi$. For the converse direction,
we extract a subset $S$ of $V(G)$ from a satisfying assignment of $\phi$ by
selecting all the vertices $v$ for which $x_v$ is True. It is clear that $S$
cannot contain both endpoints of an edge that is incident on $M$, by the
clauses of step 1. Suppose then that $S$ contains both endpoints of an edge of
$G-M$ and let $t$ be the Join node of the tree whose corresponding subgraph
contains this edge and which is as far from the root as possible. Suppose that
the two children of $t$ are $t_1$ and $t_2$, and by the selection of $t$ each
contains one endpoint of the offending edge. We claim that in this case we must
have $a_{t_1},a_{t_2}$ both set to True, thus falsifying a clause constructed
in step 4. To see this, observe that in a satisfying assignment of $\phi$,
whenever we have $x_v$ set to True, we must have $a_t$ set to True for all
nodes $t$ such that $v$ is mapped to a descendant of $t$. This can be seen by
induction: clauses of step 2 ensure that it holds for leaves; and clauses of
step 3 ensure that it holds for all higher levels of the tree.

To conclude the proof we have to argue that the satisfiability algorithm
implied by the assumption that the \aldseth\ is false and
\cref{thm:aldsethmaxsat} allows us to solve the resulting instance. To see
this, we first observe that the algorithm of \cref{thm:aldsethmaxsat} can
handle instances where clauses have (polynomially-bounded) weights (for
instance, by repeating each clause a number of times equal to its weight).  We
can therefore assign to all clauses of our constructed instance a high enough
weight (polynomial in the construction) ensuring that all clauses need to be
satisfied. Then, for each vertex $v$, we have constructed a variable $x_v$. We
simply add to the instance the clause $x_v$ with weight equal to the weight of
the corresponding vertex in the input graph.  \end{proof}

\subsection{Block Graphs} 

A graph is a block graph if every $2$-vertex-connected component is a clique
\cite{BandeltM86}. Block graphs are a well-studied special subclass of chordal
graphs which also includes cluster graphs and trees as subclasses. Clearly,
\textsc{Weighted Independent Set} can be solved in polynomial time on block
graphs. In this section we investigate the complexity of solving this problem
parameterized by the size of a modulator to a block graph. Interestingly, even
though block graphs are incomparable to cographs, the general tenor of our
results is the same for the two problems, namely: we show that beating
brute-force \emph{is equivalent to} refuting the \aldseth\ (using again the
intuition that we can focus on acyclic circuits and that block graphs contain a
tree-like structure in the form of the block decomposition).

\subsubsection{Lower Bound}

\begin{theorem}\label{thm:block-lower} Suppose there exist $\eps>0$ and an
algorithm which takes as input a graph $G=(V,E)$ and a modulator $M\subseteq V$
such that $G-M$ is a block graph and computes the maximum independent set of
$G$ in time $(2-\eps)^{|M|}|V|^{O(1)}$. Then, the \aldseth\ is false.
\end{theorem}

\begin{proof}

Similarly to \cref{thm:cograph-lower}, we will assume that we are given a
boolean circuit with $N$ inputs, size $s$, bounded fan-in, fan-out $1$, and
dept at most $\eps N$ and we want to decide its satisfiability in time
$2^{(1-\eps)N}s^{O(1)}$, for some $\eps>0$. Suppose that an algorithm solves
\textsc{Independent Set} in time $2^{(1-\delta)|M|}|V|^{O(1)}$, given a
modulator $M$ to a block graph.

We construct a graph from the given circuit as follows. To represent the $N$
variables we (again) use the construction of \cref{prop:modulator} to obtain a
set $M$ of at most $(1+\delta)N+O(1)$ vertices. The variables $N$ are
partitioned into groups of size $\rho$ and the groups are matched to disjoint
groups of $\gamma$ vertices. Furthermore, assignments to each group are
injectively mapped to sets of $\frac{\gamma-1}{2}$ vertices inside the
corresponding set of size $\gamma$.

We now inductively build a block graph from the circuit. For each gate $g$ we
build a graph representing the circuit rooted at $g$. If there are $s(g)$ gates
in total in the sub-tree rooted at $g$, then this graph will have maximum
independent set of size exactly $s(g)$. Furthermore, it will contain a special
output vertex $o_g$. We will have the property that it is possible to find an
independent set of size $s(g)$ without selecting $o_g$ if and only if $g$
evaluates to True if we select the input assignment encoded by the independent
set we selected in the modulator $M$.

For an input gate labeled with the literal $\ell$, using variable $x$, we
construct a clique of size $2^{\rho-1}+1$ (so the maximum independent set
indeed has size $s(g)=1$). Label one of the vertices of the clique as the
output vertex. Each other vertex is labeled with an assignment $\sigma$ to the
variables of the group of $x$ that sets $\ell$ to True. For each such vertex
representing assignment $\sigma$, we make it adjacent to the
$\frac{\gamma+1}{2}$ vertices of the corresponding group in $M$ which are the
complement of the encoding of $\sigma$. We can now see that we can form an
independent set of size $1$ while avoiding the output vertex if and only if we
select one of the vertices representing an assignment, which can be done if and
only if we encode an assignment to the input that set $\ell$ to True.

For an OR gate $g$ with two children $g_1,g_2$ we construct a clique of size
$3$. One vertex of this clique is adjacent to $o_{g_1}$, one to $o_{g_2}$, and
one is the output $o_g$ of this gate. Since we added a clique to the graphs
constructed for the sub-trees rooted at $g_1,g_2$, the maximum independent set
size is at most $1+s(g_1)+s(g_2)=s(g)$. If it is possible to find an
independent set of size $s(g)$ avoiding $o_g$, then we must take one of the
other vertices of the new clique, say the one adjacent to $o_{g_1}$; then,
since we must have taken an independent set of size $s(g_1)$ on the sub-tree
for $g_1$, $g_1$ must evaluate to True, and so does $g$. Conversely, if $g$
evaluates to True, then one of its children (say $g_1$) does, so we can find an
independent set that avoids $o_{g_1}$; we then add the corresponding vertex of
the new clique in the independent set and avoid $o_g$.

For an AND gate $g$ we construct two adjacent vertices: one is labeled $o_g$
and is the output, while the other is adjacent to the outputs of the children
of $g$. It is not hard to see that the desired properties are satisfied
inductively in this case as well.

Finally, for the output gate we remove the output vertex that the above
construction would have added. We set the target independent set size to be
$T+s$, where $s$ is the number of gates of the circuit. Correctness is not hard
to see, because by induction we cannot take more than $s$ vertices in the block
graph part, while by \cref{prop:modulator} we cannot take more than $T$ in the
modulator part; hence we take exactly $T$ in the modulator, ensuring that we
can only take $s$ vertices in the block graph part if we satisfy the output
gate (as we have removed its output vertex). To see that removing $M$ results
in a block graph, observe that in the modulator gadget we get a disjoint union
of cliques (which is a block graph), while to represent the circuit we have a
rooted tree of cliques, so each $2$-vertex-connected component is a clique. The
remaining calculations are identical to those of \cref{thm:cograph-lower}.
\end{proof}

\subsubsection{Upper Bound}

\begin{theorem}\label{thm:block-upper} If the \aldseth\ is false, then we have
the following: there exist $\eps>0$ and an algorithm which takes as input a
graph $G=(V,E)$ and a modulator $M\subseteq V$ such that $G-M$ is a block graph
and computes the maximum independent set of $G$ in time
$(2-\eps)^{|M|}|V|^{O(1)}$.  The same holds if the vertices of $G$ have weights
bounded by $|V|^{O(1)}$ and we want to compute the independent set of maximum
weight. \end{theorem}

\begin{proof}

Suppose the \aldseth\ is false, so there is an algorithm that takes as input a
target integer $t$, a CNF formula $\phi$, a modulator $M'$, and a tree-depth
decomposition of the primal graph of $\phi-M'$ of width $O(\log|\phi|)$, and
decides if we can satisfy at least $t$ clauses of $\phi$, in time
$(2-\eps)^{|M'|}|\phi|^{O(1)}$.  We will use this algorithm to solve
\textsc{Independent Set}.  We are given a graph $G=(V,E)$ and a modulator $M$
as in the statement of the lemma. Let $|M|=m$ and $|V\setminus M|=n$. 

We first formulate a claim stating that \textsc{Independent Set} on block
graphs can be encoded by a CNF formula of tree-depth $O(\log n)$.

\begin{claim}\label{claim:block} There is a polynomial time algorithm which
takes as input a block graph on $n$ vertices and produces a CNF formula $\phi$
of size $n^{O(1)}$ and a tree-depth decomposition of $\phi$ of width $O(\log
n)$ as well as a special set of variables $X$ which are in one-to-one
correspondence with the vertices of the graph. The algorithm guarantees that
for all sets of vertices $S$ of the given graph, $S$ is independent if and only
if $\phi$ has a satisfying assignment that sets the variables of $X$
corresponding to $S$ to True.  \end{claim}

Before we give the proof of the claim let us explain why it implies the lemma,
using a strategy that is essentially identical to the proof of
\cref{thm:cograph-upper}.  We apply \cref{claim:block} to the graph $G-M$ and
obtain a formula $\phi$.  We extend this formula with a variable $x_v$ for each
$v\in M$ and add clauses ensuring that for each edge incident on at least one
vertex of $M$, the two endpoints of the edge cannot be selected (this can be
done as $\phi$ has a set of variables $X$ representing the vertices of $G-M$).
We set $M'$ to be the set of new variables and therefore $\phi-M'$ has
tree-depth $O(\log n)$. We have that the new formula has a satisfying
assignment setting $\ell$ variables representing vertices to True if and only
if $G$ has an independent set of size $\ell$. To reduce this maximization
version to the \textsc{MaxSAT} case of \cref{thm:aldsethmaxsat} we again give
very high weight to necessary clauses and introduce clauses $(x_v)$, with
weight equal to the weight of vertex $v$.  We omit the details as they are
identical to the proof of \cref{thm:cograph-upper}.

Let us then proceed to the proof of the claim.

\begin{claimproof}[Proof of \cref{claim:block}]

Recall that a block is a maximal $2$-vertex-connected component in a graph. In
any graph, if $B_1,B_2$ are blocks we have that $|B_1\cap B_2|\le 1$, because
if $B_1,B_2$ share two vertices, then $B_1\cup B_2$ is also
$2$-vertex-connected, contradicting maximality. Furthermore, for any two blocks
$B_1,B_2$ we have at most one edge connecting them. Indeed, if we have two
edges with distinct endpoints connecting $B_1,B_2$, then $B_1\cup B_2$ is
$2$-vertex-connected; while if two edges incident on $x\in B_1$ have their
endpoints in $B_2$, then $B_2\cup\{x\}$ is $2$-vertex-connected. 

Given an input block graph $H$, we construct a graph $H'$ that contains a
vertex for each block of $H$ and has an edge between $B_1,B_2$ if $B_1,B_2$ are
adjacent or share a vertex. We can see that $H'$ is acyclic, as otherwise one
of the blocks would not be maximal. We will use this fact to show that the
\emph{treewidth} of our formula is $O(1)$. Because for all graphs $G$ we have
$\td(G)= O( \tw(G)\log n)$, this will imply the claim.

We begin our formula by constructing, for each vertex $v$ a variable $x_v$
(indicating whether $v$ is in the independent set). These variables will form
the set $X$. We now prove that we can build a CNF formula such that assignments
to $X$ can be extended to satisfying assignments if and only if the set of True
variables of $X$ represents an independent set.

Consider each node of $H'$, which represents a clique $C=\{v_1,\ldots,v_c\}$
with $c=|C|$ of $H$. We introduce $c-1$ new variables $y_1,\ldots,y_{c-1}$
whose informal meaning is that if $y_i$ True then all the variables
$x_{v_1},x_{v_2},\ldots,x_{v_i}$ are False. We add the clause $(\neg y_1\lor
\neg x_{v_1})$ and for all $i\in\{2,\ldots,c-1\}$ the clauses $(\neg y_i\lor
y_{i-1})$ and $(\neg y_i\lor \neg x_{v_i})$. Furthermore, for all
$i\in\{2,\ldots,c-1\}$ we add the clause $(\neg x_{v_i}\lor y_{i-1})$. We now
claim that the new clauses can only be satisfied if and only if at most $1$ of
the $x_{v_i}$ variables is set to True.  To see this, observe that if we set
$x_{v_i}$ and $x_{v_j}$ to True, for $j>i$, we must set $y_{j-1}$ to True; then
we must set $y_{j-2}$ to True, and proceeding in this way $y_i$ to True;
therefore $x_{v_i}$ must be set to False, contradiction. It is easy to see,
however, that setting a single $x_{v_i}$ to True can always be extended to a
satisfying assignment. In other words, the clauses we have constructed so far
ensure that any assignment to $X$ that sets to True to variables representing
vertices of the same clique will fail to satisfy the formula. To conclude the
construction we consider any pair of vertices $u,v$ of $H$ such that $u,v$ are
adjacent but belong in distinct 2-connected components. For any two such
vertices we add the clause $(\neg x_u\lor \neg x_v)$.

It is now not hard to see that the formula we have constructed so far is
correct in the sense that satisfying assignments are always extensions of
assignments that encode independent sets of $H$. Let us then argue that we have
constant treewidth. We consider each vertex of $H'$ separately. Consider the
node of $H'$ that represents the clique $C=\{v_1,\ldots,v_c\}$ and recall that
we used the variables $x_{v_1},\ldots,x_{v_c}$ and the $c-1$ new variables
$y_1,\ldots,y_{c-1}$. Because our clauses have a path-like structure (each
$y_i$ is involved in clauses with $y_{i-1}$ and $x_{v_i}$) we can produce a
tree (in fact path) decomposition with at most $3$ variables per bag.

We now have a tree decomposition for the part of the formula representing each
vertex of $H'$. To transform this into a decomposition of the whole formula, we
consider any two $h_1,h_2$ which are adjacent in $H'$. If these correspond to
blocks that share a vertex $u$, we select a bag in each decomposition that
contains $u$ and make the two bags adjacent; if $h_1,h_2$ represent blocks
connected by an edge with endpoints $u,v$, then we construct a new bag that
contains $u,v$ and make it adjacent to a bag that contains $u$ and a bag that
contains $v$. It is now not hard to see that because $H'$ is acyclic,
performing these operations will merge the tree decompositions for the vertices
of $H$ into a single tree decomposition for the whole formula.
\end{claimproof} 

Having established \cref{claim:block} we can conclude that a fast algorithm
falsifying the \logpwmseth\ would imply a fast algorithm for
\textsc{Independent Set} using a modulator to a block graph.  \end{proof}

\subsection{Weighted Interval Graphs}

In this section we consider the \textsc{Weighted Independent Set} problem on
graphs which are ``close'' to being interval graphs. Our goal is to prove the
following:

\begin{theorem}\label{thm:interval} The following two statements are
equivalent:

\begin{enumerate}

\item There exist $\eps>0$ and an algorithm which takes as input a
vertex-weighted graph $G=(V,E)$ with unary weights and a modulator $M\subseteq
V$ such that $G-M$ is an interval graph and computes the maximum weight
independent set of $G$ in time $(2-\eps)^{|M|}|V|^{O(1)}$.

\item The \logpwmseth\ is false.

\end{enumerate}

\end{theorem}

Before we proceed, let us give some intuition. Recall that the \logpwmseth\
informally corresponds to NL, for which complete problems typically have a
\textsc{Reachability} flavor.  What we are claiming here is that
\textsc{Weighted Independent Set} on graphs which are close to interval graphs
is a problem of the same type.  Informally, this happens because the natural
algorithm for this problem would proceed as follows: first, guess which
vertices of $M$ are in an optimal solution; then, solve the problem on the
remaining instance, which is an interval graph. However, solving
\textsc{Weighted Independent Set} on interval graphs can naturally be done via
dynamic programming, which in turn can be seen as a form of a reachability
question.

Let us begin with the easier direction of the proof of \cref{thm:interval}.

\begin{lemma} The second statement of \cref{thm:interval} implies the first.
\end{lemma}

\begin{proof} By \cref{thm:2sat}, we can assume that there exists an algorithm
which takes as input a CNF formula $\phi$,  a modulator $M'$ of size $m'$, and
a path decomposition of the primal graph of $\phi-M'$ of width $O(\log|\phi|)$
and decides if $\phi$ is satisfiable in time $(2-\eps)^{m'}|\phi|^{O(1)}$, for
some $\eps>0$. 

Suppose now we are given a weighted graph $G=(V,E)$, a modulator $M\subseteq V$
of size $m$ such that $G-M$ is an interval graph, and a target weight $T$.  We
are asked if $G$ has an independent set of weight at least $T$. Let $n=|V|-m$.
Without loss of generality, assume that the sum of all weights is at most $n$
(this can be achieved by adding isolated vertices of weight $0$ to the graph so
that $|V|-m$ becomes larger than the largest weight; this can be performed in
polynomial time, since weights are given in unary) and that $n$ is a power of
$2$ (again, by adding isolated $0$-weight vertices as necessary).  Suppose also
that we have an interval representation of $G-M$, which assigns an interval
$[s_v,t_v]$ to each $v\in V\setminus M'$ such that $[s_v,t_v]\cap [s_u,t_u]\neq
\emptyset$ if and only if $uv\in E$.  Such a representation can be found in
linear time \cite{CorneilOS09}. Again, without loss of generality we may assume
that the endpoints of all intervals are distinct, so we have $2n$ endpoints in
total, and they correspond to the set $\{1,\ldots,2n\}$.

We will construct a CNF formula $\phi$ and a modulator $M'$ with the following
properties:

\begin{enumerate}

\item $\phi$ is satisfiable if and only if $G$ has a weighted independent set
of weight at least $T$.

\item $\phi$ has size polynomial in $n+m$ and $|M|=|M'|=m$. 

\item $\phi-M'$ has pathwidth $O(\log n)$.

\item $\phi,M'$ and a path decomposition of $\phi-M'$ of width $O(\log n)$ can
be constructed in time polynomial in $n+m$.

\end{enumerate}

Clearly, if we can achieve all the above we obtain the lemma, as in such a case
we could invoke the algorithm that breaks the \logpwmseth,  solving the
resulting instance in $(2-\eps)^{m} (n+m)^{O(1)}$.

Assume that the vertices of $M$ are numbered $\{1,\ldots,m\}$ and the vertices
of $V\setminus M$ are numbered $\{0,\ldots,n-1\}$.  We construct $\phi$ as
follows:

\begin{enumerate}

\item For each $v\in M$ we construct a variable $x_v$, indicating whether $v$
is in the optimal independent set. Let $M'$ be the set of these variables.

\item For each $i\in\{1,\ldots,2n\}$ we construct a variable $a_i$ indicating
whether the optimal solution has selected an interval that contains $i$.

\item For each $i\in\{1,\ldots,2n\}$ we construct $\log n$ variables
$y_{i,1},\ldots,y_{i,\log n}$ which are supposed to be the binary
representation of the index of the (at most one) vertex of the independent set
whose interval contains $i$.

\item For each $i\in\{1,\ldots,2n+m\}$ we construct $\log n$ variables
$z_{i,1},\ldots,z_{i,\log n}$.  For $j\in\{1,\ldots,2n\}$ these variables are
supposed to store the weight of the optimal independent set if we consider only
vertices whose intervals intersect $\{1,\ldots,j\}$. For $j>2n$ these variables
will represent the weight of the optimal independent set restricted to $G-M$
and the first $j-2n$ vertices of $M$.

\item For every edge $uv\in E$ such that $u,v\in M$ we construct the clause
$(\neg x_u \lor \neg x_v)$.

\item For every $i\in\{1,\ldots,2n\}$ and every vertex $u\not\in M$ such that
$i$ does \emph{not} belong in the interval associated with $u$ we construct a
clause which initially contains only $\neg a_i$. We add to this clause literals
using the variables $y_{i,1},\ldots,y_{i,\log n}$ which are falsified if and
only if these variables encode the index of $u$.

\item For every edge $uv\in E$ such that $u\in M$ and $v\not\in M$, for each
$i\in \{1,\ldots,2n\}$ such that $i$ belongs in the interval associated with
$v$, we add the clause $(\neg a_i \lor  \neg x_u)$; we then add to this clause
the variables $y_{i,1},\ldots,y_{i,\log n}$ properly negated so that these
literals are all false if and only if their assignment corresponds to the
binary representation of the index of $v$.

\item For every $i\in\{1,\ldots,2n-1\}$ and for any two distinct vertices
$u,v\not\in M$ such that $i$ and $i+1$ both belong in both intervals associated
with $u,v$ we construct a clause which initially contains only $(\neg a_i \lor
\neg a_{i+1})$. We then add to this clause literals using the variables
$y_{i,1},\ldots,y_{i,\log n}$ which are falsified if these variables encode the
index of $u$ and literals using the variables $y_{i+1,1},\ldots,y_{i+1,\log n}$
which are falsified if these variables encode the index of $v$.

\item For every $i\in \{1,\ldots,2n-1\}$ and for any vertex $u\not\in M$ such
that $i$ and $i+1$ both belong in the interval associated with $u$ we construct
a clause which initially contains $(\neg a_i \lor a_{i+1})$. We then add to
this clause literals using the variables $y_{i,1},\ldots,y_{i,\log n}$ which
are falsified if these variables encode the index of $u$.

\item For every $i\in\{1,\ldots,2n\}$ which is the starting endpoint of the
interval associated with a vertex $u\in V\setminus M$ we do the following:
Construct (but do not yet add to the formula) a base clause which initially
contains $\neg a_i$ and add to this clause literals using
$y_{i,1},\ldots,y_{i,\log n}$ which are falsified if and only if these
variables encode the index of $u$.  Now, for each assignment to the variables
$z_{i,1},\ldots, z_{i,\log n}$ and $z_{i-1,1},\ldots,z_{i-1,\log n}$ such that
the difference of the binary numbers encoded by these two sets of variables is
not equal to the weight of $u$ we construct a copy of the base clause and add
literals falsified by these assignments. (For $i=1$ we assume that
$z_{i-1,1},\ldots,z_{i-1,\log n}$ are all set to False, that is, they encode
the number $0$.) In this way we add at most $n^2$ clauses for each $i,u$.

\item For every $i\in \{1,\ldots,2n-1\}$ and for each $j\in \{1,\ldots,\log
n\}$ we add clauses ensuring that $z_{i,j}=z_{i+1,j}$. If $i$ is a starting
endpoint of an interval we add to all these clauses the literal $a_i$.

\item For every $i\in \{1,\ldots,m-1\}$ and for each $j\in \{1,\ldots,\log n\}$
add a clause ensuring that $z_{2n+i,j}=z_{2n+i+1,j}$. Add the literal $x_i$ to
each such clause.

\item For every $i\in \{1,\ldots,m-1\}$ and for each assignment to the
variables $z_{2n+i,1},\ldots, z_{2n+i,\log n}$ and
$z_{2n+i-1,1},\ldots,z_{2n+i-1,\log n}$ such that the difference of the binary
numbers encoded by these two sets of variables is not equal to the weight of
the $i$-th vertex of $M$, add a clause that contains $\neg x_i$ and is
falsified exactly by this combination of assignments to the $z$ variables.

\item Finally, for each integer strictly smaller than the target weight $T$ add
a clause that is falsified if $z_{2n+m,1},\ldots,z_{2n+m,\log n}$ encodes this
integer.

\end{enumerate}

This completes the construction, which can clearly be carried out in polynomial
time. Furthermore, it is not hard to produce a path decomposition of $\phi-M'$
by constructing $2n+m-1$ bags and including in the $i$-th bag for all
$j\in\{1,\ldots,\log n\}$ the variables
$a_i,a_{i+1},y_{i,j},y_{i+1,j},z_{i,j},z_{i+1,j}$ (if they exist). This
decomposition has width $O(\log n)$.

What remains then is to argue for correctness. If $G$ has an independent set
$S$ of the desired weight we give assignments to the variables corresponding to
their informal meanings: set $x_v$ to True if and only if $v\in M\cap S$; set
$a_i$ to True if and only if $S$ contains a vertex $u\not\in M$ whose interval
representation contains $i$ and in this case set $y_{i,1},\ldots,y_{1,\log n}$
to encode the index of $u$; set $z_{i,1},\ldots,z_{i,\log n}$ to encode the
weight of $S$ restricted to the vertices of $V\setminus M$ whose intervals
intersect $\{1,\ldots,i\}$ if $i\le 2n$ or to the weight of $S$ restricted to
$V\setminus M$ and the first $i-2n$ vertices of $M$ if $i>2n$. It is not hard
to see that this assignment satisfies all the clauses.

For the converse direction, we extract a set $S\subseteq V$ of vertices from a
satisfying assignment as follows: for the vertices  $v\in M$ we place $v$ in
$S$ if and only if the assignment sets $x_v$ to True; for $v\not\in M$ we place
$v$ in $S$ if and only if there exists $i\in\{1,\ldots,2n\}$ which belongs in
the interval of $v$ such that $a_i$ is True and $y_{i,1},\ldots,y_{i,\log n}$
encodes the index of $v$. We claim that the set $S$ we construct in this way is
independent. Indeed, it is easy to see it cannot contain the two endpoints of
an edge contained in $M$, because of the clauses of step 5. For each $v\in
S\setminus M$ there is an $i$ such that $a_i$ is True and
$y_{i,1},\ldots,y_{i,\log n}$ encodes the index of $v$. For the sake of
contradiction (to $S$ being independent) select a $v\in S\setminus M$ that
minimizes $i$ and has a neighbor in $S$. Because of the clauses of step 6 we
observe that the interval of $v$ must contain $i$. If the neighbor of $v$ in
$S$ were also in $M$, this would falsify a clause added in step 7, so this
cannot happen.  Because of the clauses of step 9 we observe that if $i+1$ also
belongs to the interval of $v$, then $a_{i+1}$ is set to True. In this case,
because of the clauses of steps 8 and 6, $y_{i+1,1},\ldots,y_{i+1,\log n}$ must
also encode the index of $v$.  Hence, for all $i'>i$ which belong in the
interval of $v$, the assignment encodes the index of $v$. We conclude that $S$
must be independent.

Given that $S$ is independent, it is not hard to see that the $z_{i,j}$
variables must be set in a way that calculates its total weight, and because of
the clauses of the last step we know that this weight must be at least $T$ for
the assignment to be satisfying.  \end{proof}

Let us now proceed to the more challenging direction, where we need to reduce
\textsc{SAT} parameterized by a modulator to logarithmic pathwidth to
\textsc{Weighted Independent Set} parameterized by a modulator to an interval
graph. To represent the modulator we will again rely on the gadget of
\cref{prop:modulator}.

The main part of the reduction will encode a path decomposition of the primal
graph of a CNF formula into an interval graph whose size is exponential in the
width of the decomposition. To simplify presentation, we will require the given
decomposition to obey a couple of simple (and easy to attain) structural
properties.

\begin{definition} A path decomposition is called uniformized if all bags have
the same size and the symmetric difference between any two consecutive
non-identical bags consists of exactly two vertices. \end{definition}

\begin{proposition}\label{prop:uniformized} There exists a polynomial time
algorithm which, given a path decomposition of a graph, produces a uniformized
decomposition of the same width. \end{proposition}

\begin{proof} Repeat the following exhaustively: if there exist two bags of
different sizes, then find two consecutive bags $B,B'$ of distinct sizes (which
must exist) and add an element of the larger bag to the smaller one. In the end
of this process all bags have the same size. Then, if two consecutive bags
$B,B'$ have symmetric difference consisting of more than two elements, insert
between them a sequence of bags that exchanges the elements of $B\setminus B'$
with elements of $B'\setminus B$ one by one. In the end of this process the
decomposition is uniformized.  \end{proof}

\begin{lemma}\label{lem:interval-gadget}

There is an algorithm which takes as input an $n$-variable CNF formula $\phi$
and a width $p$ uniformized path decomposition of its primal graph. The
algorithm runs in time $2^{O(p)}n^{O(1)}$. Its output consists of: (i) an
instace of \textsc{Weighted Independent Set} on an interval graph $G=(V,E)$
with target value $T$ (ii) an annotation function which assigns to each vertex
of $G$ a partial truth assignment to some variables of $\phi$.  The output
satisfies the following:

\begin{enumerate}

\item For all truth assignments to the variables of $\phi$ there exists an
independent set in $G$ that has weight $T$ and consists exactly of the vertices
whose annotation is consistent with this assignment.

\item All independent sets of $G$ which contain vertices whose annotations are
inconsistent on some variable have weight strictly less than $T$.

\item All independent sets of $G$ have weight at most $T$.

\item There exists an injective mapping which, given as input a bag of the
decomposition and a truth assignment to its variables, returns a vertex of $G$
annotated with this assignment.

\item All vertices are given weights in $\{1,2,3\}$.

\end{enumerate}

\end{lemma}

\begin{proof} 

Suppose the bags of the given decomposition are numbered $B_1,\ldots,B_t$.
Suppose also, in order to ease notation, that each bag has exactly $p$ (rather
than $p+1$) variables. 

We process the bags one by one in this order and will produce groups of
vertices (and their associated intervals) called ``bundles''.  Each bundle
produced for $B_i$ will contain $2^p$ vertices, such that for each assignment
to the variables of $B_i$ exactly one of the vertices of the bundle will be
annotated with this assignment (this will allow us to satisfy requirement 4).
The intervals of a bundle will all share a common point (and hence form a
clique).  We will ensure that all intervals of each bundle have distinct
endpoints and define the rank of a vertex $v$ in a bundle as the number of
vertices that have strictly smaller right endpoint than $v$ in this bundle (so
the rank is an integer in $\{0,\ldots,2^p-1\}$).

The first bundle we construct is a special Start bundle, corresponding to the
variables of $B_1$: for each $i\in\{0,\ldots,2^p-1\}$ we construct a vertex
represented by the interval $[i,2^p-1+i]$ and annotated by the binary
representation of $i$ (seen as a truth assignment to the variables of $B_i$).
These vertices have weight $2$. Furthermore, for each $i\in\{0,\ldots,2^p-1\}$,
we construct $i-1$ singleton intervals $[0,0], [1,1],\ldots, [i-1,i-1]$, each
annotated with the assignment $i$ and each with weight $1$. We refer the reader
to \cref{fig:interval-gadget}.

We have now constructed our first bundle. Before we proceed to the construction
of the remaining bundles, let us explain the invariant that this construction
intends to maintain. Recall that in an independent set of intervals there is a
natural ordering of the vertices. We want to ensure that for each vertex $v$ of
the $j$-th bundle, the maximum weight achievable by any independent set whose
last vertex is $v$ should be exactly $2j+\textrm{rank}(v)$. Furthermore, this
weight should only be achievable if and only if we select vertices whose
annotation is consistent with a global truth assignment. This invariant holds
for the Start bundle ($j=1$) as for each vertex $v$ the best independent set
that has $v$ as the last vertex has weight $2+\textrm{rank}(v)$.

We now observe that, having constructed $j$ bundles, it is easy to add a Do
Nothing bundle which maintains the invariant: for each $i\in\{0,\ldots,2^p-1\}$
construct a weight-$2$ interval that has rank $i$ in the new bundle, is
annotated with the same truth assignment as the interval of rank $i$ in the
previous bundle, and intersects the interval of rank $i+1$ of the previous
bundle, but not the interval of rank $i$. Clearly, the only way to maintain the
invariant here is to select an interval with the same annotation as in the
previous bundle, so this simple construction allows us to propagate an
assignment (we again refer to \cref{fig:interval-gadget}).

The Do Nothing bundle is actually not necessary for our construction, but it
allows us to more easily explain the rest. Having now constructed a bundle for
the first bag, we want to proceed to the next bag of the decomposition, which
has replaced one of the variables of the current bag (say $x$) with a new one.
We start with an easy case. Suppose that the intervals of the current bundle
are annotated such that for all $i\in\{0,\ldots,2^{p-1}-1\}$ we have that the
intervals with ranks $2i$ and $2i+1$ are annotated with assignments that agree
on all variables except $x$. Intuitively, this roughly corresponds to a
situtation where, if we read the assignments as binary numbers and order them
by rank, $x$ is the least significant bit. In this case we construct a New
Variable bundle, meant to allow us to replace $x$ with a new variable.  This is
done by starting with a Do Nothing bundle, and then for each
$i\in\{0,\ldots,2^{p-1}-1\}$ cutting the interval of rank $2i$ in the new
bundle into two non-intersecting intervals of weight $1$. The starting point of
the second of the two intervals is the same as that of the interval of rank
$2i+1$ in the new bundle (again see \cref{fig:interval-gadget}). Cutting up the
interval offers two new possibilities: if we had previously selected the
interval of rank $2i+1$, we can now select the interval of rank $2i$, but which
has weight $1$ (so the invariant is satisfied); conversely, if we had selected
the interval of rank $2i$, we can select the left part of the cut interval and
also the new interval of rank $2i+1$ (so the rank and the weight are increased
by $1$ and the invariant is satisfied). Observe, however, that any selection
that produces annotations that disagree on a variable still leads to a
sub-optimal solution.

What remains to explain is how to handle the general case, where the
assignments are not conveniently ranked so that the variable we are forgetting
is the only disgreement between the intervals of rank $2i$ and $2i+1$, for all
$i$. For this we introduce a Flip bundle, which allows us to switch the ranks
of two assignments which are consecutive in the ranking. If we manage to do
this repeatedly while maintaining the invariant, we will be able to produce any
permutation of the assignments within at most $(2^p)^2$ steps, hence also an
assignment satisfying the conditions needed for a New Variable bundle.

A Flip bundle is then constructed as follows: start with a Do Nothing bundle,
but make the interval of rank $i$ end a little later and the interval of rank
$i+1$ end a little sooner, so that the ranks of all other intervals remain
unchanged but the ranks of these two are exchanged. Assign weight $1$ to the
interval that was shortened and $3$ to the interval that was lengthened. It is
not hard to see that the only way to maintain an optimal solution that
satisfies the invariant is to make a selection that is consistent in the new
bundle.  In particular, a solution which previously had the interval of rank
$i+1$ still cannot take the interval which (now) has rank $i+1$ in the new
bundle, so it must take the interval that has rank $i$ and weight $1$ (so both
weight and rank are decreased by $1$); conversely a solution that had the
interval of rank $i$ must select the interval that now has rank $i+1$ in the
new bundle (and weight $3$) as the interval that now has rank $i$ has weight
$1$.

We are now ready to describe the full construction. After the Start bundle, for
each new bag we introduce a sequence of Flip bundles, ensuring that we have a
ranking of the assignments so that the variable to be forgotten is the only
point of disagreement between assignments of rank $2i$ and $2i+1$, for all
$i\in\{0,\ldots,2^{p-1}-1\}$. We then add a New Variable bundle, and continue
in the same way. Finally, for the last bag we construct a bundle that is the
reverse of the Start bundle, ensuring that a solution with rank $i$ can collect
additional weight of $2^p-1-i$. We have constructed at most $N\le 4^pt$ bundles
and we set $T=2N+2^p-1$. We can now verify by induction, using the invariant we
mentioned, that all the promised properties hold.  \end{proof}

\begin{figure}

\includegraphics[width=0.9\textwidth]{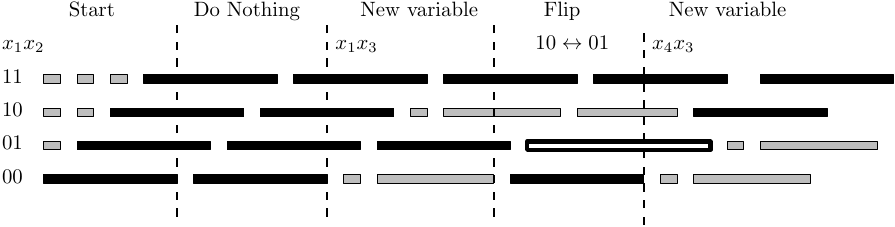} \caption{Example of the
construction of \cref{lem:interval-gadget}. Gray intervals have weight $1$,
solid black intervals weight $2$, and empty intervals weight $3$. The intervals
are annotated with the truth assignment of their line, for the active variables
of their bundle ($x_1,x_2$ for the first two bundles, $x_1,x_3$ for the next
two, and $x_4,x_3$ for the last one). The optimal solution is supposed to
collect weight $2$ per bundle, plus an extra amount equal to the rank of the
selected interval in the last bundle. The first new variable gadget allows us
to switch between $00$ and $01$ (similarly $10$ and $11$) without cost, as we
forget variable $x_2$. The flip gadget switches the ranks of $01$ and $10$,
allowing us to use a new variable gadget to forget
$x_1$.}\label{fig:interval-gadget}

\end{figure}

Armed with the previous lemma, we are now ready to prove the second part of
\cref{thm:interval}:

\begin{lemma} The first statement of \cref{thm:interval} implies the second.
\end{lemma}

\begin{proof}

Suppose there exists an algorithm which takes as input a vertex-weighted graph
$G=(V,E)$ with unary weights and a modulator $M$ such that $G-M$ is an interval
graph, and computes a maximum weight independent set in time
$2^{(1-\eps)|M|}|V|^{O(1)}$.  We will use this algorithm to refute the
\logpwmseth.

We are given as input an $n$-variable formula $\phi$, which by \cref{obs:arity}
may be assumed to be in 3-CNF (hence has total size $O(n^3)$). We are also
given a modulator $M'$ and a path decomposition of the primal graph of
$\phi-M'$ of width $O(\log n)$. The goal is to decide satisfiability of $\phi$
in time $(2-\eps')^{|M'|}n^{O(1)}$. In the remainder we will assume without
loss of generality that all clauses of $\phi$ involve a variable outside $M'$;
if this is not the case we can branch on the satisfying assignments of the
corresponding clause (which are at most $7$) and if we can do this repeatedly
we obtain an algorithm faster than brute-force. So, the interesting case is
when no such clauses exist.

We will construct a graph $G=(V,E)$, a target weight $T$, and a modulator $M$
satisfying the following:

\begin{enumerate}

\item $G$ has an independent set of weight $T$ if and only if $\phi$ is
satisfiable.

\item $|M|\le (1+\eps)|M'| + O(1)$.

\item $G$ and $M$ can be constructed in time $n^{O(1)}$.

\end{enumerate}

If we achieve the above, we obtain the lemma, as executing the reduction and
then the supposed algorithm for computing a maximum weight independent set
would run in time at most $2^{(1-\eps^2)|M'|}n^{O(1)}$.

Our construction consists of two parts: a part that encodes the modulator $M'$
and a part that represents the remaining formula using
\cref{lem:interval-gadget}. 

For the first part, we use the gadget of \cref{prop:modulator} and this encodes
the formula modulator $M'$ into a graph modulator $M$ of the desired size.  Let
$T_1$ be the maximum independent set size in the graph constructed so far, as
given by \cref{prop:modulator}.

For the second part of the construction we invoke \cref{prop:uniformized}, so
we can assume the given path decomposition of $\phi-M'$ is uniformized, and
then invoke \cref{lem:interval-gadget} to obtain an interval graph of size
polynomial in $n$ and a target weight $T_2$. Recall also that all vertices of
the interval graph are annotated with assignments to some variables of
$\phi-M'$ and that we have an injective mapping which, given a bag and an
assignment to its variables returns a vertex annotated with this assignment.
What remains now is to encode the clauses of $\phi$. Let $m$ be the number of
clauses of $\phi$. Before we proceed, we multiply all the weights of the
interval graph by $m+1$, making the target weight $T_2(m+1)$.

Assume without loss of generality that for each clause of $\phi$ there exists a
distinct bag of the given decomposition that contains all the variables of the
clause (except those of $M'$). Such a bag exists, since these variables form a
clique; we can make the mapping from clauses to bags injective by simply
repeating bags in the decomposition as needed. Consider now a clause $c$ whose
variables appear in a bag $B$. We have a mapping that gives us a vertex of the
graph for each assignment to the variables of $B$. For each such assignment
that satisfies $c$, we find the corresponding vertex and increase its weight by
$1$. For each assignment which does not satisfy $c$, let $v$ be the interval
returned by the mapping for this assignment. We will replace $v$ with several
copies of the same interval, all of which will have weight equal to the weight
of $v$ plus $1$. These intervals will be connected to $M$ in a way that ensures
that they can be selected in the independent set only if the assignment we have
encoded for $M'$ satisfies $c$. In particular, $c$ contains at most $2$
variables from $M'$, which belong in two (perhaps distinct) groups of size
$\rho$. Consider all assignments to the $2\rho$ variables of these groups and
recall that we have mapped each such assignment to a set of
$\frac{\gamma-1}{2}$ vertices of each group. For each combination of
assignments that satisfies $c$ we construct a copy of $v$ adjacent to the
$\frac{\gamma+1}{2}$ vertices of each group which are not in the independent
set when we have selected these assignments.

We now assign to all vertices of the first part weight $T_2(m+1)+m+1$ and set
the total target weight to $T=T_1(T_2(m+1)+m+1)+T_2(m+1)+m$. What remains is to
argue for correctness. For the easy direction: if $\phi$ is satisfiable, then
we select vertices in each group of $M$ according to the mapping we have
decided from truth assignments to sets of size $\frac{\gamma-1}{2}$ and
complete the selection as guaranteed by \cref{prop:modulator}. This gives us a
weight of $T_1(T_2(m+1)+m+1)$.  We then move on to the second part and select
all intervals annotated with an assignment consistent with our global truth
assignment. According to \cref{lem:interval-gadget} this allows us to obtain a
weight of $T_2(m+1)$.  Observe, however, that in some cases we have constructed
several copies of an existing interval $v$. In case these copies are annotated
in a way consistent with the truth assignment, we know that the assignment
satisfies the clause represented by these copies using variables of the
modulator. Therefore, there must exist a copy that we can add to the solution.
To see that we attain the desired total weight observe that for each clause of
$\phi$ we have increased the weight of a selected interval by $1$.

For the converse direction, we have made vertices of the first part
sufficiently expensive to ensure that $T_1$ of them must be selected and
therefore that the selection encodes a truth assignment to the modulator. In
particular, the total weight we can obtain in the second part is at most
$T_2(m+1)+m$, because by \cref{lem:interval-gadget} all independent sets in the
interval graph have initial weight at most $T_2$ (and we have multiplied these
weights by $m+1$) and we have added $1$ more for each clause to each potential
solution. Therefore, vertices of the first part are sufficiently expensive to
ensure we have selected $T_1$ of them, and in a similar way vertices of the
second part are sufficiently expensive to ensure we have selected $T_2$ of
them, which according to \cref{lem:interval-gadget} ensures we can extract a
global consistent truth assignment. Observe now that for each clause we can
gain an extra weight of $1$, so to reach the target weight we must have an
assignment that satisfies all clauses.  \end{proof}



\newpage

\bibliography{shades}

\newpage

\appendix

\section{Complexity Hypotheses}\label{sec:hypo}

\subsection{List of Hypotheses}\label{sec:list}

We list here the complexity hypotheses that appear in this paper.

\begin{itemize}

\item \textsc{SAT}-related

\begin{enumerate}

\item \seth: \cite{ImpagliazzoP01} For all $\eps>0$ there exists $k$ such that
there is no algorithm solving $k$-\textsc{SAT} for $n$-variable formulas $\phi$
in time $(2-\eps)^n|\phi|^{O(1)}$, where $n$ is the number of variables.

\item \maxsatseth: The same as the previous statement but for \textsc{Max-SAT},
that is, the problem of deciding if there is an assignment satisfying at least
$t$ clauses, where $t$ is given in the input.

\end{enumerate}

\item Circuits

\begin{enumerate}[resume]

\item \circseth: For all $\eps>0$ there is no algorithm deciding if an
$n$-input circuit of size $s$ is satisfiable in time $(2-\eps)^ns^{O(1)}$.

\item \aldseth: The same as the previous statement, but restricted to circuits
of depth at most $\eps n$.

\end{enumerate}

\item Backdoors

\begin{enumerate}[resume]

\item \twosatseth: For all $\eps>0$ there is no algorithm deciding if a CNF
formula $\phi$ with a given strong \textsc{2-SAT} backdoor of size $b$ is
satisfiable in time $(2-\eps)^b|\phi|^{O(1)}$.

\item \hornseth: The same as the previous statement, but for a strong Horn
backdoor.

\end{enumerate}

\item Modulators

\begin{enumerate}[resume]

\item \tdmseth: For all $\eps>0$ there exist $c,k$ such that there is no
algorithm that takes as input a $k$-CNF formula $\phi$ and a modulator $M$ of
size $m$ such that $\phi-M$ has tree-depth at most $c$ and decides if $\phi$ is
satisfiable in time $(2-\eps)^m|\phi|^{O(1)}$.

\item \pwmseth: For all $\eps>0$ there exists $c$ such that there is no
algorithm that takes as input a CNF formula $\phi$ and a modulator $M$ of size
$m$ such that $\phi-M$ has pathwidth at most $c$ and decides if $\phi$ is
satisfiable in time $(2-\eps)^m|\phi|^{O(1)}$.

\item \twmseth: The same as the previous statement but for treewidth at most
$c$.

\item \logtdmseth: The same as the previous statement but for tree-depth at
most $c\log |\phi|$.

\item \logpwmseth: The same as the previous statement but for pathwidth at most
$c\log |\phi|$.

\item \hubmaxsatseth: For all $\eps>0$ there exist $k,\sigma,\delta$ such that
there is no algorithm which takes as input a $k$-CNF $\phi$ and a \sdh\ $M$ of
size $m$ and computes an assignment satisfying the maximum number of clauses in
time $(2-\eps)^m|\phi|^{O(1)}$.

\item \sdhseth: The same as the previous statement but for deciding if the
input formula is satisfiable.

\end{enumerate}

\item Weighted \textsc{SAT}

\begin{enumerate}[resume]

\item \wsatseth: For all $\eps>0$ there is no algorithm which takes as input a
(general) Boolean formula $\phi$ of size $s$ on $n$ variables and an integer
$k$ and decides if $\phi$ is satisfiable by an assignment that sets exactly $k$
variables to True in time $n^{(1-\eps)k}s^{O(1)}$.

\item \wpseth\ (weak version): Same as the previous statement but for Boolean
circuits.

\item \wpseth\ (strong version): For all $\eps>0, k_0>0$, there is no algorithm
which for all $k>k_0$ takes as input a Boolean circuit of size $s$ and decides
if the circuit has a weight-$k$ satisfying assignment in time $s^{(1-\eps)k}$.

\end{enumerate}

\item Others

\begin{enumerate}[resume]

\item \mtsh: The same as \maxsatseth\ but for CNFs of maximum arity $3$.

\item \ncseth: The same as \aldseth\ but for circuits of depth at most
$\log^{O(1)}n$.

\item \ppseth: For all $\eps>0$ there is no algorithm that solves \tsat\ in
time $(2-\eps)^{\pw}n^{O(1)}$ for formulas on $n$ variables of primal pathwidth
$\pw$.

\end{enumerate}

\end{itemize}

\subsection{Relations between Classes of Hypotheses}

Many of the relations between the hypotheses we consider can be inferred
directly from the definitions or from the relations between graph widths. In
particular, we have $\tw \le \pw \le \td \le \tw\log n$ for any graph on $n$
vertices. This gives $1\Rightarrow 7\Rightarrow 8\Rightarrow 9 \Rightarrow 10
\Rightarrow 11 \Rightarrow 13 \Rightarrow 12$. Furthermore, it is clear that
deciding if a general circuit is satisfiable is at least as hard as a circuit
with depth $\eps n$, which is at least as hard as a CNF formula, so
$1\Rightarrow 4\Rightarrow 3$.  We have $14\Rightarrow 15$ as formulas are a
special case of circuits. Also $2\Rightarrow 12$ as the set of all variables is
a \sdh.

It is worth noting that the arity $k$ appears explicitly in statement 7, but
not in statements 8 to 11. This is because of the following observation, which
allows us to decrease the arity once the parameter we are dealing with allows
the construction of long paths.

\begin{observation}\label{obs:arity} There is a polynomial time algorithm which
takes as input a CNF formula $\phi$ and a modulator $M$ and produces an
equisatisfiable 3-CNF formula $\psi$ by adding variables to $\phi$ such that
$M$ remains a modulator of $\psi$ and we have $\pw(\psi-M)\le \pw(\phi-M)+2$,
$\tw(\psi-M)\le \tw(\phi-M)+2$, $\td(\psi-M) \le \td(\phi-M)+O(\log |\phi|)$.
\end{observation}

\begin{proof} 

We first state an easy claim.

\begin{claim} Let $G=(V,E)$ be a graph and $\mathcal{K}\subseteq 2^V$ be a set
of cliques of $G$.  Let $G'$ be a graph obtained from $G$ by adding for each
$K\in \mathcal{K}$ a path going through at most $\ell$ new vertices and
connecting some of the vertices of the path to vertices of $K$.  Then
$\pw(G')\le \pw(G)+2$, $\tw(G')\le \tw(G)+2$, and $\td(G')\le
\td(G)+\log(\ell+1)$.  \end{claim}

\begin{claimproof} Assume we start with a tree or path decomposition of $G$ or
a rooted tree proving the tree-depth of $G$. For path and tree decompositions
we know that there is a bag $B$ that contains the clique $K$.  Take a path
decomposition of the path that we connected to $K$ and add to all bags the
contents of bag $B$ (the width has now been increased by $2$). In the case of
treewidth, attach one bag of this decomposition to $B$; in the case of
pathwidth insert this decomposition after $B$. For tree-depth, let $k\in K$ be
the vertex that is furthest from the root.  We observe that all vertices of
$K\setminus\{k\}$ must be ancestors of $k$. We insert below $k$ a tree of depth
at most $\log(\ell+1)$ containing the path we connected to $k$.  Here we are
using the fact that a path of $\ell$ vertices has tree-depth at most
$\log(\ell+1)$ (which can easily be seen by induction).  \end{claimproof}

We apply the following standard trick exhaustively: as long as $\phi$ contains
a clause of arity $r>3$, say $C=(\ell_1\lor \ell_2\lor\ldots\lor \ell_r)$,
introduce $r-3$ new variables $z_1,\ldots,z_{r-3}$ and replace $C$ with the
clauses $(\ell_1\lor \ell_2\lor z_1)\land (\neg z_1\lor \ell_3\lor z_2)\land
(\neg z_2\lor \ell_4\lor z_3)\land\ldots\land (\neg z_{r-3}\lor \ell_{r-1}\lor
\ell_r)$. Observe now that the variables of $C$ form a clique in the primal
graph, while the new variables induce a path whose neighborhood is contained in
this clique. We can therefore apply the claim, with $\mathcal{K}$ being the set
of all clauses on which we applied this transformation.  \end{proof}

We note here two further easy observations justifying the organization of
hypotheses we have made in \cref{fig:results}. 

\begin{theorem} (\maxsatseth\ $\Rightarrow$ \aldseth) If there is an $\eps>0$
and an algorithm that decides satisfiability of $n$-input circuits of size $s$
and depth at most $\eps n$ in time $(2-\eps)^ns^{O(1)}$, then there is an
$\eps'>0$ and an algorithm that given as input an $n$-variable CNF formula
$\phi$ finds the assignment satisfying the maximum number of clauses in time
$(2-\eps')^n|\phi|^{O(1)}$.  \end{theorem}

\begin{proof}

Construct a circuit with $n$ inputs and a gate for each clause appropriately
connected to the inputs through negations. This circuit can be made to have
fan-in $2$ and depth $\log n$. For each target $t$ we now want to add some
layers to the circuit so that it is satisfiable if and only if we can satisfy
at least $t$ clauses of the original instance. Running the assumed algorithm
for each $t$ allows us to compute the maximum number of satisfied clauses,
while self-reducibility allows us to find an assignment that satisfies this
many clauses.

What is missing then is a circuit with $m$ inputs and depth $O(\log m)$, where
$m$ is the number of clauses, that outputs True if and only if at least $t$ of
its inputs are set to True.  This is just a threshold gate, and it is known
that such gates can be simulated by logarithmic depth circuits (that is,
TC$^0\subseteq$NC$^1$). We have a circuit of depth $C\log m$ and we observe
that if $C\log m>\eps n$, the original instance has $m>2^{\eps n/C}$ and can
therefore be solved in time $m^{O(1/\eps)}=|\phi|^{O(1)}$. We can therefore
assume that the circuit we have constructed has depth at most $\eps n$ and
invoke the assumed algorithm.  \end{proof}

\begin{theorem} (\twosatseth\ $\Rightarrow$ \circseth) If there is an $\eps>0$
and an algorithm that decides satisfiability of $n$-input circuits of size $s$
in time $(2-\eps)^ns^{O(1)}$, then there is an $\eps'>0$ and an algorithm that
given as input a CNF formula $\phi$ and a strong \textsc{2-SAT} backdoor of
size $b$ decides if $\phi$ is satisfiable in time
$(2-\eps')^b|\phi|^{O(1)}$.\end{theorem}

\begin{proof}

The idea is to construct a circuit implementing the standard polynomial-time
algorithm for \textsc{2-SAT}. Recall that one way to decide satisfiability for
\textsc{2-SAT} is to run resolution exhaustively, where resolution is the rule
that, given two clauses $C,C'$ that contain two opposite literals $x,\neg x$,
adds to the formula a new clause containing all the other literals of $C,C'$.
It is known that a 2-CNF formula is satisfiable if and only if performing this
exhaustively does not produce the empty formula.

We are given a CNF $\phi$ and a backdoor $B$ of size $b$. Let $x_1,\ldots,x_n$
be the non-backdoor variables of $\phi$. We construct a circuit with $b$
inputs. Observe that there are at most $4n^2$ clauses we can construct with the
variables $x_1,\ldots,x_n$. For each $t\in[0,4n^2]$ and clause $C$ that can be
constructed from these variables we will have in our circuit a gate $x_{t,C}$
whose meaning is that the algorithm has inferred $C$ after $t$ steps. The gates
$x_{0,C}$ can be computed from the inputs: $x_{0,C}$ is the disjunction over
all clauses of $\phi$ that contain $C$ of the conjunction of the negation of
all their other literals. It is now easy to compute $x_{i+1,C}$ from layer
$x_{i,C}$ and the output of the circuit is $x_{4n^2,\emptyset}$. The size of
the circuit is polynomial in $\phi$.  \end{proof}

\end{document}